\documentclass[11pt]{article}
\usepackage[utf8]{inputenc}
\usepackage{fullpage}
\usepackage[dvipsnames]{xcolor}

\usepackage{amsthm}
\usepackage{complexity}
\usepackage{mathtools,amsmath,amssymb}
\usepackage{graphicx}
\usepackage{multirow}
\usepackage{tabularx}
\usepackage[colorinlistoftodos]{todonotes}
\usepackage{enumitem}
\usepackage{float}
\usepackage{physics}
\usepackage{parskip}
\usepackage{relsize}

\usepackage{microtype}

\usepackage{times}
\normalfont
\usepackage[T1]{fontenc}
\usepackage{float}
\usepackage{mathpazo}

\DeclareMathSizes{10.95}{10.95}{8}{7.5}

\usepackage{thmtools}
\usepackage{thm-restate}

\usepackage[colorlinks=true,linktoc=page]{hyperref}
\hypersetup{colorlinks,
    linkcolor=blue, 
    citecolor=ForestGreen,      
    urlcolor=red,    
}

\usepackage[breakable]{tcolorbox}
\newtcolorbox{sgd}[2][]
{
  breakable,
  colframe = gray!50,
  colback  = gray!10,
  coltitle = gray!10!black,
  before skip = 10pt,
  after skip = 10pt,
  title    = \textbf{#2.},
  #1,
}
\newtcolorbox{pca}[2][]
{
  breakable,
  colframe = red!50,
  colback  = red!10,
  coltitle = red!10!black,
  before skip = 10pt,
  after skip = 10pt,
  title    = \textbf{Example: Local convergence for streaming $k$-PCA #2},
  #1,
}

\usepackage{collectbox}



\usepackage{algorithm2e}

\newif\ifdraft
\draftfalse

\newcommand{\Snotes}[1]{\ifdraft{\color{orange}[JS: #1]}\else\fi}

\usepackage{tikz}
\newenvironment{qcircuit}{\begin{tikzpicture}[qcircuit] }{\end{tikzpicture}}
\def\circleH{.16} 
\def\targetH{.75*\circleH} 
\def\controlH{.375*\circleH} 

\tikzset{fc/.style={path picture={
	\filldraw[fill=white] (0,0) circle (\circleH);
	}}} 
\tikzset{fcb/.style={path picture={
	\filldraw[fill=black] (0,0) circle (\controlH);
	}}} 
\tikzset{plus/.style={path picture={
	\draw (0,0) circle (\targetH);
	\filldraw[fill=black] (-\targetH,0) -- (\targetH,0) (0,-\targetH) -- (0,\targetH) ;
	}}} 
\tikzstyle{gate} = [draw,fill=white,minimum size=1.5em,inner sep=2pt] 
\tikzstyle{ket} = [fill=white,minimum size=1.5em]
\tikzset{qcircuit/.style={thick, minimum size=3ex}} 

\tikzset{meter/.append style={draw, inner sep=5, rectangle, font=\vphantom{A}, minimum width=23, line width=.8,
 path picture={\draw[black] ([shift={(.1,.2)}]path picture bounding box.south west) to[bend left=50] ([shift={(-.1,.2)}]path picture bounding box.south east);\draw[black,-latex] ([shift={(0,.1)}]path picture bounding box.south) -- ([shift={(.3,-.1)}]path picture bounding box.north);}}}

\usepackage{amsthm}
\usepackage{thmtools,thm-restate}

\makeatletter
\renewcommand{\paragraph}{%
  \@startsection{paragraph}{4}%
  {\z@}{1.75ex \@plus 1ex \@minus .2ex}{-1em}%
  {\normalfont\normalsize\bfseries}%
}
\makeatother

\numberwithin{equation}{section}

\declaretheoremstyle[bodyfont=\it,qed=\qedsymbol,headpunct=.\vphantom{$p_{p_{p_p}}$},postheadspace=\newline,headformat=\NAME\  \NUMBER\,\NOTE]{noproofstyle} 

\newtheoremstyle{break}%
{}{}%
{\itshape}{}%
{\bfseries}{.\vphantom{$p_{p_{p_p}}$}}%
{\newline}
{\thmname{#1}\thmnumber{ #2}\thmnote{\ \,\textmd{(#3)}}}

\theoremstyle{break}

\declaretheorem[name=Observation,numbered=no]{observation*}

\declaretheorem[numberlike=equation]{problem}

\declaretheorem[numberlike=equation]{theorem}

\declaretheorem[name=Theorem,numbered=no]{theorem*}

\declaretheorem[numberlike=equation]{lemma}
\declaretheorem[name=Lemma,numbered=no]{lemma*}

\declaretheorem[numberlike=equation]{corollary}
\declaretheorem[name=Corollary,numbered=no]{corollary*}

\declaretheorem[numberlike=equation]{parameter}
\declaretheorem[name=Parameter,numbered=no]{parameter*}

\declaretheorem[numberlike=equation]{proposition}
\declaretheorem[name=Proposition,numbered=no]{proposition*}

\declaretheorem[name=Claim,numbered=no]{claim*}

\declaretheorem[numberlike=equation]{conjecture}
\declaretheorem[name=Conjecture,numbered=no]{conjecture*}

\declaretheorem[name=Question,numbered=no]{question*}

\declaretheoremstyle[bodyfont=\it,headpunct=.\vphantom{$p_{p_{p_p}}$},postheadspace=\newline,headformat=\NAME\  \NUMBER\,\NOTE]{defstyle} 

\declaretheorem[numberlike=equation,style=defstyle]{definition}
\declaretheorem[unnumbered,name=Definition,style=defstyle]{definition*}

\declaretheorem[unnumbered,name=Example,style=defstyle]{example*}

\declaretheorem[unnumbered,name=Notation=defstyle]{notation*}

\declaretheorem[unnumbered,name=Construction,style=defstyle]{construction*}

\declaretheoremstyle[]{rmkstyle} 

\newcommand{\handout}[5]{
   \renewcommand{\thepage}{#1-\arabic{page}}
   \noindent
   \begin{center}
   \framebox{
      \vbox{
    \hbox to 5.78in { {\bf #1}
     	 \hfill #2 }
       \vspace{4mm}
       \hbox to 5.78in { {\Large \hfill #5  \hfill} }
       \vspace{2mm}
       \hbox to 5.78in { {\it #3 \hfill #4} }
      }
   }
   \end{center}
   \vspace*{4mm}
}

\newcommand{\iprod}[1]{\left\langle{#1}\right\rangle}
\DeclareMathOperator*{\lE}{\mathbb{E}}
\renewcommand{\E}{\lE}
\newcommand{\eps}{\epsilon}
\newcommand{\val}{\mathsf{val}}
\newcommand{\QAOA}{\text{QAOA}}

\newenvironment{proof-sketch}{\noindent{\bf Sketch of Proof:}\hspace*{1em}}{\qed\bigskip}
\newenvironment{proof-idea}{\noindent{\bf Proof Idea}\hspace*{1em}}{\qed\bigskip}
\newenvironment{proof-of-lemma}[1]{\noindent{\bf Proof of Lemma #1}\hspace*{1em}}{\qed\bigskip}
\newenvironment{proof-attempt}{\noindent{\bf Proof Attempt}\hspace*{1em}}{\qed\bigskip}

\newenvironment{remark}{\noindent{\bf Remark}\hspace*{1em}}{\bigskip}

\newcommand{\ourparam}{Let $(k,\eta,\tau,d,n,p)$ be parameters satisfying~\autoref{param}.}

\newcommand{\ourparamwithgamma}{Let $(k,\eta,\tau,d,n,p)$ be parameters satisfying~\autoref{param} and $\gamma>0$ from~\autoref{thm:coup-hypergraph-qaoa-output-overlap}.}

\makeatletter
\def\fnum@figure{{\bf Figure \thefigure}}
\def\fnum@table{{\bf Table \thetable}}
\long\def\@mycaption#1[#2]#3{\addcontentsline{\csname
  ext@#1\endcsname}{#1}{\protect\numberline{\csname
  the#1\endcsname}{\ignorespaces #2}}\par
  \begingroup
    \@parboxrestore
    \small
    \@makecaption{\csname fnum@#1\endcsname}{\ignorespaces #3}\par
  \endgroup}
\def\mycaption{\refstepcounter\@captype \@dblarg{\@mycaption\@captype}}
\makeatother

\newcommand{\mathify}[1]{\ifmmode{#1}\else\mbox{$#1$}\fi}
\newcommand{\bigO}O

\newcommand{\remove}[1]{}
\newcommand{\ignore}[1]{}

\def\R{\mathbb{R}}

\def\N{\mathbb{N}}

\def\cH{{\cal H}}

\def\cX{{\cal X}}

\def\bbE{{\mathbb E}}

\newcommand{\Ex}{\mathop{\bbE}}

\newcommand{\kxors}{\ensuremath{\mathsf{Max}\text{-}k\text{-}\mathsf{XOR}}}%
\newcommand{\maxcut}{\ensuremath{\mathsf{Max}\text{-}\mathsf{Cut}}}%
\newcommand{\poisson}{\ensuremath{\mathrm{Poisson}}}%

\author{Chi-Ning Chou\thanks{School of Engineering \& Applied Sciences, Harvard University,     Cambridge, Massachusetts, USA. Supported by NSF awards CCF 1565264 and CNS 1618026.         Email: \url{chiningchou@g.harvard.edu}.}
\and 
Peter J. Love\thanks{Department of Physics \& Astronomy, Tufts University, Medford, Massachusetts, USA. Supported by DARPA ONISQ program award HR001120C0068. Email: \url{peter.love@tufts.edu}.} \and
Juspreet Singh Sandhu\thanks{School of Engineering \& Applied Sciences, Harvard University, Cambridge, Massachusetts, USA. Supported by DARPA ONISQ program award HR001120C0068. Email: \url{jus065@g.harvard.edu}.} \and
Jonathan Shi\thanks{Department of Computer Science, Bocconi University, Milan, Italy. Supported by European Research Council (ERC) award No. 834861. Email: \url{jonathan.shi@unibocconi.it}.}
}
\title{Limitations of Local Quantum Algorithms on Random \kxors~and Beyond}

\begin{document}
\maketitle
\begin{abstract}

We introduce a notion of \emph{generic local algorithm} which strictly generalizes existing frameworks of local algorithms such as \emph{factors of i.i.d.} by capturing local \emph{quantum} algorithms such as the Quantum Approximate Optimization Algorithm (QAOA).

Motivated by a question of Farhi et al.~[arXiv:1910.08187, 2019] we then show limitations of generic local algorithms including QAOA on random instances of constraint satisfaction problems (CSPs).
Specifically, we show that any generic local algorithm whose assignment to a vertex depends only on a local neighborhood with $o(n)$ other vertices (such as the QAOA at depth less than $\epsilon\log(n)$) cannot arbitrarily-well approximate boolean CSPs if the problem satisfies a geometric property from statistical physics called the coupled overlap-gap property (OGP) [Chen et al., Annals of Probability, 47(3), 2019].
We show that the random $\mathsf{MAX}$-$k$-$\mathsf{XOR}$ problem has this property when $k\geq4$ is even by extending the corresponding result for diluted $k$-spin glasses. 

Our concentration lemmas confirm a conjecture of Brandao et al.~[arXiv:1812.04170, 2018] asserting that the landscape independence of QAOA extends to logarithmic depth---in other words, for every fixed choice of QAOA angle parameters, the algorithm at logarithmic depth performs almost equally well on almost all instances.

One of these concentration lemmas is a strengthening of McDiarmid's inequality, applicable when the random variables have a highly biased distribution, and may be of independent interest. 

\end{abstract}
\thispagestyle{empty}

\newpage
\thispagestyle{empty}
\setcounter{tocdepth}{2}
{
    \hypersetup{linkcolor=blue}  
    \tableofcontents
    \thispagestyle{empty}
}

\newpage
\clearpage
\pagenumbering{arabic}

\section{Introduction}\label{sec:intro}
Recent developments~\cite{arute2020hartree,gong2021quantum,ebadi2021quantum} of noisy intermediate-scale quantum (NISQ) devices~\cite{preskill2018quantum} have brought us to the door of near-term quantum computation. As experimentalists can now build programmable quantum simulators up to 256 qubits~\cite{ebadi2021quantum}, this motivates an important theoretical question: what computational advantage can such a NISQ device provide?

One of the constraints of NISQ devices is the inability to create high-fidelity global entanglement. This motivates the study of the power of quantum algorithms that are \emph{local}.~A leading candidate in this regime of quantum algorithms is the \emph{Quantum Approximate Optimization Algorithm (QAOA)}~\cite{farhi2014quantum} at shallow depths. While there have been some recent results~\cite{hastings2019classical, barak2021classical, marwaha2021local} that formally examine the QAOA algorithm at depth $p = 1$ or $2$, very few results exist for super-constant depth QAOA~\cite{farhi2020quantum, farhi2020quantumw}.

Given the imminent quest of demonstrating quantum computational advantage, it is important to clarify
for what optimization problems can near-term quantum algorithms (such as \emph{local} quantum algorithms) reliably be expected to demonstrate computational advantage.


We show that local \emph{quantum} algorithms, a large natural class of NISQ algorithms, are obstructed by a geometric property of the solution space known as the \emph{coupled Overlap-Gap Property}~\cite{chen2019suboptimality}.
We conjecture that this property is satisfied by most CSPs (\autoref{prob:csp-ogps}).
Specific problems known to have this property include the diluted $k$-spin glass Hamiltonian (equivalent to a max-cut problem on random $k$-hypergraphs)~\cite{chen2019suboptimality}, independent set on random graphs~\cite{farhi2020quantum}, planted clique~\cite{gamarnik2019landscape}, and many other problems that so far seem to elude efficient algorithms and be algorithmically hard~\cite{gamarnik2021overlap}.
In this manuscript, we also demonstrate that the random \kxors{} problem has this property (see~\autoref{sec:signed-interpolation}). 


Critical to our approach is a new definition of local algorithms we term \emph{generic local algorithms} (See~\autoref{sec:p-local-algorithms}).
Previous work relating statistical-physics-derived OGPs to local algorithms leveraged the \emph{factors of i.i.d.}~framework for local algorithms, which fails to contain local quantum algorithms, as we demonstrate in \autoref{prop:p-local-vs-factors}.
Our definition of \emph{generic local algorithms} subsumes local quantum and classical algorithms (see~\autoref{prop:p-local-vs-factors} and~\autoref{prop:qaoa-2p-local}) but still satisfies strong concentration properties (see~\autoref{lem:concentration-local-functions} and~\autoref{lem:concentration-random-coupled-hypergraphs}), allowing obstruction techniques for local classical algorithms \cite{chen2019suboptimality} to apply to the quantum case.
Two of our core technical contributions involve showing that the random $\mathsf{MAX}$-$k$-$\mathsf{XOR}$ problem has a coupled Overlap Gap Property (see~\autoref{sec:signed-interpolation}) by extending the techniques of Chen et al~\cite{chen2019suboptimality} and deriving a strengthened version of McDiarmid's inequality for highly-biased random variables using a martingale argument (see~\autoref{lem:stronger-mcdiarmids}).

The rest of the paper is organized as follows: In~\autoref{sec:intro diluted spin glass} we give a brief introduction to the motivating spin glass literature, defining the notion of a diluted $k$-spin glass; in~\autoref{sec:related} we introduce the relevant prior work; in~\autoref{subsec:results-informal} we state our main theorems (informally); in~\autoref{sec:tech overview} we briefly explain the architecture of our proof and compare our techniques with those of Chen et al.~\cite{chen2019suboptimality} and Farhi et al.~\cite{farhi2020quantum}; in~\autoref{sec:prelim} we introduce the necessary mathematical preliminaries and notation, including a rigorous definition of local classical algorithms, the QAOA algorithm and Overlap-Gap Properties; in~\autoref{sec:p-local-algorithms} we introduce the notion of a \emph{generic local algorithm}, how to sample from correlated runs of them, and finally show separation of different families of local algorithms; in~\autoref{sec:obstruction-thm-proof} we state our main theorems formally and give proof sketches; in~\autoref{sec:concentration-results} we state and prove multiple concentration of measure statements about \emph{generic local algorithms}; in~\autoref{sec:proofs-main-thms} we make the proof showing obstructions against~\emph{generic local algorithms} using the same interpolation procedure of Chen et al.~\cite{chen2019suboptimality}; in~\autoref{sec:stronger-mcdiarmids} we state and prove a strengthened version of \emph{McDiarmid's} inequality; in~\autoref{sec:signed-interpolation} we state and prove an OGP for random~\kxors{}; in~\autoref{sec:conjectures} we conclude by summarizing our results and mentioning many natural open problems closely related to and/or motivated by our work.

\subsection{Diluted \texorpdfstring{$k$}{k}-spin glasses, maximum cut of sparse hypergraphs, and \texorpdfstring{\kxors{}}{MAX-k-XOR}}\label{sec:intro diluted spin glass}

Spin glass theory is a central theoretical framework in statistical physics. The Sherrington-Kirkpatrick model (SK model) \cite{sherrington1975solvable} is one of the most well studied mathematical models in the theory and consists of two variables: spins $\{\sigma_i\}_{i\in[n]}$ and interactions $\{J_{i,j}\}_{i,j\in[n]}$. A spin $\sigma_i$ takes values in $\{\pm 1\}$ and the interaction $J_{i,j}$ between two spins $\sigma_i,\sigma_j$ is a real-valued variable that captures whether the physical system prefers the two spins to be the same ($J_{i,j}>0$) or different ($J_{i,j}<0$). The goal is to understand what spin configurations $\sigma\in\{-1,1\}^n$ maximize the following quantity  (a.k.a. Hamiltonian):
\[
H(\sigma) = \sum_{i,j} J_{i,j}\sigma_i\sigma_j \, .
\]
The setting is easily generalized to higher order interactions, i.e., $J_{i_1,\dots,i_k}$ acting on $k$ spins, and this is known as the $k$-spin model. See Panchenko~\cite{panchenko2014introduction} for a comprehensive survey.

There is a natural correspondence between spin glass theory and combinatorial optimization problems. In a combinatorial optimization problem (e.g., \maxcut), a variable corresponds to a spin and a constraint corresponds to an interaction. Through this correspondence, the maximization of the above Hamiltonian $H(\sigma)$ serves as a proxy for maximizing the number of satisfied constraints in the combinatorial optimization problem.

A spin glass model additionally specifies a particular distribution on the interactions $\{J_{i,j}\}$ for all $i,j\in[n]$. The quantity of interest is the asymptotic maximum value
\[
  H^* := \lim_{n \to \infty}\frac{1}{n}\max_{\sigma} H(\sigma)\, ,
\]
(a.k.a.~the ground state energy density). 
Also of interest are spin configurations $\sigma$ with $H(\sigma) \approx H^*$.
There are many well-studied spin glass models in physics and various mathematical insights about these have been discovered over the years \cite{concetti2018full, dembo2017extremal, sen2018optimization, panchenko2004bounds}. For example, for the SK model \cite{sherrington1975solvable} Parisi~\cite{parisi1980sequence} proposed the infamous \emph{Parisi Variational Principle} to capture the exact value of $H^*$. This was later rigorously proved by Talagrand~\cite{talagrand2006parisi} and again by Panchenko~\cite{panchenko2014parisi} in greater generality. 
These successes give hope to design local algorithms that simulate the physical system and output a final configuration as an approximation to the corresponding combinatorial optimization problem.

While traditional spin glass models consider the underlying non-trivial interactions as either lying on a certain physically-realistic graph (e.g., the non-zero $J_{i,j}$ form a 2D-grid) or being a \emph{mean field approximation} (for example, where every $J_{i,j}$ is non-trivial), the applications in combinatorial optimization often require the underlying constraint graphs to be sparse and arbitrary.
We use two methods of bridging the gap between the two settings:
\begin{itemize}
    \item By studying the \emph{diluted $k$-spin glass model} where one first samples a sparse hypergraph and then assigns non-trivial interactions on top of its hyperedges. Intuitively, approximating the $H^*$ of the diluted $k$-spin glass corresponds to approximating the maximum cut over random sparse hypergraphs. This correspondence is made more precise in~\autoref{sec:prelim diluted k spin}.
    \item Using the techniques of \emph{Gaussian interpolation} \cite{guerra2002thermodynamic} and \emph{Poisson interpolation} \cite{franz2003replica,chen2019suboptimality} from statistical physics to relate the behavior of random dense spin glass models to random sparse CSPs.
    More specifically, we relate the random \kxors{} problem to mean-field $p$-spin glasses (\autoref{sec:signed-interpolation}).
\end{itemize}

\begin{table}[H]
\centering
\begin{tabular}{|c|c|}
\hline
\textbf{Spin glass models}                   & \textbf{Combinatorial optimization problems}              \\ \hline
Spins $\sigma\in\{-1,1\}^n$            & An assignment to boolean variables                 \\ \hline
Interactions $\{J_{i_1,\dots,i_k}\}_{i_1,\dots,i_k\in[n]}$ & Constraints (i.e., hyperedges)                     \\ \hline
Hamiltonian $H(\sigma)$                & Value of an assignment (i.e., $\val_\Psi(\sigma)$) \\ \hline
Ground state energy $H^*$              & Optimal value (i.e., $\val_\Psi$)                  \\ \hline
Mean field model (e.g., SK model)      & The underlying hypergraph being complete                \\ \hline
Diluted spin glass model               & The underlying hypergraph being sparse                  \\ \hline
\end{tabular}
\caption{A dictionary between spin glass models and combinatorial optimization problems.}
\label{tab:dictionary}
\end{table}

\subsection{Related work}\label{sec:related}

\paragraph{Constraint-satisfaction problems \& hardness for classical algorithms}
CSPs (described formally in \autoref{def:k-d-csp}) are a natural class of combinatorial optimization problems that have been studied extensively in theoretical computer science~\cite{brailsford1999constraint, kumar1992algorithms}. Many $\mathsf{NP\text{-}Complete}$ problems such as $\mathsf{k\text{-}SAT}$, $\mathsf{k\text{-}NAE\text{-}SAT}$, $\mathsf{MAX\text{-}CUT}$ and $\mathsf{k\text{-}XOR}$, can be framed as CSPs. Consequently, unless $\mathcal{P} = \mathcal{NP}$, finding optimal solutions to these problems is infeasible. A natural question then is to understand how well can approximate answers to instances of these problems be constructed by efficient algorithms. Under the now widely believed \emph{Unique Games Conjecture}~\cite{khot2005unique}, upper bounds on the approximability of CSPs are known~\cite{khot2007optimal, raghavendra2008optimal}. These bounds, however, are only worst-case and do not necessarily explicitly demonstrate a family of instances of a CSP that are hard to approximate. Additionally, they remain conditional on a positive resolution to the Unique Games Conjecture, which is still a difficult open problem in the field. In the average-case regime, the goal is to ask how well a \emph{typical} instance of a CSP can be approximated, where the instance is chosen from a ``natural" distribution over the set of instances. Perhaps surprisingly, great insight has been drawn about the algorithmic hardness (or lack thereof) about random instances of many CSPs based on work originating in the Statistical Physics community, particularly in Spin Glass Theory~\cite{mezard2001bethe, franz2003replica, parisi2014diluted}. This was so because the problem of finding spin configurations of particles in many spin glass models that put a system in the ground state could naturally be interpreted as a CSP. Various iterative algorithms were proposed to study the problem of explicitly finding near-ground states of \emph{typical} instances of various spin glass models~\cite{yedidia2003understanding, braunstein2005survey}. It was observed that these algorithms either consistently got better with the number of iterations, or hit a threshold which they could not exceed. To understand this, the work of Achlioptas et al~\cite{achlioptas2006solution} studied the solution geometry of the $\mathsf{k\text{-}SAT}$ problem and found that most good solutions were in well separated clusters. Additionally, most variables in a good solution could only take a single value (i.e., they were ``frozen"). This observation was stated as an intuitive reason for the failure of local algorithms on random instances of $\mathsf{k\text{-}SAT}$. Gamarnik et al~\cite{gamarnik2014limits} made this more formal and precise by showing that \emph{no} classical local algorithm (described formally as \emph{factors of i.i.d.}, see \autoref{sec:factors-of-iid}) could approximate the $\mathsf{MAX\text{-}IND\text{-}SET}$ problem arbitrarily well on sparse random graphs. Critical to their argument was the fact that \emph{all} (not \emph{most}) nearly-optimal solutions to the problem satisfied the \emph{Overlap Gap Property} - they were in well separated clusters. In various works that followed up, many problems have been shown to have near-optimal solutions conform to this solution geometry and algorithmic hardness for various families of classical algorithms has been established~\cite{chen2019suboptimality, gamarnik2020low, gamarnik2021overlap}.

\paragraph{Results about QAOA.}
In their seminal work, Farhi et al.~\cite{farhi2014quantum} introduced $\QAOA$ as a possible way to approximately solve certain hard combinatorial optimization problems. To illustrate the capabilities of $\QAOA$, its performance at $p=1$ was shown to achieve an approximation ratio of at least $.6924$ for the~\maxcut{}~ problem on triangle-free 3-regular graphs~\cite{farhi2014quantum}. In a follow up work, Wurtz et al.~\cite{wurtz2020bounds} improved this to $.7559$ for $3$-regular graphs with $p = 2$ and made the empirical observation that the bound was tight for graphs with no cycles of length $< 7$. 
Shortly after $QAOA_p$ was proposed, however, a local classical algorithm was designed that outperformed it on these graphs at depth 1~\cite{hastings2019classical}. Consequently, because of a flurry of follow up results, QAOA has been shown to be outmatched by local classical algorithms up to depth 2~\cite{marwaha2021local,barak2021classical} for the $\mathsf{MAX\text{-}CUT}$ problem on $d$-regular graphs with large girth. In fact, under the widely-believed conjecture in the Spin-Glass Theory community that the SK model does not satisfy the \emph{Overlap Gap Property}~\cite{auffinger2020sk}, an AMP algorithm was recently proposed that outputs arbitrarily good cuts for large (but constant) degree random regular graphs~\cite{alaoui2021local}. However, this result~\cite{alaoui2021local} does not \emph{completely} rule out a possibility for quantum advantage (see~\autoref{subsec:max-cut-d-reg}). To analyze the performance of $QAOA_p$ on a problem that possesses an OGP, Farhi et al.~\cite{farhi2020quantum} established that $QAOA_p$ with depth $p \leq \epsilon\log(n)$ could not output independent sets of size better than $.854$ times the optimal for sparse random graphs. This work suggested that the OGP may broadly prove to be an obstacle for $QAOA_p$ while it is local as much as it does for various classical algorithms. However, $\mathsf{MAX\text{-}IND\text{-}SET}$ is not a (maximum) CSP and, additionally, the prior work~\cite{farhi2020quantum} does not give an analysis that generalizes to CSPs. Our work establishes this generalization and also immediately positively resolves the ``landscape independence" conjecture of $QAOA_{\epsilon\log(n)}$ proposed by Brandao et al.~\cite{brandao2018fixed}. This immediately suggests that quantum advantage is unlikely to be found up to this depth for CSPs with an OGP, and we conjecture that \emph{almost all} CSPs will have an OGP (see \autoref{subsec:which-csps-have-ogp}).

It was shown by Bravyi et al.~\cite{bravyi2020obstacles} that $\QAOA$ would not output cuts better than $\frac{5}{6} + O(\frac{1}{\sqrt{d}})$ times the optimal value for some infinite family of $d$-regular graphs (which happen to be bipartite). This was achieved as a corollary to their proof for a $\log(n)$-depth version of the NLTS conjecture. Farhi et al.~\cite{farhi2020quantumw} improved on this via an indistinguishability argument which utilized the fact that local neighborhoods of random $d$-regular graphs are trees with high probability to then conclude that there are $d$-regular graphs (specifically random bipartite ones) on which $\QAOA$ wouldn't do better than $\frac{1}{2} + O(\frac{1}{\sqrt{d}})$ for sufficiently large $n$.

\paragraph{Local algorithms for spin glasses.}
The performance and limitations of various algorithms, such as factors of i.i.d.\ and message passing algorithms, have been established on different models of spin glasses \cite{chen2019suboptimality, gamarnik2021overlap, montanari2021optimization, elalaoui2021optimization}. In particular, the literature often provides two kinds of results: An arbitrary approximation to the ground state in the \emph{absence} of an Overlap Gap Property via an appropriate algorithm \cite{elalaoui2021optimization, montanari2021optimization} or a barrier to arbitrary arbitrary approximation for some family of algorithms in the \emph{presence} of an Overlap Gap Property \cite{chen2019suboptimality, gamarnik2021overlap}. The first work, to the authors' knowledge, that analyzed the performance of $\QAOA$ on a spin glass model was by Farhi et al.~\cite{farhi2019quantum}. In this work, \cite{farhi2019quantum} provide an analytic expression for the expected value that $\QAOA_p$ outputs on typical instances of the SK model, which can be evaluated by a "looping procedure" implemented on a circuit with $O(16^p)$ gates.~Numerical results are provided demonstrating evidence that at $p = 11$ this beats the best known SDP-based solver. In this paper, we show that the \emph{Overlap Gap Property} of diluted $k$-spin glasses poses an obstacle for fixed angle $\QAOA_p$ when $p < \epsilon\log(n)$. The generalization to the $k$-spin mean field model is substantially more challenging to analyze, as in that setting the $\QAOA_p$ algorithm is not local even at depth $p = 1$. However, a \emph{coupled} Overlap Gap Property is known to exist for the $k$-spin mean field model \cite{chen2019suboptimality}. 

\subsection{Our results}\label{subsec:results-informal}
In this work, we show that at shallow-depth the $\QAOA$ algorithm cannot output a spin configuration that has Hamiltonian $(1-\epsilon_0)$-close to the $H^*$ in a random diluted $k$-spin glass.

\begin{theorem}[Obstruction to $\QAOA$ over diluted $k$-spin glass, informal]\label{thm:main-informal-one}
For every even $k\geq4$, there exists $d_0\in\N$ and the following holds: There exists $\epsilon_0 > 0$ such that if $\QAOA_p$ outputs a solution $\sigma\in\{-1,1\}^n$ with $H(\sigma)$ being $(1-\epsilon_0)$-close to the $H^*$ of a random diluted $k$-spin glass of average degree $d\geq d_0$ with probability at least $0.99$, then $p=\Omega(\log n)$.
\end{theorem}
The formal version of the theorem is stated in~\autoref{thm:obstruction-1}. This result can be interpreted as a weak obstruction to logarithmic-depth $\QAOA$ in approximating a random diluted $k$-spin glass, which is equivalent to the random \kxors{} problem when all clauses check for odd parity of non-negated variables.
We also demonstrate the same result for the general case of random \kxors{}.

\begin{theorem}[Obstruction to $\QAOA$ on random \kxors, informal]\label{thm:main-informal-kxors}
For every even $k\geq4$, there exists $d_0\in\N$ and the following holds: There exists $\epsilon_0 > 0$ such that if $\QAOA_p$ outputs a solution $\sigma\in\{-1,1\}^n$ with $H(\sigma)$ being $(1-\epsilon_0)$-close to the $H^*$ of a random \kxors{} instance of average degree $d\geq d_0$ with probability at least $0.99$, then $p=\Omega(\log n)$.
\end{theorem}
This is stated formally in \autoref{thm:obstruction-kxors}, and answers a question of~\cite{farhi2019quantum}, where the authors ask if $QAOA_p$ would perform well on $k$-spin generalizations of the SK model, citing \kxors{} in particular~\cite{chen2019suboptimality}.

In fact, we can prove results stronger in three ways: (i) the same approximation resistance holds for a more general family of algorithms defined as generic local algorithms (\autoref{def:local}), (ii) the same approximation resistance holds for a broader family of optimization problems (\autoref{def:k-d-csp}) provided they exhibit a certain solution geometry (\autoref{def:coupled-ogp-informal}), and (iii) we show that random \kxors{} with negations is one of the optimization problems with this geometry (\autoref{thm:ogp-kxors-coupled}), whereas previous work only handled it without negations so that all clauses needed to be of odd parity. We begin by informally introducing generic local algorithms, random constraint satisfaction problems (CSPs), and the coupled overlap gap property (OGP).

\paragraph{Generic local algorithms.}
As traditional notions of local algorithms do not capture $\QAOA$\footnote{This is made formal in \autoref{prop:p-local-vs-factors}}, we generalize the definition to a broader family and call it \emph{generic local algorithms}. A randomized algorithm $A$ on a hypergraph $G=(V,E)$ can be viewed as outputting labels $A(G)\in S^V$ from a label set $S$ (e.g., $S=\{-1,1\}$). As both $A$ and $G$ are random, $A(G)$ is a set of random variables and the \emph{independence structure} of $A(G)$ captures how local $A$ is. Next, for a hypergraph $G$ and a vertex set $L\subset G$, the $p$-neighborhood of $L$ is the induced subgraph of $G$ by the vertices that can reach $L$ in $p$ steps. 

\begin{definition}[Generic local algorithms, informal]\label{def:local informal}
Let $p\in\N$ and let $S$ be a finite label set.
We say an algorithm $A$ (which takes a hypergraph $G$ as an input) is generic $p$-local if the following hold:
\begin{itemize}
    \item (Local distribution determination.) For every set of vertices $L \subset V$, the joint marginal distribution of the labels $(A(G)_v)_{v\in L}$ depends only on the union of the $p$-neighborhoods of $v 
    \in L$ in $G$.
    \item (Local independence.) $A(G)_{v}$ is statistically independent of the joint distribution of $\{A(G)_{v'}\}$ for every $v'$ that is farther than a distance of $2p$ from $v$.
\end{itemize}
\end{definition}

The main difference between our notion of generic local algorithms and the ones used by previous works is that~\autoref{def:local informal} captures the \textit{evolution of correlations}, without assuming any concrete model of randomness. This is crucial in the interpolation step of the proof (see~\autoref{sec:tech overview}). See~\autoref{def:local} for the formal definition.

\begin{definition}[Random $(k, d)$-$\mathsf{CSP}(f)$]\label{def:k-d-csp}
    A (signed) random $(k, d)$-$\mathsf{CSP}(f)$ instance with a local constraint function $f: \{-1, 1\}^k \to \{0, 1\}$ is constructed as follows:
    \begin{enumerate}
        \item Choose $r \sim \mathrm{Poisson}(dn/k)$.
        \item Sample $r$ clauses of size $k$ by choosing each clause $C_i$ independently as a collection of $k$ variables uniformly at random from $\{x_1,\dots,x_n\}^k$, and, in the case of a signed random CSP, random signs $s_{i,1}$, $\dots$, $s_{i,k} \in \{\pm 1\}$.
    \end{enumerate}
    To each clause $C_i$ there are $k$ variables associated: $\{x_{i_1},\dots,x_{i_k}\}$. A clause is satisfied if there is some assignment to every $x_{i_j} \in \{-1, 1\}$, such that, $f(x_{i_1},\dots,x_{x_k}) = 1$ (or $f(s_{i,1}x_{i_1},\dots,s_{i,k}x_{x_k}) = 1$ if signed). The value of an assignment $\sigma\in\{-1,1\}^n$ is defined as $\mathsf{val}_{\Psi}(\sigma):=\#\{C_i: f(\sigma_{i_1},\dots,\sigma_{i_k})=1\}$ (or $\#\{C_i: f(s_{i,1}\sigma_{i_1},\dots,s_{i,k}\sigma_{i_k})=1\}$ if signed). The optimal value of $\Psi$ is defined as $\mathsf{val}(\Psi):=\max_{\sigma}\mathsf{val}_{\Psi}(\sigma)$.
\end{definition}
When unspecified, we will be referring to unsigned CSPs.

We say that a random (un)signed $(k, d)$-$\mathsf{CSP}(f)$ satisfies a \textit{coupled overlap-gap property (OGP)} if, given two instances  $\Psi, \Psi'$ constructed so that they share a random $t$-fraction of clauses with the remaining $(1-t)$-fraction chosen independently, any two ``good" solutions $\sigma$ of $\Psi$ and $\sigma'$ of $\Psi'$ are either very similar of dissimilar.

\begin{definition}[Coupled OGP, informal]\label{def:coupled-ogp-informal}
A signed or unsigned $(k, d)$-$\mathsf{CSP}(f)$ satisfies a coupled OGP if there exists $\epsilon_0>0$ and $0<a<b<1$ such that the following hold for every $t \in [0, 1]$: Given two $(k, d)$-$\mathsf{CSP}(f)$ instances $\Psi, \Psi'$ constructed so that they share a random $t$-fraction of their clauses and have the remaining $(1-t)$-fraction of clauses chosen independently and uniformly at random, then for every $0<\epsilon<\epsilon_0$, the overlap between any $(1-\epsilon)$-optimal solution $\sigma$ of $\Psi$ and $\sigma'$ of $\Psi'$ satisfies
\[
\frac{1}{n}\langle \sigma, \sigma' \rangle \notin [a, b]\,
\]
with high probability.
\end{definition}

A formal definition of the property above is provided in \autoref{thm:ogp-spin-glasses-coupled}, and the formal definition of the interpolation used to create two ``$t$-coupled" instances is given in \autoref{def:coupled-interpolation}. Note that an instance of a  $(k, d)\text{-}\mathsf{CSP}(f)$ can be thought of as a random sparse $k$-hypergraph, and this is made more precise in the proof of \autoref{thm:qaoa-conc-log-depth}.

Now, we are able to state the most general form of our main result.
\begin{theorem}[Obstruction to generic local algorithms given coupled OGP, informal]\label{thm:informal-obstruct-everything}
For every $k\geq2$, and a constraint function $f:\{-1,1\}^k\to\{0,1\}$, suppose there exists $d_0$ where a random signed or unsigned $(k, d)$-$\mathsf{CSP}(f)$ satisfies the coupled OGP (i.e., \autoref{def:coupled-ogp-informal}) for every $d\geq d_0$ and $p(n)$ is such that it satisfies the requirements of~\autoref{lem:concentration-local-functions} and~\autoref{lem:concentration-random-coupled-hypergraphs}, then the following holds: There exists $\epsilon_0>0$ such that a generic $p(n)$-local algorithm cannot output a solution that is better than $(1-\epsilon_0)$-optimal with high probability.
\end{theorem}
The formal version of the above theorem is stated in~\autoref{thm:obstruction-2}. The theorem effectively obstructs \emph{any} algorithm that makes assignments for variables by looking at $o(n)$ sized local neighborhoods \emph{irrespective} of how these decisions are made and what kind of randomness is used, provided the problem exhibits a coupled OGP.

\paragraph{Confirmation of landscape independence.}
As a consequence of our proof techniques, we also confirm a prediction of Brandao et al.~\cite{brandao2018fixed} in the $\Theta(\log n)$-depth regime for $\QAOA$ by showing that the output values of $\QAOA$ on a random $(k, d)$-$\mathsf{CSP}(f)$ instance (with depth $p$ as stated in \autoref{thm:vanish-nghbd}) concentrate very heavily around the expected value. Once again, the expectation here is with respect to the input distribution as well as the internal randomness of the algorithm.

\begin{theorem}[Confirmation of landscape independence, informal]
    Given a random instance $\Psi$ of a $(k, d)$-$\mathsf{CSP}(f)$ chosen as stated in \autoref{def:k-d-csp}, and a $QAOA_p$ circuit with depth $p < g(d, k)\log(n)$ for some function $g$, the solution $\sigma$ output by $QAOA_p$ with value $\val_{\Psi}(\sigma)$ concentrates as,
    \[
        \Pr\left[\left|\val_{\Psi}(\sigma) - \lE[\val_{\Psi}(\sigma)]\right| \geq \delta n\right] \leq o_n(1)\, ,
    \]
    for every every $\delta > 0$ and the probability taken over both the input distribution and internal randomness of the algorithm.
\end{theorem}
The theorem above is made formal in \autoref{thm:qaoa-conc-log-depth}, and the proof follows by encoding a random $(k, d)$-$\mathsf{CSP}(f)$ instance in a random sparse $k$-hypergraph and then applying \autoref{cor:conc-p-local-energy} with the local function on every hyperedge being set to $f$.

\paragraph{Discussion and open problems.}
Our results reveal that a coupled OGP is tightly related to the obstruction of $\QAOA$. This motivates many open problems that are either inspired by this work or are closely related to it, and these are discussed in greater detail in~\autoref{sec:conjectures}.

\subsection{Technical overview}\label{sec:tech overview}
Our proofs for the main theorems follow the analysis framework of~\cite{chen2019suboptimality}, which shows the approximation resistance of random diluted $k$-spin glasses to a weaker\footnote{In particular, $\QAOA$ is not captured by factors of i.i.d.\hspace{-1mm} and we show a separation in~\autoref{prop:p-local-vs-factors}. Refer to~\autoref{sec:proof-bell-generalization} for more details.} class of classical algorithms called \textit{factors of i.i.d.}\hspace{-1mm} local algorithms. We start with briefly giving an overview of their proof and pointing out where their analysis does not extend to $\QAOA$. See also~\autoref{fig:tech overview} for a pictorial overview. %

\paragraph{Chen et al.~\cite{chen2019suboptimality} analysis.} They establish a coupled overlap-gap property (OGP) for diluted $k$-spin glasses (\autoref{thm:ogp-spin-glasses-coupled}). The property says that for two ``coupled" random instances and any nearly optimal solutions $\sigma_1,\sigma_2\in\{-1,1\}^n$ of these, the solutions either have large or small overlap on the assignment values to the variables, i.e., there exists an interval $0<a<b<1$ such that $\langle\sigma_1,\sigma_2\rangle/n\notin[a,b]$. The coupled OGP holds over an interpolation of a pair of hypergraphs $\{(G_1(t),G_2(t))\}_{t\in[0,1]}$ with the following three properties: for every $t\in[0,1]$, denote $\sigma_1(t)$ and $\sigma_2(t)$ as the outputs of a factors i.i.d. algorithm on inputs $G_1(t)$ and $G_2(t)$ respectively. (i) when $t=0$, $(G_1(0),G_2(0))$ are independent random hypergraphs and $\langle\sigma_1(0),\sigma_2(0)\rangle/n < a$ with high probability; (ii) when $t=1$, $G_1(1)=G_2(1)$ are the same random hypergraph and $\langle\sigma_1(1),\sigma_2(1)\rangle/n=1$ with high probability; (iii) for each $t\in[0, 1]$, the correlation $\langle\sigma_1(t),\sigma_2(t)\rangle/n$ between the two solutions is highly concentrated (with respect to the randomness of $G_1(t),G_2(t)$ and the algorithm) to a value $R(t)$, and $R(t)$ is a continuous function of $t$. This contradicts the OGP if the solutions are nearly optimal and hence no such factors of i.i.d.\hspace{-1mm} algorithm can exist. Note that it is also important to assert that the hamming weight and the objective function values output by the algorithm also concentrate. %

\paragraph{Our analysis.} The key part of the Chen et al.~\cite{chen2019suboptimality} proof that does not work for $\QAOA$ is item (iii) of step 2. Specifically, $\QAOA$ is not a factors of i.i.d.\hspace{-1mm} local algorithm and hence their concentration analysis on the correlation between solutions to coupled instances does not apply. Intuitively, this is because local quantum circuits can induce entanglement between qubits in a local neighborhood which cannot be explained by a local hidden variable theory~\cite{bell1964einstein}. We overcome this issue by first generalizing the notion of factors of i.i.d.\hspace{-1mm} algorithms to what we call \textit{generic local algorithms} (\autoref{def:local}). %

To establish concentration of overlap for generic local algorithms, the challenge lies in how to capture the local correlations of $G_1(t)$ and $G_2(t)$. We achieve this by defining a new notion of a random vector being \textit{locally independent} (\autoref{def:local indp vec}). 
The locally independent structure enables us to show concentration on a fixed instance over multiple runs of the generic local algorithm with respect to its internal randomness (\autoref{lem:concentration-local-functions}). Finally, to establish concentration between a pair of correlated instances ($G_1(t)$ and $G_2(t)$), we strengthen McDiarmid's inequality for biased distributions (\autoref{lem:stronger-mcdiarmids}) and this allows the concentration analysis of the correlation function $R(t)$ to pull through (\autoref{thm:coup-hypergraph-qaoa-output-overlap}). We complete the analysis by showing that the hamming weight and objective function values output by a generic $p$-local algorithm also concentrate (\autoref{cor:conc-p-local-hamming-wt}, \autoref{cor:conc-p-local-energy}). In fact, we show this for a broader class of problems (\autoref{thm:qaoa-conc-log-depth}).

\begin{figure}[ht]
    \centering
    \includegraphics[width=16cm]{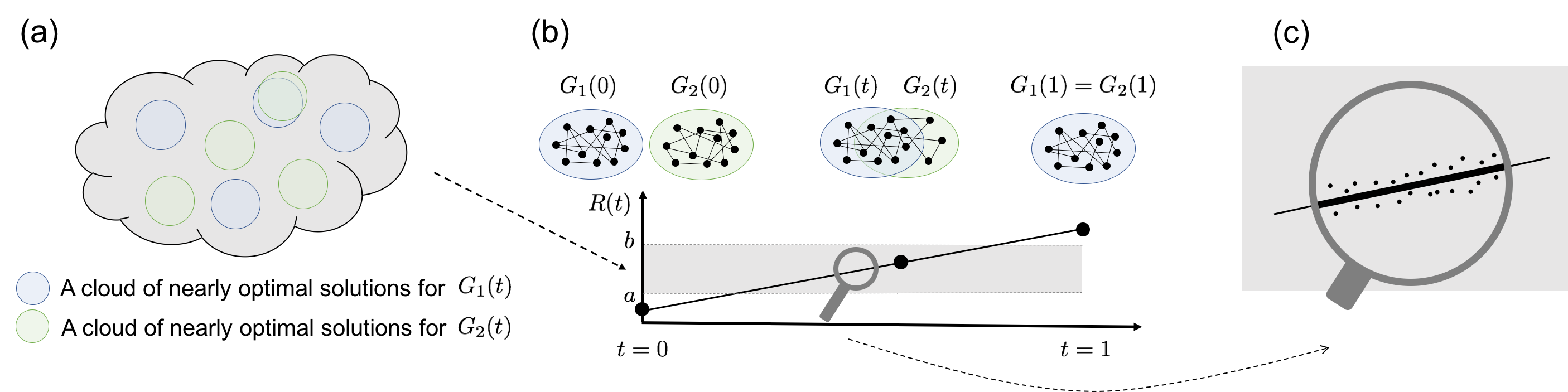}
    \caption{Overview of the proof ideas. (a) The coupled overlap-gap property (OGP) by~\cite{chen2019suboptimality}. Pictorially, the property guarantees that the nearly optimal solutions of a pair of independent instances form multiple disjoint clouds in the solution space $\{-1,1\}^n$. (b) Interpolation of a pair of diluted $k$-spin glass instances. $\{(G_1(t),G_2(t))\}_{t\in[0,1]}$ describes the interpolation from totally independent (i.e., $t=0$) to totally correlated (i.e., $t=1$). $R:[0,1]\to[0,1]$ is the correlation function of a local algorithm (e.g., $\QAOA$) on this interpolation. The coupled OGP prevents $R(t)$ to take values in $[a,b]$ (i.e., the grey area). (c) To contradict the coupled OGP, we would like to show that $R(t)$ is a continuous function and this requires showing that the correlation between the outputs of $G_1(t)$ and that of $G_2(t)$ is concentrated. This is the main challenging step in the proof.}
    \label{fig:tech overview}
\end{figure}

\paragraph{Comparison with Chen et al.~\cite{chen2019suboptimality} and Farhi et  al.~\cite{farhi2020quantum}.} We augment the techniques of \cite{farhi2020quantum} to handle a coupled OGP over a continuous interpolation, as opposed to the coupled OGP in \cite{farhi2020quantum} which is over a fixed discrete interpolation. 
The advantage of this is to enable the use of a broader family of coupled OGPs provable using statistical mechanics methods, whereas the coupled OGP of \cite{farhi2020quantum} requires reasoning about explicit sequences of instances in a way that does not clearly generalize from their independent set analysis to the setting of CSPs.
Our statements additionally are more general and show much stronger concentration than those of \cite{chen2019suboptimality} which is necessary to demonstrate that polynomially many runs of the algorithm will (with high probability) not succeed.
We also show a locality bound of $\log n$ instead of constant, requiring different techniques for analyzing locality than those used in \cite{chen2019suboptimality}, who study only regimes where all neighborhoods are locally isomorphic to trees.
Finally, we demonstrate the coupled OGP (and therefore obstruct generic local algorithms) for general-case random \kxors{}, rather than the case where all clauses require odd parity of their associated variables, without negations.


\section{Preliminaries}\label{sec:prelim}
In this paper, we adopt the following conventions on notations in a CSP (in a spin glass model). 
$n$ denotes the number of variables (the number particles); $k$ denotes the arity of a constraint (the number of particles involved in an interaction) and $k\geq4$ throughout the paper; $m$ denotes the number of hyperedges (the total number of non-trivial interactions); $d$ denotes the degree of a variable (the number of interactions a particle is involved in on average).

The rest of this section is organized as follows. We first recall some elementary definitions and results from spin glass theory in~\autoref{sec:prelim spin glass}. Then, we formally define local algorithms such as factors of i.i.d.~and $\QAOA$ in~\autoref{sec:prelim local alg}. We then provide the complete definition of the OGP and coupled OGP as well as some relevant theorems in~\autoref{sec:prelim ogp}. Finally, in~\autoref{sec:vanishing local neighborhood} we end with a statement which states that sufficiently local neighborhoods of sparse random hypergraphs see a vanishing fraction of the total hypergraph.

\subsection{Spin glass theory}\label{sec:prelim spin glass}
As introduced in~\autoref{sec:intro diluted spin glass}, a spin-glass model is specified by a collection of interactions $\{J_{i_1,\dots,i_k}\}$ on $n$ particles. Physicists are interested in studying Hamiltonians of the form,
\[
H(\sigma) = \sum_{i_1,\dots,i_k}J_{i_1,\dots,i_k}\sigma_{i_1}\cdots\sigma_{i_k}\, ,
\]
where $\sigma\in\{-1,1\}^n$ is the spin configuration. Specifically, it is of importance to understand the \textit{ground state energy}, i.e., $H^*:=H(\sigma)$, as well as the spin configurations $\sigma$ that have energy $H(\sigma)$ close to $H^*$. Note that this naturally connects spin glass theory to combinatorial optimization problems where $\{J_{i_1,\dots,i_k}\}_{i_1,\dots,i_k\in[n]}$ specifies the input constraints, $\sigma$ corresponds to the variables, and $H(\sigma)$ is the objective function.

In condensed matter physics, it is critical to understand the average-case setting and, therefore, the interactions $\{J_{i_1,\dots,i_k}\}_{i_1,\dots,i_k\in[n]}$ are sampled from a certain distribution. We now introduce two common spin glass models:
diluted $k$-spin glasses (\autoref{sec:prelim diluted k spin}) and the $k$-spin mean field model (\autoref{sec:prelim k spin mean field}).

\subsubsection{Diluted $k$-spin glasses}\label{sec:prelim diluted k spin}
In the diluted $k$-spin glass model, the interactions are sampled from a random sparse k-uniform hypergraph defined as follows.

\begin{definition}[Hypergraphs on $n$ Vertices]
    A hypergraph $G$ on $n$ vertices with $m$ hyperedges is characterized by its set of labelled vertices $V = \{1, \dots, n\}$ and hyperedges $E = \{e_1, \dots, e_m\}$, where every hyperedge $e_i = (v_{i, 1}, \dots, v_{i, k}) \in E$ is an ordered $k_i$-tuple in $V$, and $k_i \in \mathbb{N}$, $\forall i\ \in [m]$.
\end{definition}
We restrict our attention to sparse instances of such hypergraphs, which amounts to asserting that the number of hyperedges is $m = nd/k = O(n)$. Additionally, we also restrict to the case that our hyperedges are $k$-uniform, that is, each of them contains $k$-vertices. For the rest of the paper, we will always assume that $d = O(1)$, $k\geq 4$ and $k$ is even, and that $\mathcal{H}_{n, d, k}$ denotes the set of all such $k$-uniform hypergraphs over $n$ vertices with $nd/k$ hyperedges.
\begin{definition}[Random Sparse $k$-Uniform Hypergraphs]\label{def:dil-k-spin-glass}
    A hypergraph $G \sim \mathcal{H}_{n, d, k}$ is chosen by first choosing the number of edges $m=|E| \sim \mathrm{Poisson}(dn/k)$, and then choosing hyperedges $e_1,\dots,e_{m}$ i.i.d. uniformly at random from the set $[n]^k$ of all vertex $k$-tuples.
\end{definition}
Let $G$ be a sparse $k$-uniform hypergraph, the Hamiltonian of the corresponding diluted $k$-spin glass is
\begin{equation}\label{eq:diluted-ham}
    H^G(\sigma) = -\sum_{i=1}^m\prod_{j = 1}^k\sigma_{v_{i,j}}\, ,
\end{equation}
where $\sigma_{v_{i,j}}$ denotes the spin of the $j$-th vertex in the $i$-th hyperedge. Note that maximizing this Hamiltonian corresponds to finding a configuration $\sigma$ such that $H^G(\sigma)$ is maximized. \\

\paragraph{Correspondence to \kxors.}
Maximizing the Hamiltonian in a diluted $k$-spin glass is equivalent to maximizing the number of satisfying constraints in a certain instance $\Psi$ of \kxors. Recall that a \kxors\ instance consists of constraints of the form $x_{i_1}\oplus\cdots\oplus x_{i_k} = b_i$. Let $x\in\{0,1\}^n$ be a boolean assignment, the value of $x$ on $\Psi$ is then defined as $\val_\Psi(x):=\sum_{i}b_i \oplus x_{i_1}\oplus\cdots\oplus x_{i_k}$. Let $G$ be a sparse $k$-uniform hypergraph, we associate it with a \kxors\ instance $\Psi^G$ with constraint $x_{v_{i,1}}\oplus\cdots x_{v_{i,k}} = 0$ for every $i=1,\dots,m$. Finally, we associate a spin configuration $\sigma\in\{-1,1\}^n$ to a boolean assignment $x\in\{0,1\}^n$ by sending $-1\mapsto 1$ and $1\mapsto 0$. Thus, the product of spins on a hyperedge $e$ is mapped to the parity of the corresponding boolean variables:
\[
    \prod_{j=1}^k\sigma_{v_{i,j}} = (-1)^{\bigoplus_{j=1}^kx_{v_{i,j}}}
\]
for each $i=1,\dots,m$. Moreover, the Hamiltonian and the CSP value have the following correspondence.
\begin{equation*}
    H^{G}(\sigma) = -\sum_{i=1}^m \prod_{j=1}^k\sigma_{v_{i,j}} = \sum_{i=1}^m \left(2\cdot(-1)^{\bigoplus_{j=1}^kx_{v_{i,j}}}-1\right) = 2\val_{\Psi^G}(x) - m \, .
\end{equation*}
That is, maximizing that Hamiltonian $H^G$ is equivalent to maximizing the value $\val_{\Psi^G}$. As a remark, note that when $k=2$, \kxors\ becomes \maxcut. As a \emph{signed} extension of the diluted $k$-spin glass hamiltonian, one can define a hamiltonian for the random~\kxors~problem.

\begin{definition}[Random \kxors]\label{def:random-k-xor}
    Sample a hypergraph $G \sim \mathcal{H}_{n, d, k}$ by the same procedure mentioned in~\autoref{def:dil-k-spin-glass}. The hamiltonian corresponding to the random \kxors~instance $\Psi_G$ generated by this hypergraph $G$ is
    \begin{equation}
        H^G_{\mathrm{signed}} = -\sum_{i=1}^m\prod_{j=1}^k p_{ij}\sigma_{v_{ij}}\, ,
    \end{equation}
    where every $p_{ij} \sim \{\pm 1\}$ is an i.i.d.~Rademacher random variable.
\end{definition}

\paragraph{Typical behavior.}
It is important to understand the typical value of $H^G$ when $G$ is a random sparse graph. For example, the following quantity
\begin{equation}\label{eq:diluted-ham-limit}
    M(k, d) := \lim_{n \rightarrow \infty}\lE_{G\sim\mathcal{H}_{n,d,k}}\left[\max_{\sigma \in \{-1,1\}^n} \frac{H^G(\sigma)}{n} \right] \, ,
\end{equation}
is a well-defined limit whose existence is inferred from arguments similar to those presented in \cite{de2004random}. Furthermore, by standard concentration arguments, the ground state energy concentrates heavily around $M(k, d)$. The exact computation of the value of $M(k, d)$ is beyond the scope of this paper. However, as mentioned in \autoref{sec:prelim k spin mean field}, the value can be related to the free-energy of a typical instance of the $k$-spin mean field Hamiltonian in the large-$d$ regime.

\subsubsection{$k$-spin mean field model}\label{sec:prelim k spin mean field}
The $k$-spin mean field model is a special case of the infinite-range model with each interaction $J_{i_1,\dots,i_k}=g_{i_1,\dots,i_k}/\sqrt{n^{k-1}}$ where $g_{i_1,\dots,i_k}$ are i.i.d. standard Gaussian random variables. Just as in the case of diluted $k$-spin glasses, spin configurations that maximize the Hamiltonian are of particular interest. Specifically, we are interested in typical ground state configurations in the thermodynamic limit ($n \rightarrow \infty$). The optimal (normalized) value of the ground state is characterized by the following term,
\begin{equation}\label{eq:mean-field-ham-limit}
    P(k) := \lim_{n \rightarrow \infty}\frac{1}{n} \lE_{g_{i_1,\dots,i_k}}\left[\max_{\sigma \in \{-1,1\}^n}H(\sigma)\right]\, ,
\end{equation}
where $P(k)$ denotes the famous Parisi constant. In a sequence of recent works \cite{dembo2017extremal, sen2018optimization} this limit was precisely related to the limit of the ground state energy of diluted $k$-spin glasses in the large degree limit as,
\begin{equation}\label{eq:diluted-to-mean-field}
    M(k, d) \overset{d \rightarrow \infty}{\longrightarrow} \sqrt{\frac{d}{k}}\cdot P(k) + o\left(\sqrt{d}\right)\, .
\end{equation}

\subsection{Local algorithms}\label{sec:prelim local alg}
A local algorithm assigns a (random) label to each vertex $v$ independently at the beginning and then updates it based on the labels of a small neighborhood of $v$. 
Intuitively, the labels associated to the vertices form a stochastic process and in the end the local algorithm assigns a value to each vertex according to its final label.
In~\autoref{sec:factors-of-iid} we introduce factors of i.i.d.~algorithms~\cite{gamarnik2014limits, chen2019suboptimality}. These algorithms parameterize a family of local algorithms that capture most common classical local algorithms. We then introduce the QAOA in~\autoref{sec:qaoa-algorithm}.

\subsubsection{Factors of i.i.d.~algorithms}\label{sec:factors-of-iid}
A local algorithm takes an input (hyper-) graph $G$ and a label set $S$, runs a stochastic process $\{X^G(t)\}_{t}$ that associates to each vertex $v$ a label $X^G_v(t)\in S$ at time $t$, and outputs the assignment $\sigma_v$ to each vertex $v$ according to its final label. While there is a huge design space for local algorithms, a factors of i.i.d.~algorithm of radius $p$ has the following restrictions: (i) the initial label for each vertex $v$ is set to be an i.i.d.~set of random variables $X^G_v(0)$. (ii) For each vertex $v$, the assignment $\sigma_v$ is a random variable that only depends on the labels from a $p$-neighborhood of $v$. (iii) The assignment function for each vertex is the same. Common local algorithms such as Glauber dynamics and Belief Propagation are examples of factors of i.i.d.~algorithms.

To be more concrete, let us start with a formal definition of the $p$-neighborhood of a vertex in a hypergraph, which is a generalization from the $p$-neighborhood of a graph by considering two vertices $v$ and $w$ to be adjacent if they belong to the same hyperedge $e$.

\begin{definition}[$p$-neighborhood and hypergraphs with radius $p$]\label{def:p-neighborhood}
Let $G$ be a hypergraph, $v\in V(G)$, and $p\in\N$. The $p$-neighborhood of $v$ is defined as
\[
    B_G(v, p) := \{ w\in V(G)\, |\, w\text{ is }p\text{ hyperedges away from }v\}\, .
\]
Let $G$ be a hypergraph, $v\in V(G)$, and $p\in\N$. We say $(G,v)$ has radius $p$ if $B_p(G_v)=G$. Further, let $k\in\N$, we define
\[
\mathcal{G}_p := \left\{ (G,v)\, |\, \text{$(G,v)$ has radius $p$ and $G$ is connected, finite, and $k$-uniform}  \right\}
\]
be the collections of hypergraphs with radius at most $p$.
\end{definition}

Next, to capture the fact that local algorithms assign the value of a vertex $v$ by only looking at a $p$-neighborhood, it is natural to define an equivalent classes of local induced subgraphs rooted at $v$ as follows.

\begin{definition}[Rooted-isomorphic graphs]\label{def:rooted-isomorphism}
Let $G_1,G_2$ be two hypergraphs and $v\in V(G_1)\cap V(G_2)$. We say $G_1$ and $G_2$ are rooted-isomorphic at $v$, denoted as $G_1 \cong_v G_2$, if there exists a hypergraph isomorphism $\phi: V(G_1) \to V(G_2)$ such that $\phi(v) = v$.
Similarly, let $L\subseteq V(G_1)\cap V(G_2)$, we say $G_1 \cong_L G_2$ if there exists a hypergraph isomorphism $\phi: V(G_1) \to V(G_2)$ with $\phi(v) = v$ for all $v \in L$.

\end{definition}

In the future usage of~\autoref{def:rooted-isomorphism}, we think of $G_1$ and $G_2$ as some neighborhoods. Intuitively, when the neighborhood of $v_1$ and $v_2$ are rooted-isomorphic, then the local algorithm will give the same output to them.

The last notion of local algorithms to capture is the assigning process $f$ from the labels of a $p$-neighborhood to a value. In particular, a local algorithm should produce the same output value for $v_1$ and $v_2$ when the induced subgraphs of their $p$-neighborhood are rooted-isomorphic. For simplicity, we focus on the case where the label set $S=[0,1]$.

\begin{definition}[Factor of radius $p$, {\cite[Section 2]{chen2019suboptimality}}]
Let $p\in\N$. We define the collection of all $[0,1]$-labelled hypergraphs of radius at most $p$ as 
\[
\Lambda_p := \left\{ (G,v,X) \, |\, \text{$(G,v)\in\mathcal{G}_p$ and $X\in[0,1]^{V(G)}$} \right\}
\]

We say $(G_1, v_1, X_1), (G_2, v_2, X_2) \in \Lambda_r$ are isomorphic if there exists a hypergraph isomorphism $\phi:V(G_1)\to V(G_2)$ such that (i) $\phi(v_1)=v_2$ and (ii) $X_2\circ\phi=X_1$.

Finally, we say $f: \Lambda_p \rightarrow \{-1, 1\}$ is a factor of radius $p$ function if (i) it is measurable and (ii) $f(G_1,v_1,X_1)=f(G_2,v_2,X_2)$ for every isomorphic $(G_1, v_1, X_1), (G_2, v_2, X_2) \in \Lambda_p$.
\end{definition}

Now, we are ready to define factors of i.i.d.~algorithms. Intuitively, the output distribution of a factors of i.i.d. algorithm with radius $p$ on a vertex $v$ is determined by the $p$-neighborhood of $v$.

\begin{definition}[Factors of i.i.d., {\cite[Section 2]{chen2019suboptimality}}]
Let $k,p\in\N$. A factors of i.i.d.~algorithm $A$ with radius $p$ is associated with a factor of radius $p$ function $f$ with the following property. On input a $k$-uniform hyper graph $G$, the algorithm $A$ samples a random labeling $X=\{X(v)\}_{v\in V(G)}$ where $X(v)$'s are i.i.d.~uniform random variables on $[0,1]$. The output of $A$ is $\sigma\in\{-1,1\}^{V(G)}$ where
\[
\sigma_v := f(B_p(G,v), v, \{X(w)\}|_{w\in B_p(G,v)})
\]
for each $v\in V(G)$.
\end{definition}

\subsubsection{The QAOA algorithm}\label{sec:qaoa-algorithm}
\paragraph{The algorithm.}
The $\QAOA$ algorithm was proposed by Farhi et al.~\cite{farhi2014quantum} as a way to approximately solve hard combinatorial optimization problems. The $\QAOA$ algorithm works by applying, in alternation, weighted rotations in the $X$-basis to introduce mixing over the uncertainty in the solution space and weighted cost $e^{-i\beta_j H_c(G)}$ unitaries to introduce correlation spreading encoded by the Hamiltonian of the desired cost-function to maximize. This weighting is accomplished by giving some assignment of weights to weight vectors $\hat{\beta} = (\beta_1,\dots,\beta_p)$ and $\hat{\gamma} = (\gamma_1,\dots,\gamma_p)$, and then running a classical optimizer to help find the ones that maximize the output of $\QAOA$. Various methods, including efficient heuristics, to optimize these angles are studied in the literature~\cite{brandao2018fixed, zhou2020quantum}.

The $\QAOA$ circuit parametrized by angle vectors $\hat{\gamma}$ and $\hat{\beta}$ looks as follows,
\[
    U_p(\hat{\beta}, \hat{\gamma}) = \prod_{j=1}^pe^{-i\beta_j\sum_{k=1}^nX_k}e^{-i\gamma_jH_c(G)}\, .
\]
Typically, the initial state on which the circuit is applied is a \emph{symmetric product state}, most notably $\ket{0}^{\otimes n}$ or $\ket{+}^{\otimes n}$.  The expected value that $\QAOA$ outputs after applying the circuit on some initial state $\ket{\psi_0}$ is,
\[
    \bra{\psi_0}U^{\dagger}H_c(G)U\ket{\psi_0}\, .
\]
It is the expectation value above that is optimized (by maximizing) for various choices of $\hat{\beta}$ and $\hat{\gamma}$ under a classical optimizer, and the solution corresponding to this solution comes from a measurement in the $Z$-basis of the state $U_p(\hat{\beta}, \hat{\gamma})\ket{\psi_0}$. We will notate by $\QAOA(\hat{\beta}, \hat{\gamma})$ a $\QAOA$ circuit of depth $2p$ with angle parameters $\hat{\beta}$ and $\hat{\gamma}$. In our regime, we will work with \emph{any} collection of \emph{fixed} angles $(\hat{\beta}, \hat{\gamma})$. The fixed angles regime is necessary when reasoning about the concentration of overlaps of the solutions produced by a $p$-local algorithm on coupled instances with shared randomness. If the angles of $\QAOA$ vary between the coupled instances, then we cannot assert that the coupled instances will share randomness when labeling vertices with identical neighborhoods.

\paragraph{The diluted \texorpdfstring{$k$}{k}-spin glass Hamiltonian and \texorpdfstring{$\QAOA$}{QAOA}.}
We can rewrite $H_{k, d, n}$ as a Hamiltonian for a quantum system by replacing $\sigma_i$ with the Pauli $Z$ matrix. This yields a $k$-local Hamiltonian that the $\QAOA$ ansatz tries to maximize. The Hamiltonian is,
\begin{equation}\label{eq:diluted-ham-quatnum}
    \widehat{H}^G_{k, d, n}(\sigma) = -\hspace{-5mm}\sum_{(v_1,\dots,v_k) \in E(G)}\bigotimes_{i=1}^{k}\sigma_z(v_i)\, ,
\end{equation}
where $\sigma_z(v_i)$ is a $2 \times 2$ Pauli Z matrix for the $i$-th vertex (qubit) in the hypergraph $G$. We want to maximize the following expectation value,
\begin{equation}\label{eq:qaoa-diluted-max}
    \max_{\alpha_j, \beta_j, j \in [p]}\bra{\psi_0}U^{\dagger}H_{k, d, n}(G)U\ket{\psi_0} = \max_{\alpha_j, \beta_j, j \in [p]}-\hspace{-4mm}\sum_{(v_1,\dots,v_k) \in E(G)}\left( \bra{\psi_0}U^{\dagger}\left(\bigotimes_{i=1}^{k}\sigma_z(v_i)\right)U\ket{\psi_0}\right)\, .
\end{equation}

\paragraph{\texorpdfstring{$\QAOA$}{QAOA} at shallow depth.}
Note that the only "spreading" of correlation is introduced by the operator $e^{-i\beta_jH_C(G)}$ which is applied only $p$ times. The hamlitonian in consideration is $k$-local, and therefore, after $p$ operations a qubit $i$ will interact with no more than 
\[
|B_G(v_i, p)| \leq \left((k-1)\cdot\max_{i \in [n]}|\Pi_i(E(HG))|\right)^p\, 
\]
vertices, where,
\begin{equation}\label{eq:projection-qubit-int-hyperedges}
    \Pi_i(E(G)) = \{ e \in E(G)\ |\ e\text{ is a hyperedge that contains }v_i\}\, .
\end{equation}
\autoref{prop:qaoa-2p-local} makes a precise statement about the locality of $\QAOA$ with fixed angles. Bounding the size of $\max_{i \in [n]}\Pi_i(E(HG))$ is the main subject of \autoref{thm:vanish-nghbd}, which parameterizes $p$ appropriately as a function of the number of vertices $n$ in the graph (logarithmic) as well as the parameters of $k$ and $d$ so that this size is $o(n)$. This is sufficiently small for the purposes of our obstruction theorem.

\subsection{Overlap-gap properties}\label{sec:prelim ogp}
We now state the OGP as it holds for the diluted $k$-spin glass model in both uncoupled and coupled form. To do so, we begin by introducing the notion of an overlap between two spin-configurations $\sigma_1$ and $\sigma_2$, which is equivalent to the number of spins that are the same in both configurations subtracted by the number of different spins, normalized by the number of particles in the system. Formally,
\begin{definition}[Overlap between spin configuration vectors]\label{def:overlap}
Given any two vectors $\sigma_1, \sigma_2 \in \{-1,1\}^n$, the overlap between them is defined as,
\[
    R(\sigma_1, \sigma_2) = \frac{1}{n}\langle \sigma_1, \sigma_2 \rangle = \frac{1}{n}\sum_{i\in[n]} (\sigma_1)_i(\sigma_2)_i\, .
\]
\end{definition}
We first state the OGP for diluted $k$-spin glasses about the overlap gaps in a \emph{single instance}.
\begin{theorem}[OGP for Diluted $k$-Spin Glasses, {\cite[Theorem 2]{chen2019suboptimality}}]\label{thm:ogp-spin-glasses}
    For every even $k \geq 4$, there exists an interval $0 < a < b < 1$ and parameters $d_0 > 0$, $0 < \eta_0 < P(k)$ and $n_0 > 1$, such that, for $d \geq d_0$, $n \geq n_0$ and $L = L(\eta_0, d)$, with probability at least $1 - Le^{-n/L}$ over the random hypergraph $G \sim \mathcal{H}_{n, d, k}$, whenever two spins $\sigma_1,\ \sigma_2$ satisfy
    \[
        \frac{H^G(\sigma_i)}{n} \geq M(k, d)\left(1 - \frac{\eta_0}{P(k)}\right)\, ,
    \]
    then also, $|R(\sigma_1, \sigma_2)| \notin (a, b)$.
\end{theorem}

A more general version of the OGP excludes, with high probability, a certain range of overlaps between any two solutions of two different instances jointly drawn from a coupled random process. We first introduce this process, and then state the \emph{coupled} version of the OGP as proven in \cite{chen2019suboptimality}.

\begin{definition}[Coupled Interpolation, {\cite[Section 3.2]{chen2019suboptimality}}]\label{def:coupled-interpolation}
The coupled interpolation $\cH_{d,k,n,t}$
generates a coupled pair of hypergraphs $(G_1, G_2) \sim \cH_{d,k,n,t}$ as follows:
\begin{enumerate}
    \item First, a random number is sampled from $\mathrm{Poisson}(tdn/k)$, and that number of random $k$-hyperedges are uniformly drawn from the set $[n]^k$ and put into a set $E$.
    \item Then, two more random numbers are independently sampled from $\mathrm{Poisson}((1-t)dn/k)$, and those numbers of random $k$-hyperedges are independently drawn from $[n]^k$ to form the sets $E_1$ and $E_2$ respectively.
    \item Lastly, the two hypergraphs are constructed as $G_1 = (V, E \cup E_1)$ and $G_2 = (V, E \cup E_2)$. 
\end{enumerate}
\end{definition}

\begin{theorem}[OGP for Coupled Diluted $k$-Spin Glasses, {\cite[Theorem 5]{chen2019suboptimality}}]\label{thm:ogp-spin-glasses-coupled}
    For every even $k \geq 4$, there exists an interval $0 < a < b < 1$ and parameters $d_0 > 0$, $0 < \eta_0 < P(k)$ and $n_0 > 1$, such that, for any $t \in [0,1]$, $d \geq d_0$, $n \geq n_0$ and constant $L = L(\eta_0, d)$, with probability at least $1 - Le^{-n/L}$ over the hypergraph pair $(G_1, G_2) \sim \mathcal{H}_{n, d, k, t}$, whenever two spins $\sigma_1,\ \sigma_2$ satisfy
    \[
        \frac{H^{G_i}(\sigma_i)}{n} \geq M(k, d)\left(1 - \frac{\eta_0}{P(k)}\right) \, ,
    \]
    then their overlap satisfies $|R(\sigma_1, \sigma_2)| \notin [a, b]$.
\end{theorem}

We also provide a corresponding coupled OGP for random \kxors{} in \autoref{thm:ogp-kxors-coupled}.

\subsection{Vanishing local neighborhoods of random sparse \texorpdfstring{$k$}{k}-uniform hypergraphs}\label{sec:vanishing local neighborhood}
We state a bound on sufficiently local neighborhoods of random sparse $k$-uniform hypergraphs. 

\begin{lemma}[Vanishing local neighborhoods of random sparse $k$-uniform hypergraphs]\label{thm:vanish-nghbd}
Let $k \ge 2$ and $d \ge 2$ and $\tau \in (0,1)$. Then there exists $a > 0$ and $0 < A < 1$, such that, for $n$ large enough and $p$ satisfying
\[2p+1 \leq \frac{(1-\tau)\log n}{\log(\frac{d(k-1)}{\ln 2})}, \]
the following are true:
\begin{equation*}
    \Pr_{G\sim\cH_{n,d,k}}[\max_i B_G(v_i, 2p) \geq n^A] \leq e^{-n^a}\, ,
\end{equation*}
and
\begin{equation*}
    \Pr_{G\sim\cH_{n,d,k}}[\max_i B_G(v_i, p) \geq n^{\frac{A}{2}}] \leq e^{-n^{\frac{a}{2}}}\, .
\end{equation*}
\end{lemma}
Intuitively, the above lemma says that the local neighborhood of each vertex is vanishingly small with high probability. 
To prove \autoref{thm:vanish-nghbd}, we utilize a modified version of the proof of Farhi et al.~\cite[Neighborhood Size Theorem]{farhi2020quantum} to handle the case of sparse random hypergraphs and we defer the complete proof to \autoref{sec:proof-vanishing-nhbhd}.

\section{Locality and Shared Randomness}\label{sec:p-local-algorithms}
\subsection{Generic \texorpdfstring{$p$}{p}-local algorithms}

We introduce a concept of ``local random algorithm'' which will allow for different runs of the same local algorithm to "share their randomness", even when run on mostly-different instances.
Later we will demonstrate that $\QAOA$ is a local algorithm under this definition.

\begin{definition}[Generic local algorithms]
\label{def:local}
We consider randomized algorithms on hypergraphs whose output $A(G) \in S^{V}$ assigns a label from some set $S$ to each vertex in $V$. Such an algorithm is generic $p$-local if the following hold.
\begin{itemize}
    \item (Local distribution determination.) For every set of vertices $L \subset V$, the joint marginal distribution of its labels $(A(G)_v)_{v\in L}$ is identical to the joint marginal distribution of $(A(G')_v)_{v\in L}$ whenever $\bigcup_{v\in L} B_{G}(v,p) \cong_L \bigcup_{v\in L} B_{G'}(v,p)$, and,
    \item (Local independence.) $A(G)_{v}$ is statistically independent of the joint distribution of $A(G)_{v'}$ over all $v' \not\in B_{G}(v,2p)$.
\end{itemize}
\end{definition}

Consequently, it will be possible to sample $A(G)_{v}$ without even knowing what the hypergraph looks like beyond a distance of $p$ away from $v$.

This definition is more general than the factors of i.i.d.\ concept used in probability theory \cite{gamarnik2014limits, chen2019suboptimality}. Our definition, for instance, encompasses local quantum circuits whereas factors of i.i.d.\ algorithms satisfy Bell's inequalities and do not capture quantum mechanics.

\begin{proposition}[Generic local strictly generalizes factors of i.i.d.]\label{prop:p-local-vs-factors}
    There exists a generic $1$-local algorithm as defined in \autoref{sec:p-local-algorithms} that is not a $1$-local factors of i.i.d.\ algorithm as defined in \autoref{sec:factors-of-iid}.
\end{proposition}
A proof of this proposition is provided in \autoref{sec:proof-bell-generalization}, and consists of setting up a Bell's inequality experiment within the framework of a generic $1$-local algorithm.

\subsection{Locality properties of \texorpdfstring{$\QAOA$}{QAOA(p)} for hypergraphs}
We show that any $\QAOA$ circuit of depth $p$ with some \emph{fixed} angle parameters $(\hat{\beta}, \hat{\gamma})$ is a $p$-local algorithm. This allows us to describe a process to sample outputs of this circuit when it is run on two different input hypergraphs.

\begin{proposition}
\label{prop:qaoa-2p-local}
For every $p>0$, angle vectors $\hat{\beta}$ and $\hat{\gamma}$,
$\QAOA_p(\hat{\beta}, \hat{\gamma})$ is generic $p$-local under \autoref{def:local}.
\end{proposition}
\begin{proof}
    To see this, consider the structure of QAOA: we start with a product state $\ket{\psi_0}$ where each qubit corresponds to a vertex in the hypergraph, apply the unitary transformation $U = U_p(\hat{\beta}, \hat{\gamma})$ to the state, and then measure each vertex $v$ in the computational basis with the Pauli-Z operator $\sigma_z(v)$.
    Equally valid and equivalent is the Heisenberg picture interpretation of this process, where we keep the product state $\ket{\psi_0}$ fixed but transform the measurements according to the reversed unitary transformation $U^{\dagger}$, so that we end up taking the measurements $U^{\dagger}\sigma_z(v)U$ on the fixed initial state.
    
    Because the $\sigma_z(v)$ operators all commute with each other, their unitarily transformed versions $U^{\dagger}\sigma_z(v)U$ also mutually commute, and the measurements can be taken in any order without any change in results.
    Let $M(v) = U^{\dagger}\sigma_z(v)U$ and $M(L) = \{U^{\dagger}\sigma_z(u)U \mid u \in L\}$.
    
    To show that QAOA satisfies the first property of generic $p$-local algorithms, we need to show that the marginal distribution of its assignments to any set $L' \subseteq V$ of vertices depends only on the union of the $p$-distance neighborhoods of $L'$.
    To show this, since we are allowed to take the measurements in any order, take the measurements in $M(L')$ before any other measurement.
    Then since the action of the unitary $U = U_p(\hat{\beta}, \hat{\gamma})$ on qubits in $L'$ does not depend on any feature of the hypergraph outside of a radius of $p$ around $L'$, the operators $M(L')$ are fully determined by the $p$-local neighborhoods of $L'$, and since we take them before every other measurement, the qubits are simply in their initial states when we make these measurements, thus the distribution of outputs is fully determined.
    
    The same type of reasoning shows that the assignment to each $v \in V$ is statistically independent of the assignments to any set of vertices outside of a $2p$-distance neighborhood of $v$. 
    Take $L'' \subset V \setminus B(v, 2p)$.
    Then $M(v)$ acts on a radius-$p$ ball around $v$, and each measurement in $M(L'')$ acts on a radius-$p$ ball around a vertex in $L''$, and by taking $\{M(v)\} \cup M(L'')$ before any of the measurements in $M(\{B(v, 2p) \setminus \{v\})$, we ensure that the qubits being measured by $M(v)$ are disjoint from and unentangled with those measured by anything in $M(L'')$.
    Hence the measurement $M(v)$ is independent of all measurements in $M(L'')$. We conclude that $\QAOA_p$ is a generic $p$-local algorithm.
\end{proof}

\subsection{Shared randomness between runs of a generic local algorithm}\label{sec:shared-core-measurements}

We describe a process to sample the outputs of a generic local algorithm when run twice on two different hypergraphs, so that the two runs of the algorithm can share randomness when the hypergraphs have some hyperedges in common.

This is not meant as a constructive algorithm, but a statistical process with no guarantee of feasible implementation.

The idea is to start with two $t$-coupled hypergraphs, which for large enough $n$, are likely to have some set of vertices $L^+$ whose $p$-neighborhoods are identical between the two hypergraphs.
Since these vertices have identical $p$-neighborhoods, a generic $p$-local algorithm behaves identically on the vertices in $L^+$.
We pick a random $t^+$ fraction of the elements of $L^+$, and assign the same labels to those vertices in the two coupled instances.
Then the remaining labels on each hypergraph are assigned by generic $p$-local algorithms, conditioned on the output being consistent with the already assigned labels.

\begin{definition}[Randomness-sharing for generic local algorithms] \label{def:coupled-runs}
Let $A$ be a generic $p$-local algorithm, $(G_1,G_2)\sim\mathcal{H}_{d,k,n,t}$, and $S$ be a label set.
A pair of runs with \emph{$t^+$-shared randomness} of $A$ on $G_1 = (V,E \cup E_1)$ and $G_2=(V,E \cup E_2)$ with some shared edge set $E$ is defined follows:
\begin{enumerate}
\item Let $L^+$ be the set of all vertices $v\in V$ such that $E(B_{G_1}(v,p)) \subset E$ and $E(B_{G_2}(v,p)) \subset E$.
Generate the vertex set $L \subset L^+$ by including each element of $L^+$ independently with probability $t^+$.

\item Since $\bigcup_{v \in L} (B_{G_1}(v,p)) = \bigcup_{v \in L} (B_{G_2}(v,p))$, the algorithm has the same joint marginal distribution for its outputs on $L$ when it is run on $G_1$ or $G_2$.
Let $\sigma \in S^L$ be a sample from this joint marginal distribution.

\item Let $\sigma_1$ be a sample of $A(G_1)$, conditioned on $(\sigma_1)_v = \sigma_v$ for all $v \in L$. Similarly for $\sigma_2$ being a conditioned sample of $A(G_2)$.
Then $\sigma_1$ and $\sigma_2$ are individually distributed the same as independent runs of the algorithm on $G_1$ and $G_2$ respectively, and together are the output of the two runs with $t^+$-shared randomness.
\end{enumerate}

\end{definition}

\section{Main Theorems}\label{sec:obstruction-thm-proof}
We formally state our main theorems and give the informal proof sketches in this section. The formal proofs are provided in~\autoref{sec:proofs-main-thms}. First, let us specify the choice of parameters we are going to work with in the rest of the paper.

\begin{parameter}\label{param}
For every even $k\geq4$, there exists $\eta_0>0$ such that the following holds: for every $\tau \in (0, 1)$ and $0<\eta \leq \eta_0$, there exist $d_0,n_0 > 0$ from~\autoref{thm:ogp-spin-glasses-coupled} and we consider running a generic $p$-local algorithm on a random $d$-sparse $k$-uniform hypergraph $G$ with size $n \geq n_0$, degree $d \geq d_0$ and $p$ satisfying,
\[
2p+1 \leq \frac{(1-\tau)\log n}{\log(\frac{d(k-1)}{\ln 2})}\, .
\]
\end{parameter}

\subsection{Obstruction for generic local algorithms on diluted \texorpdfstring{$k$}{k}-spin glasses}

\begin{theorem}[Obstruction theorem for diluted $k$-spin glasses]\label{thm:obstruction-1}
\ourparamwithgamma~ Then, on running a generic $p$-local algorithm on a random $d$-sparse $k$-uniform hypergraph $G$ with size $n$ and degree $d$, the probability that the algorithm will output an assignment that is at least ($1 - \eta$)-optimal is no more than $e^{-O(n^{\gamma})}$.
\end{theorem}
\begin{proof-sketch}
The proof follows the coupled interpolation argument in \cite[Section 3.3]{chen2019suboptimality}, and we sketch it briefly - The expected overlap between coupled solutions is continuous (\autoref{lem:cont-overlap}) with the overlap being less than $a$ at $t = 0$ with high probability if the solutions are nearly optimal (\autoref{lem:overlap-t-0}). The overlap is $1$ at $t = 1$ (\autoref{lem:overlap-t-1}). Concentration of the overlap for any value of $t$ is then shown by invoking \autoref{thm:coup-hypergraph-qaoa-output-overlap}. The intermediate value theorem then immediately yields a contradiction to the \emph{coupled} Overlap Gap Property (\autoref{thm:ogp-spin-glasses-coupled}). This completes the proof for the obstruction.
\end{proof-sketch}

\subsection{Obstructions for generic local algorithms on \texorpdfstring{$(k, d)\text{-}\mathsf{CSP}(f)$}{(k,d)-CSP(f)} with coupled OGP}

\begin{theorem}[Obstructions for $(k, d)\text{-}\mathsf{CSP}(f)$ with coupled OGP]\label{thm:obstruction-2}
\ourparamwithgamma~Let $\Psi$ be a random problem instance of a signed or unsigned $(k,d)\text{-}\mathsf{CSP}(f)$ constructed as in \autoref{def:k-d-csp} that satisfies a coupled OGP --- That is, the hypergraph encoding $G_{\Psi}$ of $\Psi$ satisfies \autoref{thm:ogp-spin-glasses-coupled} with the only difference that $k$ can be \emph{any} number $\geq 2$\ ---\ as well as having the property that the $(1-\eta_0)$-multiplicatively optimal pairs of solutions to two independent instances of the CSP have overlap no more than the lower bound ($a$) of the OGP.
Let $\sigma$ be the output of a generic $p$-local algorithm on $G_{\Psi}$.
Then the following holds:
\[
\Pr_{\Psi, \sigma}[\val_{\Psi}(\sigma) \geq (1 - \eta_0)\val_\Psi] \leq e^{-O(n^{\gamma})}\, .
\]
\end{theorem}
\begin{proof-sketch}
    Once again, we first encode the problem instance $\Psi$ into a representative hypergraph $G_{\Psi}$. Note that, by definition, the encoded instance $G_{\Psi}$ satisfies a coupled OGP as stated in \autoref{thm:ogp-spin-glasses-coupled} over the underlying coupled interpolation stated in \autoref{def:coupled-interpolation}. The concentration of the hamming weight $|\sigma|$ of the solution is established by \autoref{cor:conc-p-local-hamming-wt} and the concentration of the objective value $\val_{\Psi}$ is established by \autoref{thm:qaoa-conc-log-depth}. The concentration of overlap of solutions for coupled instances over the interpolation specified in \autoref{def:coupled-interpolation} is established by \autoref{thm:coup-hypergraph-qaoa-output-overlap}. Then, by an argument similar to the one in the proof sketch of \autoref{thm:obstruction-1}, the obstruction follows.
\end{proof-sketch}

\begin{corollary}[Obstructions for random \kxors{}]\label{thm:obstruction-kxors}
\ourparamwithgamma~Let $\Psi$ be a random problem instance of \kxors{} with $k \ge 4$ even. Let $\sigma$ be the output of a generic $p$-local algorithm on $G_{\Psi}$.
Then the following holds:
\[
\Pr_{\Psi, \sigma}[\val_{\Psi}(\sigma) \geq (1 - \eta_0)\val_\Psi] \leq e^{-O(n^{\gamma})}\, .
\]
\end{corollary}
\begin{proof}
Combine \autoref{thm:obstruction-2} with \autoref{thm:ogp-kxors-coupled} and \autoref{lem:kxors-overlap}.
\end{proof}

\subsection{Concentration of objective function values of \texorpdfstring{$\QAOA$}{QAOA(p)}}
For $p = O(1)$, concentration of the objective function value output by $\QAOA_p(\hat{\beta}, \hat{\gamma})$ for sparse random constraint satisfaction problems (CSPs) is shown in \cite{brandao2018fixed}. We state a result below which extends this to $\QAOA_p$ at depth $p < g(d, k)\log n$. While we state the result for $\QAOA_p(\hat{\beta}, \hat{\gamma})$ specifically, this result will apply to any generic local algorithm.
\cite{brandao2018fixed} cite a barrier in applying their techniques to $\QAOA_p$ at depth greater than $p=O(1)$ due to the limitation of McDiarmid's inequality as stated. We overcome this limitation by strengthening the inequality (\autoref{lem:stronger-mcdiarmids}) for highly biased distributions. This confirms the prediction of \cite{brandao2018fixed} about the ``landscape independence" of $\QAOA_p(\hat{\beta}, \hat{\gamma})$ at depth greater than $p=O(1)$. 

\begin{theorem}[$QAOA_p$ landscape independence at $p < g(d, k)\log n$]\label{thm:qaoa-conc-log-depth}
\ourparam~Let $U(\hat{\beta}, \hat{\gamma})$ be the unitary for a $\QAOA_p(\hat{\beta}, \hat{\gamma})$ circuit.
Furthermore, let the hamiltonian $H_{\Psi}$ encode a problem instance $\Psi$ of a $(k,d)\text{-}\mathsf{CSP}(f)$ constructed as in \autoref{def:k-d-csp}, such that,
\[
H_{\Psi} = \sum_{i=1}^{\text{|E|}}H_i\, ,
\]
where each $H_i$ is a $k$-local hamiltonian encoding $f$ for the $i$-th clause. Then, the output $\ket{\psi} = U(\hat{\beta}, \hat{\gamma})\ket{s}$, where $\ket{s}$ is a symmetric product state, has an objective value that concentrates around the expected value as,
\[
\Pr_{\Psi, \ket{\psi}}[|\bra{\psi}H_{\mathcal{P}}\ket{\psi} - \E_{\Psi, \ket{\psi}}[\bra{\psi}H_{\mathcal{P}}\ket{\psi}]| \geq \epsilon\cdot n] \leq e^{-O(n^{\gamma})}\, ,\ \forall \epsilon > 0\, .
\]    
\end{theorem}
\begin{proof-sketch}
    We encode the problem instance $\Psi$ into a hypergraph $G_{\Psi}$. The constraint function $f$ for every clause is set to be the local energy functions $h_i$ on the appropriate $k$-subset of variables. 
    Then \autoref{lem:concentration-local-functions} and \autoref{lem:concentration-random-coupled-hypergraphs} show concentration.
\end{proof-sketch}

\section{Concentration Analysis for Generic Local Algorithms}\label{sec:concentration-results}

The most technical part of this work is to establish concentration theorems for generic local algorithms. Recall from~\autoref{fig:tech overview} that to get obstruction from the coupled OGP, we have to show that the correlation between the outputs of a generic local algorithm on $t$-coupled hypergraphs is highly concentrated around its expected value $R(t)$. This is formally stated in the following theorem.

\begin{restatable}[Generic local algorithm's outputs overlap on coupled hypergraphs]{theorem}{localoutput}\label{thm:coup-hypergraph-qaoa-output-overlap}
\ourparam~When two random $t$-coupled hypergraphs $(G_1,\ G_2) \sim \mathcal{H}_{d, k, n, t}$ are sampled and a pair of $t$-shared-randomness runs of a generic $p$-local algorithm are made on $G_1$ and $G_2$, the overlap between the respective outputs $\sigma_1$ and $\sigma_2$ concentrates. That is, $\forall\ \delta' > 0$,  $\exists\ \gamma' > 0$, such that,
\begin{equation*}\label{eq:coup-hypergraph-qaoa-output-overlap}
    \Pr_{G_1,G_2,\sigma_1,\sigma_2}\left[\left|\langle \sigma_1,\ \sigma_2\rangle - \!\E_{G_1,G_2,\sigma_1,\sigma_2}[\langle \sigma_1, \sigma_2 \rangle]\right| \geq \delta'\cdot n\right] \leq 2e^{-\delta' n^{\gamma'}}\, ,
\end{equation*}
where the probability and expectation are over both the random sample of hypergraphs and the randomness of the algorithm.
\end{restatable}

To prove~\autoref{thm:coup-hypergraph-qaoa-output-overlap}, we have to show concentration with respect to both the internal randomness of the algorithm and the randomness from the problem instances. It turns out that the former is quite non-trivial due to the correlation between coupled hypergraphs as well as the dependencies introduced by each round of the local algorithm. This results in a generalization and strengthening of \cite[Concentration Theorem]{farhi2020quantum} and \cite[Lemmas 3.1 \& 3.2]{chen2019suboptimality}.

To resolve the correlation issue, we introduce the notion of \textit{locally mixed random vectors} (\autoref{def:local indp vec}) that capture the shared randomness between different runs of the local algorithm. We then show in~\autoref{lem:concentration-local-functions} that the correlation between a locally mixed random vector and the output of a generic $p$-local algorithm will still concentrate around its expectation with high probability.

\begin{definition}[Locally mixed random vectors]\label{def:local indp vec}
    For a hypergraph $G$ with vertex set $[n]$, a vector $r \in \R^m$ is a \emph{$(G,p)$-mixed random vector over $S_1, \dots, S_m \subseteq [n]$ with respect to $\sigma \in \R^n$} if: 
    \begin{itemize}
    \item $|r_j| \le 1$ for all $j \in [m]$,
    \item $r_j$ is jointly independent of $r_{j'}$ for $j'$ where $dist(S_j, S_{j'}) > 2p$ as well as $\sigma_i$ for $i$ where $dist(\{i\},S_j) > 2p$, where,
    \[
        dist(U,V) = \min_{u\in U, v\in V} dist(u, v)\, .
    \]
    \end{itemize}
\end{definition}
\begin{remark}
The purpose of the $r$ vector is to enable reasoning about functions of more than one run of the algorithm, possibly with shared randomness between the runs.
When considering functions of the output of a single run of the algorithm, it will suffice to take $r_j \equiv 1$, which is trivially $(G,p)$-mixed over $S_1, \dots, S_m$ with respect to $\sigma$ for all $G$, $p$, $\sigma$, and $S_1, \dots, S_m$.
\end{remark}

\begin{restatable}[Concentration of local functions of spin configurations]{theorem}{concentrationlocal}
\label{lem:concentration-local-functions}
\ourparam~Let $\ell,\ m \in \N$.
Let $\sigma$ be the output of a generic $p$-local algorithm on a fixed hypergraph $G$.
Let $v_{j,i} \in [n]$ for $j\in[m]$ and $i\in[\ell]$.
Let $r = (r_1, \dots, r_m)$ be a $(G',p)$-mixed random vector (\autoref{def:local indp vec}) over $\{v_{1,i} \mid i\in [\ell]\},\dots,\{v_{m,i}\mid i\in[\ell]\}$ with respect to $\sigma$ for some hypergraph $G'$ for which $G$ is a subgraph of $G'$. Now, consider a sum,
\[
X = \sum_{j \in [m]} h(\sigma_{v_{j,1}},\dots,\sigma_{v_{j,\ell}})r_j\, ,
\]
where $|h|\le 1$. Suppose that each vertex $v$ occurs at most $C$ times among the different $v_{j,i}$.
Then, provided that $B_{G'}(i, 2p)$ has at most $n^A$ vertices in it for each $i \in [n]$, the following holds:
\[ 
\Pr_{\sigma, r} [|X - \lE_{\sigma, r}[X]| \ge \delta n] \le e^{-\Omega(\delta^2m/(C\ell n^A))}\, . 
\]
\end{restatable}

Next, we show concentration (over the randomness of the problem instances) of functions of hypergraphs which satisfy a bounded-differences inequality with respect to small changes in the hypergraphs. This lemma is itself an application of the strengthening of McDiarmid's inequality, stated in \autoref{lem:stronger-mcdiarmids}.

\begin{restatable}[Concentration of bounded local differences on coupled hypergraphs]{theorem}{concentrationhypergraph}\label{lem:concentration-random-coupled-hypergraphs}
\ourparam~Let $A$ and $a$ be the corresponding exponents from \autoref{thm:vanish-nghbd}.
Let $f$ be a function of two hypergraphs over $n$ vertices $V$, such that 
\[|f(G_1,G_2) - f(G_1',G_2')| \le r(n)\] for some $r$ whenever $(G_1, G_2)$ differs from $(G_1',G_2')$ by the addition and/or removal of a single hyperedge $e \in [n]^k$ from one or both graph and
\[\max_{i \in [n]} \max_{G \in \{G_1, G_2, G_1',G_2'\}}|B_{G}(i,p)| \le n^{A/2}.\]
Then
\[ \Pr_{G_1,G_2 \sim \cH_{n,k,d,t}}\!\!\left[\left|f(G_1,G_2) - \!\E_{G_1,G_2} f(G_1,G_2)\right| \ge \delta n\, r(n)\right] \le                 2\exp\left(\frac{-\delta^2 n}{4(2{-}t)d/k{+}2\delta/3}\right) + 2\exp(-n^{a/2}).  
\]
\end{restatable}
The theorem above is a generalization of the second part of \cite[Concentration Theorem]{farhi2020quantum}.

\paragraph{Organization of this section.}
In the rest of this section, we prove~\autoref{lem:concentration-local-functions} and~\autoref{lem:concentration-random-coupled-hypergraphs} in~\autoref{sec:concentration internal randomness} and~\autoref{sec:concentrstion randomness instance} respectively. Finally, we present the proof for~\autoref{thm:coup-hypergraph-qaoa-output-overlap} in~\autoref{sec:concentration main proof} and show several useful corollaries in~\autoref{sec:concentration corollaries}.

\subsection{Concentration over the internal randomness of the algorithm}\label{sec:concentration internal randomness}

In this subsection, we prove~\autoref{lem:concentration-local-functions} (restated below) which shows the concentration over the internal randomness of the algorithm. The proof is based on a Chernoff-style argument with a careful analysis on the combinatorial structure of the moment generating function.

\concentrationlocal*

\begin{proof}
We begin the proof by centering the variables $X_j$ so that we can crucially conclude that they contribute $0$ to the moment generating function when their expected value is taken.
For technical reasons (to make odd moments zero), we also introduce a global independent random sign $s \sim \{\pm 1\}$.
\[
    Z_j = s\left[X_j - \E_{\sigma, r}[X_j]\right] = s\left[h(\sigma_{v_{j, 1}},\dots,\sigma_{j_{j, \ell}})r_j - \E_{\sigma, r}[h(\sigma_{v_{j, 1}},\dots,\sigma_{j_{j, \ell}})r_j]\right] ,\ \forall j \in [m]\, . 
\]
Note as an immediate consequence that $\E_{\sigma, r}[Z_j] = 0$, $\forall j \in [m]$. Also, the goal now becomes showing $\Pr_{\sigma}[|\sum_jZ_j|\geq\delta m]<e^{-\Omega(\delta^2m/(C\ell n^A))}$.

We start with analyzing the moment generating function of $Z_j$ as follows,
\[
    \E_{\sigma, r}\left[\left(\sum_{j \in [m]}Z_j\right)^t\right] = \sum_{1 \leq j_1, \dots, j_t \leq m}\E_{\sigma, r}\left[\prod_{k \in \{j_1, \dots, j_t\}}Z_{k}\right]\, ,\ \forall t \in \N\, .
\]
For the $j$-th summand on the right to be non-zero, every factor $Z_k$ must be statistically dependent on at least one other term in $ \{Z_{j_1}, \dots, Z_{j_t}\}$.

We count the number of summands that can be non-zero.
If any $Z_{j_i}$ is independent from all other $Z_{j_{i'}}$, then the entire term is zero, since $\E[Z_{j_i}] = 0$.
Therefore if a summand is non-zero, then the interference graph between the factors $Z_{j_i}$ has no isolated vertices.
As a relaxation of this condition, the interference graph contains a forest with at most $t/2$ components as a subgraph.
Therefore, we can upper bound the number of non-zero terms by summing over all such forests the number of ways to assign factors $Z_{j_i}$ into vertices of that forest so that any two connected factors are not independent from each other.

We will split up the forests by the number of components they have, so first we count the number of forests over $t$ vertices with $w$ components for $w \le t/2$.
By a generalization of Cayley's formula \cite{TAKACS1990321}, there are $wt^{t-w-1}$ such forests if we assume that the first $w$ vertices are in different components.
Since we do not have the corresponding requirement on our factors, we may simply choose $w$ vertices arbitrarily to be in different components, multiplying by $\binom{t}{w}$ to find that there are at most $\binom{t}{w}wt^{t-w-1}$ ways to draw a forest with $w$ components over our $t$ indices $j_1, \dots j_t$.

We now count the number of ways to assign the $m$ factors into these forests so that any two factors connected by an edge are statistically dependent on each other.
A vertex may be arbitrarily selected from each component of the factor graph (say, choose the one with the lowest index in $[t]$) to be assigned any of the $m$ factors.
Once that vertex has been assigned, each of its neighboring vertices has at most $C\ell n^{A}$ choices of dependent factors, since each factor is dependent on at most $\ell n^{A}$ of the spins $\sigma_i$, and there are at most $C$ factors which are a function of each spin.
The same applies to all other vertices in the component.

Therefore, there are at most 
\begin{equation}
\label{eq:sum-over-components}
    \sum_{w = 1}^{t/2} \binom{t}{w}wt^{t-w-1}m^w(C\ell n^{A})^{t-w}
\end{equation} 
ways to generate a possible non-zero term.

Let $a_w = \binom{t}{w}wt^{t-w-1}m^w(C\ell n^{A})^{t-w}$, and we will find the index $w^*$ that maximizes $a_{w^*}$ by computing the ratio
\[ \frac{a_w}{a_{w-1}} = \frac{\frac{t!}{w!(t-w)!}wm}{\frac{t!}{(w-1)!(t-w+1)!}(w-1)tC\ell n^A}
= \frac{(t-w+1)m}{(w-1)tC\ell n^A}. \]
At this point, we do some casework.
In the case where $t \le m/(C\ell n^A)$, the above ratio is greater than $1$ whenever $w < 1 + t/2$, so $w^* = t/2$, recalling that we only have even moments because $Z_j$ contains a factor of a global sign $s \sim \{\pm 1\}$.
In the case where $t \ge m/(C\ell n^A)$, the ratio is less than $1$ whenever $w > ((1+1/t)m+C\ell n^A)/(m/t + C\ell n^A)$, which is implied by $w > 1 + m/(2C\ell n^A)$, so $w^* \le 1 + m/(2C\ell n^A)$

In the former case where $t \le m/(C\ell n^A)$ and $w^* = t/2$, we have
\[ a_{w^*} = \frac{1}{2}\binom{t}{t/2}t^{t/2}m^{t/2}(C\ell n^{A})^{t/2}
\le 2^{t}t^{t/2}m^{t/2}(C\ell n^{A})^{t/2}. \]
In the latter case where $t \ge m/(C\ell n^A)$ and $w^* > 1 + m/(2C\ell n^A)$,
\begin{align*}
a_{w^*}
&= \binom{t}{w^*}w^*t^{t-w^*-1}m^{w^*}(C\ell n^{A})^{t-w^*}
\\&\le \frac{e^{w^*}}{(w^*)^{w^*}}w^*t^{t-1}m^{w^*}(C\ell n^{A})^{t-w^*}
\\&\le e^{w^*}t^{t-1}m(C\ell n^{A})^{t-1}
\\&\le e^{t/2}t^{t-1}m(C\ell n^{A})^{t-1}
. \end{align*}
In either case,
\[ a_{w^*} \le 2^{t}t^{t/2}m^{t/2}(C\ell n^{A})^{t/2} + e^{t/2}t^{t-1}m(C\ell n^{A})^{t-1}, \]
so the number of non-zero terms in the expression of $\E \left[(\sum Z_j)^t\right]$ is at most $t/2$ times this bound on the maximum value of $a_{w^*}$, and, recalling that $|Z_j| \le 1$,
\[ \E_{\sigma}\left[\left(\sum_{j \in [m]}Z_j\right)^t\right] \le 2^{t-1}t^{t/2+1}m^{t/2}(C\ell n^{A})^{t/2} + \tfrac{1}{2}e^{t/2}t^{t}m(C\ell n^{A})^{t-1}. \]

We multiply both sides by $\theta^t$ and divide by $t!$ and sum over even $t$ (recalling that the definition of $Z_j$ contains a random global sign making all odd moments zero) to obtain a bound on the moment generating function
\[ \E_{\sigma}\left[e^{\theta \sum_{j \in [m]}Z_j}\right] \le
 \sum_{t \in \N \text{ even}} \frac{\theta^{t}}{t!}2^{t-1}t^{t/2+1}m^{t/2}(C\ell n^{A})^{t/2} + \sum_{t \in \N \text{ even}} \frac{\theta^{t}}{2 \cdot t!}e^{t/2}t^{t}m(C\ell n^{A})^{t-1}
 . \]
 We handle the two terms separately.
 For the first one, we reparameterize the index $t$ and then make use of Stirling's approximation:
 \begin{align*} \sum_{t \in \N \text{ even}} \frac{\theta^{t}}{t!}2^{t-1}t^{t/2+1}m^{t/2}(C\ell n^{A})^{t/2}
 &= \sum_{t \in \N} \frac{\theta^{2t}}{(2t)!}2^{2t-1}(2t)^{t+1}m^{t}(C\ell n^{A})^{t}
\\&= \sum_{t \in \N} \frac{t^{t+1}}{(2t)!}(8\theta^{2}mC\ell n^{A})^{t}
\\&\le \sum_{t \in \N} \frac{t^{t+1}}{\sqrt{2\pi}(2t)^{2t+1/2}e^{-2t}}(8\theta^{2}mC\ell n^{A})^{t}
\\&\le \sum_{t \in \N} \frac{t^{t+1}t^{t+1/2}}{\sqrt{2\pi}(2t)^{2t+1/2}e^{-t-1}t!}(8\theta^{2}mC\ell n^{A})^{t}
\\&= \frac{e}{2\sqrt{\pi}} \sum_{t \in \N} \frac{t(2e\theta^{2}mC\ell n^{A})^{t}}{t!}
\\&= \frac{\theta^{2}mC\ell n^{A}e^{2e\theta^{2}mC\ell n^{A} + 2}}{\sqrt{\pi}}.
 \end{align*}
 For the second term,
 \begin{align*}
 \sum_{t \in \N \text{ even}} \frac{\theta^{t}}{2 \cdot t!}e^{t/2}t^{t}m(C\ell n^{A})^{t-1}
 &\le  \sum_{t \in \N} \frac{\theta^{t}}{2 \cdot t!}e^{t/2}t^{t}m(C\ell n^{A})^{t-1}
 \\&\le \frac{m}{2\sqrt{2\pi}\, C\ell n^{A}} \sum_{t \in \N} \frac{(e^{3/2}\theta C\ell n^{A})^{t}}{\sqrt{t}}
 \\&\le \frac{m}{2\sqrt{2\pi}\, C\ell n^{A}} \sum_{t \in \N} (e^{3/2}\theta C\ell n^{A})^{t}
 \\&\le \frac{m}{2\sqrt{2\pi}\, C\ell n^{A}(1-e^{3/2}\theta C\ell n^{A})}.
 \end{align*}

 Putting these bounds together,
 \[ \E_{\sigma}\left[e^{\theta \sum_{j \in [m]}Z_j}\right] \le
 \frac{\theta^{2}mC\ell n^{A}e^{2e\theta^{2}mC\ell n^{A} + 2}}{\sqrt{\pi}} + \frac{m}{2\sqrt{2\pi}\, C\ell n^{A}(1-e^{3/2}\theta C\ell n^{A})}
 . \]
 By Markov's inequality,
 \[ \Pr_{\sigma}\left[\sum Z_j \ge \delta m\right] \le e^{-\theta \delta m}\E_{\sigma}\left[e^{\theta \sum_{j \in [m]}Z_j}\right] \]
for all $\theta$.
Choosing $\theta = \delta/(2e^{3/2} C\ell n^{A})$ with $0<\delta \le 1$, we get
 \[ \E_{\sigma}\left[e^{\theta \sum_{j \in [m]}Z_j}\right] \le
 \frac{\delta^2me^{\delta^2m/(2e^2C\ell n^A)}}{4e(C\ell n^{A})\sqrt{\pi}} + \frac{m}{\sqrt{2\pi}\, C\ell n^{A}}
 . \]
 Therefore,
 \[ \Pr_{\sigma}\left[\sum_j Z_j \ge \delta m\right] \le
 \frac{\delta^2me^{(-\sqrt{e} + 1)\delta^2m/(2e^2C\ell n^A)}}{4e(C\ell n^{A})\sqrt{\pi}} + \frac{m e^{-\delta^2 m/(2e^{3/2} C\ell n^{A})}}{\sqrt{2\pi}\, C\ell n^{A}}. \]
 The identical bound holds for $\Pr_{\sigma}\left[\sum_j Z_j \le -\delta m\right]$.
 
\end{proof}

\subsection{Concentration over randomness of problem instance}\label{sec:concentrstion randomness instance}
In this subsection, we prove~\autoref{lem:concentration-random-coupled-hypergraphs} which shows concentration of coupled hypergraphs with respect to the randomness of problem instances.
We start with stating and proving a special case of~\autoref{lem:concentration-random-coupled-hypergraphs} to illustrate the structure of the argument in a simpler setting.
\begin{lemma}[Concentration of Bounded Local Differences on Random Hypergraphs] \label{lem:concentration-random-hypergraphs}
For every $p \in \N$, let $A$ and $a$ be the corresponding exponents from \autoref{thm:vanish-nghbd}.
Let $f$ be a function of a hypergraph on $n$ vertices, such that $|f(G_1) - f(G_2)| \le r(n)$ whenever $\max_{i \in [n]} \max(|B_{G_1}(i,p)|, |B_{G_2}(i,p)|) \le n^{A/2}$ and $G_1$ differs from $G_2$ by the addition or removal of a single edge.
Then
\[ \Pr_{G \sim \cH_{n,k,d}}\left[\left|f(G) - \E_G f(G)\right| \ge \delta n\, r(n)\right] \le 2\exp\left(\frac{-\delta^2 n}{4d/k +2\delta/3}\right) + \exp(-n^{a/2}).  \]
\end{lemma}
\begin{proof}
The proof follows by essentially the same arguments used in the second part of \cite[Concentration Theorem]{farhi2020quantum}, although we have to derive a strengthening of McDiarmid's inequality (\autoref{lem:stronger-mcdiarmids}), deferred to \autoref{sec:stronger-mcdiarmids}.

Let $K_n$ be the set of hypergraphs over $n$ vertices $V$ with small neighborhoods
\[ K_n = \{G \in \cH_{n,k} \mid \max_{v \in V} |B_G(v,p)| \le n^{A/2} \}. \]
Let $\rho(G,G')$ be equal to $|E \mathbin{\triangle} E'|$ where $\triangle$ is the symmetric difference, for $G=(V,E)$ and $G' = (V, E')$.
Then let 
\[ g(G) = \min_{G' \in K_n} f(G') + \rho(G, G')\,r(n). \]
Now $g$ has the property that $|g(G_1) - g(G_2)| \le r(n)$ whenever $G_1$ differs from $G_2$ by the addition or removal of a single edge.

$G \sim \cH_{n,k,d}$ may be viewed as the agglomeration of $n^k$ different independent random variables, one for each possible hyperedge $e$, each random variable denoting the multiplicity of that hyperedge, distributed as $\poisson(dn/n^k/k)$ 
, since the sum of independent Poisson random variables is itself another Poisson random variable.
And each of these variables is highly biased, being equal to $0$ with probability $\exp(-dn/n^k/k) \ge 1-dn/n^k/k$.
Therefore, applying \autoref{lem:stronger-mcdiarmids} on $g$ as a function of $n^k$ independent variables, simplifying, and applying the bound $2-dn/n^k/k \le 2$,
\[ \Pr_{G \sim \cH_{n,k,d}}[|g(G) - \E_G g(G)| \ge \delta n\,r(n)] \le 2\exp\left(\frac{-\delta^2 n}{4d/k +2\delta/3}\right).  \]

Finally, $f(G) = g(G)$ whenever $G \in K_n$, and by \autoref{thm:vanish-nghbd},
\[ \Pr_{G \sim \cH_{n,k,d}}[G \not\in K_n] \le e^{-n^{a/2}}.\]
Therefore, by a union bound,
\[ \Pr_{G \sim \cH_{n,k,d}}[|f(G) - \E_G f(G)| \ge \delta n\,r(n)] \le 2\exp\left(\frac{-\delta^2 n}{4d/k +2\delta/3}\right) + \exp(-n^{a/2}).  \]
\end{proof}

We now prove~\autoref{lem:concentration-random-coupled-hypergraphs} (restated below) which shows that the expected values for certain functions on coupled hypergraphs also concentrate. More specifically, we will assert that all sufficiently \emph{local} functions of pairs of hypergraphs will concentrate very heavily around the expected value of the function.

\concentrationhypergraph*

\begin{proof}[Proof]
We mostly follow the proof of \autoref{lem:concentration-random-hypergraphs}.

Let $\rho(G_1,G_2,G_1',G_2') = |(E_1 \mathbin{\triangle} E_1') \cup (E_2 \mathbin{\triangle} E_2')|$, where $G_1 = (V,E_1)$, $G_2 = (V,E_2)$, $G_1' = (V,E_1')$, and $G_2' = (V,E_2')$.
Then let 
\[ g(G_1,G_2) = \min_{G_1' \in K_n,G_2' \in K_n} f(G_1',G_2') + \rho(G_1,G_2,G_1',G_2')\,r(n), \]
so that $|g(G_1,G_2) - g(G_1',G_2)| \le r(n)$ whenever $(G_1,G_2)$ differs from $(G_1',G_2')$ by the addition or removal of a single hyperedge in $[n]^k$ %
 from one or both graphs.

To apply \autoref{lem:stronger-mcdiarmids}, we consider $g$ as a function of $n^k$ variables, one for each possible hyperedge.
The set of possible values for each variable is $\{\emptyset, \{G_1\}, \{G_2\}, \{G_1,G_2\}\}$, specifying which of the hypergraphs have that edge.
In the coupled random hypergraph model according to \autoref{def:coupled-interpolation}, where here $G_1 = (V, E \cup E_1)$ and $G_2 = (V, E \cup E_2)$ each edge's multiplicity in $E$ is given by a $\poisson(tdn/n^k/k)$ distribution, and its multiplicities in $E_1$ and $E_2$ are given by $\poisson((1-t)dn/n^k/k)$ distributions, for a total probability of $(1 -\exp(-tdn/n^k/k))(1 -\exp(-(1-t)dn/n^k/k))^2 \le (2-t)dn/n^k/k$ that this edge is in either hypergraph.
Therefore, applying \autoref{lem:stronger-mcdiarmids},
\[ \Pr_{G_1,G_2 \sim \cH_{n,k,d,t}}\left[\left|g(G_1,G_2) - \E_{G_1,G_2} g(G_1,G_2)\right| \ge \delta n\,r(n)\right] \le 2\exp\left(\frac{-\delta^2 n}{4(2-t)d/k +2\delta/3}\right).  \]

Finally, $f(G_1,G_2) = g(G_1,G_2)$ whenever $(G_1,G_2) \in K_n \times K_n$.
Since the marginal distribution of $G_1$ in $G_1,G_2 \sim \cH_{n,k,d,t}$ is the same as the distribution of $G_1 \sim \cH_{n,k,d}$, and the same holds for $G_2$, by a union bound,
\[ \Pr_{G_1,G_2 \sim \cH_{n,k,d,t}}[(G_1,G_2) \not\in K_n \times K_n] \le 
\Pr_{G_1 \sim \cH_{n,k,d}}[G_1 \not\in K_n] + \Pr_{G_2 \sim \cH_{n,k,d}}[G_2 \not\in K_n] 
\le 2e^{-n^{a/2}}\, .\]
Therefore, by another union bound,
\[ \Pr_{G_1,G_2 \sim \cH_{n,k,d,t}}\!\!\left[\left|f(G_1,G_2) - \!\E_{G_1,G_2} f(G_1,G_2)\right| \ge \delta n\,r(n)\!\right] \le 2\exp\!\left(\!\frac{-\delta^2 n}{4(2{-}t)d/k{+}2\delta/3}\!\right)\! {+} 2\exp(-n^{a/2}).  \]

\end{proof}

\subsection{Proof of Theorem~\ref{thm:coup-hypergraph-qaoa-output-overlap}}\label{sec:concentration main proof}

Finally, we are ready to prove~\autoref{thm:coup-hypergraph-qaoa-output-overlap} (restated below) using~\autoref{lem:concentration-local-functions} and~\autoref{lem:concentration-random-coupled-hypergraphs}.

\localoutput*

\begin{proof}[Proof]
By \autoref{lem:concentration-local-functions}, taking $G = G_2$, $m = n$, $C = \ell = 1$, $v_{j,1} = j$, $h(s) = s$, and $r_j = u_j$,
\[\Pr_{\sigma_2} \left[\left|\iprod{\sigma_2,u} - \E_{\sigma_2} \iprod{\sigma_2,u}\right| \ge 2\delta n\right] \le \exp(-\Omega(\delta^2n^{1 - A}))\, ,\]
for every bounded vector $u \in [-1,1]^{n}$, with the probability and expectation over the randomness of the algorithm, where $A$ corresponds to the exponent in \autoref{thm:vanish-nghbd}.

Again by \autoref{lem:concentration-local-functions}, taking $G = G_1$, $m = n$, $C = \ell = 1$, $v_{j,1} = j$, $h(s) = s$, $r_j = (\sigma_2)_j$, and $G' = G_1 \cup G_2$, using the fact that $\iprod{\sigma_1, \sigma_2} = \sum_{i \in [n]} (\sigma_1)_i(\sigma_2)_i$,
\begin{equation}
\label{eq:time-for-mcdiarmids}
    \Pr_{\sigma_1, \sigma_2} \left[\left|\iprod{\sigma_1, \sigma_2} - \E_{\sigma_1, \sigma_2} \iprod{\sigma_1, \sigma_2}\right| \ge 2\delta n\right] \le \exp(-\Omega(\delta^2n^{1 - A})) + \exp(-\Omega(\delta^2n^{1-  A'}))\, ,
\end{equation}
with the probability and expectation over the randomness of the algorithm. In the above, we utilize the fact that $G' \sim \mathcal{H}_{d(2 - t), k, n}$. To bound $A'$, one bounds the $2p$-neighborhood of $G'$ using \autoref{thm:vanish-nghbd} invoked with degree $d(2 - t)\leq 2d$.

We argue now that $f(G_1,G_2) := \E_{\sigma_1, \sigma_2} \iprod{\sigma_1, \sigma_2}$ satisfies a bounded-differences inequality, so that $|f(G_1,G_2) - f(G_1',G_2')| \le 2kn^A$ whenever $\max_{i \in [n]}\max_{G\in\{G_1,G_2,G_1',G_2'\}} |B_G(i,p)| \le n^{A}$ and $|(E(G_1)\mathbin{\triangle}E(G_1')) \cup (E(G_2)\mathbin{\triangle}E(G_2'))| \le 1$.
This is because by the definition of $p$-local, the only coordinates of $\sigma_1$ or $\sigma_2$ that can change in marginal distribution when an edge is changed are those that are within a distance of $p$ from any of the $k$ endpoints of the changed hyperedge, in $G_1$ and $G_1'$, or $G_2$ and $G_2'$ respectively.

Therefore by \autoref{lem:concentration-random-coupled-hypergraphs},
\[ \Pr_{G_1, G_2} \left[ \left|\E_{\sigma_1, \sigma_2} \iprod{\sigma_1, \sigma_2} - \!\E_{G_1,G_2} \E_{\sigma_1, \sigma_2} \iprod{\sigma_1, \sigma_2}\right| \ge 2k\delta' n^{A}n\right] \le 2\exp\left(\frac{-(\delta')^2 n}{4(2-t)d/k +2\delta'/3}\right) + 2\exp(-n^{a/2}), \]
where $a$ is from the statement of \autoref{thm:vanish-nghbd}.

By a union bound and triangle inequality with \eqref{eq:time-for-mcdiarmids} then, and taking $\delta' = \delta/(kn^A)$,
\begin{align*}
\Pr_{G_1, G_2,\sigma_1,\sigma_2} &\left[ \left|\iprod{\sigma_1, \sigma_2} - \E_{G_1,G_2} \E_{\sigma_1, \sigma_2} \iprod{\sigma_1, \sigma_2}\right| \ge 4\delta n\right] \le 
\\ &2\exp\left(\frac{-\delta^2n}{4(2-t)kn^{2A}d +2kn^A\delta/3}\right) + 2\exp(-n^{a/2}) + \exp(-\Omega(\delta^2n^{1-A})) + \\
&\exp(-\Omega(\delta^2n^{1-A'}))\, . 
\end{align*}
\end{proof}

\subsection{Corollaries of concentration analysis}\label{sec:concentration corollaries}
As a corollary to \autoref{thm:coup-hypergraph-qaoa-output-overlap}, we can inherit the concentrated overlap of two \emph{independent} hypergraphs $G_1, G_2 \sim \mathcal{H}_{d, k, n}$. This is essential to reasoning about the fact that the overlap at $t = 0$ between the solutions is desirably small with high probability.

\begin{corollary}[Generic local algorithm output overlap on independent hypergraphs]\label{thm:ind-hypergraph-qaoa-output-overlap}
    Suppose $p$ is as in \autoref{thm:vanish-nghbd}.
    When two random hypergraphs $G_1,\ G_2$ are sampled $i.i.d.$ from $\mathcal{H}_{d, k, n}$ and a generic $p$-local algorithm is run on both hypergraphs, the overlap between the respective solutions $\sigma_1$ and $\sigma_2$ concentrates. That is, $\exists\ \delta' > 0$ and $\gamma' > 0$, such that,
    \begin{equation}\label{eq:ind-hypergraph-qaoa-output-overlap}
        \Pr_{G_1, G_2,\sigma_1,\sigma_2}[|\langle \sigma_1, \sigma_2 \rangle - \lE[\langle \sigma_1, \sigma_2\rangle]| \geq \delta'\cdot n] \leq 2e^{-\delta' n^{\gamma'}}\, ,
    \end{equation}
    where the probability is over both the random sample of hypergraphs and the randomness of the algorithm.
\end{corollary}
\begin{proof}
    By \autoref{thm:coup-hypergraph-qaoa-output-overlap} at $t = 0$ with $G_1$ independent of $G_2$.
\end{proof}

We now assert the concentration of hamming weight of the solutions output by a generic $p$-local algorithm over the randomness of the algorithm. This result is equivalent to the second part of \cite[Concentration Theorem]{farhi2020quantum} and \cite[Lemma 3.2]{chen2019suboptimality}. In our case, it is a corollary of \autoref{lem:concentration-local-functions}.

\begin{corollary}[Concentration of hamming weight of generic local algorithm output]\label{cor:conc-p-local-hamming-wt}
    Given a $d$-sparse $k$-uniform hypergraph $G \sim \mathcal{H}_{d, k, n}$, suppose that,
    \begin{equation*}
        \Pr_{G}[\max_i B_G(v_i, 2p) \geq n^A] \leq e^{-n^a}\, ,
    \end{equation*}
    and
    \begin{equation*}
        \Pr_{G}[\max_i B_G(v_i, p) \geq n^{\frac{A}{2}}] \leq e^{-n^{\frac{a}{2}}}\, ,
    \end{equation*}
    Let $\sigma$ be the output of a generic $p$-local algorithm. Then, $\exists$ $\gamma > 0$ such that, $\forall\ \delta > 0$,
    \begin{equation}\label{eq:conc-hamming-wt-qaoa}
        \Pr_{G, \sigma}\left[\left||\sigma| - \E_{G, \sigma}[|\sigma|]\right| \geq \delta n\right] \leq e^{-\delta^2 n^{\gamma}}\, ,
    \end{equation}
    where %
    $|\sigma| = \sum_{i=1}^n \frac{1}{2}(1 - \sigma_i)$ is the number of $-1$s in $\sigma$.
\end{corollary}
\begin{proof}
    Instantiate \autoref{lem:concentration-local-functions} by asserting that the following holds for the graph $G$,
    \begin{equation}\label{eq:conc-hamming-wt-single-instance}
        \Pr_{\sigma}\left[\left| \sum_{i=1}^n \sigma_i - \E_{\sigma}[\sigma]\right| \ge 2\delta n\right] \leq \exp(-\Omega(\delta^2n^{1-A}))\, .  
    \end{equation}
    with $\ell = 1$, $C = 1$, $m = n$, $v_{j,1} = j$, $r_j = 1$, and $h(s) = s$. We now invoke \autoref{lem:concentration-random-hypergraphs} with $f = \E_{\sigma}[|\sigma|]$. Now, notice that the definition of generic $p$-local algorithms (\autoref{def:local}) implies that on adding or removing a hyperedge $e$, $f$ can change by no more than $2kn^A$, since changes are restricted to the $p$-neighborhood. Therefore, provided $G$ is a graph with $\max_i B(v_i, 2p) \leq n^A$, we obtain,
    \begin{align*}
        \Pr_{G, \sigma}&\left[\left|\E_{\sigma}[|\sigma|] - \E_{G}\E_{\sigma}[|\sigma|]\right| \geq 2k\delta'n^An\right] \leq 
        2\exp\left(\frac{-\delta'^2 n}{4d/k +2\delta'/3}\right) + \exp(-n^{a/2})\, .
    \end{align*}
    By setting $\delta' = \delta/kn^A$, and taking a triangle inequality over \autoref{eq:conc-hamming-wt-single-instance} and a union bound,
    \begin{align*}
        \Pr_{G, \sigma}&\left[\left||\sigma| - \E_{G, \sigma}[|\sigma|]\right| \geq 2\delta n\right] \leq 2\exp\left(\frac{-\delta^2 n^{1-2A}}{4dk +2\delta k/(3n^A)}\right) + \exp(-n^{a/2}) + \exp(-\Omega(\delta^2n^{1-A}))\, .
    \end{align*}
\end{proof}

The last corollary we obtain from \autoref{lem:concentration-local-functions} is that the objective function value, that is, the energy corresponding to the spin configuration output by a generic $p$-local algorithm, also concentrates heavily around the expected value. This is equivalent to \cite[Lemma 3.1]{chen2019suboptimality} with substantially stronger concentration.

\begin{corollary}[Concentration of energy of generic local algorithm output]\label{cor:conc-p-local-energy}
    Given a $d$-sparse $k$-uniform hypergraph $G \sim \mathcal{H}_{d, k, n}$, suppose that %
    \begin{equation*}
        \Pr_{G}[\max_i B_G(v_i, 2p) \geq n^A] \leq e^{-n^a}\, ,
    \end{equation*}
    and
    \begin{equation*}
        \Pr_{G}[\max_i B_G(v_i, p) \geq n^{\frac{A}{2}}] \leq e^{-n^{\frac{a}{2}}}\, .
    \end{equation*}
    Let $\sigma$ be the output of a generic $p$-local algorithm. Then, $\exists \gamma > 0$, such that, $\forall \delta > 0$,
    \[
        \Pr_{G, \sigma}\left[\left|H^G_{d, k, n}(\sigma) - \E_{G, \sigma}[H^G_{d, k, n}(\sigma)]\right| \geq \delta n\right] \leq e^{-\delta^2 n^{\gamma}}\, .
    \]
\end{corollary}
\begin{proof}
    For sufficiently large $n$, the degree of the vertices in $G$ are distributed as $\poisson(d)$. Therefore, by a standard Chernoff bound for the Poisson distribution~\cite[Theorem 1]{canonne2017short}, %
    \[
        \Pr_{\poisson(d)}\left[\text{deg}(v) \geq d + n^{\frac{a}{2}}\right] \leq \exp\left(\frac{-n^a}{2\left(d + n^{a/2}\right)}\right)\, .
    \]
    Applying a union bound to the above for every vertex $v \in V(G)$ yields that the degree of every vertex can be upper bounded by $d + n^{\frac{a}{2}} \leq 2n^{\frac{a}{2}}$, for sufficiently large $n$. \\
    An instantiation of \autoref{lem:concentration-local-functions} on a graph $G$ chosen so that the degree of every vertex is not more than $2n^{\frac{a}{2}}$ implies that,
    \begin{equation}\label{eq:conc-energy-random-instance}
        \Pr_{\sigma}\left[\left|H^G_{d, k, n}(\sigma) - \E_{\sigma}[H^G_{d, k, n}(\sigma)]\right| \geq 2 \delta n\right] \leq \exp(-\Omega\left(\frac{\delta^2dn^{1 - A - a/2}}{2k}\right))\, ,
    \end{equation}
    where $m = dn$, $\ell = k$, $C = 2n^{\frac{a}{2}}$ and $h_j(s_{j, 1},\dots s_{j, k}) = -\prod_{i=1}^{\ell}s_{j, i}$, where $s_{j, i}$ denotes the spin of the $i$-th vertex of the $j$-th hyperedge. \\
    We instantiate \autoref{lem:concentration-random-hypergraphs} with $f = \E_{\sigma}[H_{d, k, n}]$.  Note that if we change one hyperedge $e$ of $G$, then because of the fact that a generic $p$-local algorithm will only affect the joint distribution in a $p$-neighborhood around every vertex, we notice that $f$ cannot change by more than $2Ckn^A = 4kn^{A + \frac{a}{2}}$ when $G$ satisfies the bounded neighborhood proposition in the hypothesis. Therefore,
    \begin{align*}
        \Pr_{G, \sigma}&\left[\left|\E_{\sigma}[H^G_{d, k, n}(\sigma)] - \E_{G}\E_{\sigma}[H^G_{d, k, n}]\right| \geq 2Ck\delta'n^An\right] \leq \\
        &2\exp\left(\frac{-\delta'^2 n}{4d/k +2\delta'/3}\right) + \exp(-n^{a/2}) + \exp\left(\frac{-n^a}{2\left(d + n^{a/2}\right)}\right)\, . 
    \end{align*}
    By a union bound and triangle inequality with \autoref{eq:conc-energy-random-instance}, followed by taking $\delta' = \delta/(Ckn^A)$,
    \begin{align*}
        &\Pr_{G, \sigma}\left[\left|H^G_{d, k, n}(\sigma) - \E_{G, \sigma}[H^G_{d, k, n}(\sigma)]\right| \geq 4\delta n \right] \leq \\
        &2\exp\!\left(\frac{-\delta^2n^{1-2A-a}}{4dk {+} \frac{2}{3}k\delta{/}n^{A{+}a/2}} \right)\! + \exp(-n^{a/2}) + \exp\!\left(\frac{-n^a {+} \log(n)}{2\left(d {+} n^{a/2}\right)}\right)\! + \exp\!\left(\frac{-\Omega(\delta^2dn^{1 - A - a/2})}{2k}\right)\, .
    \end{align*}
    Note that we can guarantee that $2a + A < 1$ with an appropriate choice of $a, A > 0$, which is made possible by choosing an appropriate value of $\tau > 0$ in \autoref{param}. For an explicit characterization, refer to \autoref{sec:proof-vanishing-nhbhd} and \cite[Equation (75), Neighborhood Size Theorem]{farhi2020quantum}. 
\end{proof}

\section{Proofs of main theorems} \label{sec:proofs-main-thms}
We conclude by putting the previous results together to prove our main results. As stated prior, we work in the setting of \autoref{param}.

\subsection{Proof of Theorem~\ref{thm:obstruction-1}}

We establish that, as a function of $t$, the expected overlap between solutions output by runs of a generic local algorithm on $t$-coupled instances $G_1$ and $G_2$ with $t$-shared randomness is continuous.

\begin{lemma}[Continuity of expected overlap]\label{lem:cont-overlap}
    If $G_1, G_2 \sim \mathcal{H}_{d, k, n, t}$, and $\sigma_1, \sigma_2$ are the random outputs of $t$-shared randomness runs of a generic $p$-local algorithm on each hypergraph, 
    \[
        \lE_{G_1, G_2, \sigma_1, \sigma_2}\left[R_{n, t}(\sigma_1, \sigma_2)\right] \text{ is a continuous function of $t$ for all $t \in (0,1)$}\, .
    \]
\end{lemma}
\begin{proof}
We express using linearity of expectation
\[ \E_{G_1,G_2,\sigma_1,\sigma_2} \left[R_{n,t}(\sigma_1,\sigma_2)\right]
= \frac{1}{n}\sum_{i \in [n]} \E_{_{G_1,G_2,\sigma_1,\sigma_2}} \left[(\sigma_1)_i(\sigma_2)_i\right]
= \E_{G_1,G_2,\sigma_1,\sigma_2} \left[(\sigma_1)_1(\sigma_2)_1\right]
. \]
Let $s(H_1,H_2)$ be the expectation value $\E_{\sigma_1,\sigma_2} (\sigma_1)_1(\sigma_2)_1$ when $\sigma_1$ and $\sigma_2$ are the results of a pair of runs of the $p$-local algorithm with $t$-shared randomness on $H_1$ and $H_2$ respectively.
Then 
\[ \E_{G_1,G_2,\sigma_1,\sigma_2} \left[R_{n,t}(\sigma_1,\sigma_2)\right] 
= \sum_{H_1, H_2} \Pr_{G_1,G_2 \sim \cH_{n,k,d,t}\hspace{-1em}}[H_1 = G_1 \wedge H_2 = G_2]\;  s(H_1,H_2). \]
Since the probability of each combination of $G_1$ and $G_2$ being sampled is a continuous function of $t$, so then is $\E_{G_1,G_2,\sigma_1,\sigma_2} \left[R_{n,t}(\sigma_1,\sigma_2)\right]$.
\end{proof}

Note that at $t = 0$, the hypergraphs $G_1, G_2 \sim \mathcal{H}_{d, k, n, t}$ are independent of each other. We will show that there is very little overlap between the outputs $\sigma_1, \sigma_2$ of a generic local algorithm on each instance, if both solutions $\sigma_1$ and $\sigma_2$ are better approximations to the optimum value by a factor of $(1 - \frac{\eta_0}{P(k)})$. This follows directly from a combination of \cite[Lemma 3.3]{chen2019suboptimality} and \autoref{lem:concentration-local-functions}. We begin by restating \cite[Lemma 3.3]{chen2019suboptimality}.

\begin{lemma}[Hamming weight of near-optimal solutions, {\cite[Lemma 3.3]{chen2019suboptimality}}]\label{lem:overlap-t-0-cgpr}
    Given two independent and uniformly random hypergraphs $G_1, G_2 \sim \mathcal{H}_{d, k, n}$, and solutions $\sigma_1$ and $\sigma_2$, such that, 
    \[
        \frac{H^{G_i}(\sigma_i)}{n} > M(k, d)\left(1 - \frac{\eta_0}{P(k)}\right)\, ,
    \]
    for $i \in \{1, 2\}$ and any $\eta_0 > 0$. Then, $|\sigma_i| < d^{-\frac{1}{2k}}(4\eta_0)^{\frac{1}{k}}$ with probability $\geq 1 - O(e^{-n})$.
\end{lemma}

The above lemma allows us to bound the overlap between two instances when presented in conjunction with a concentration argument about the hamming weight of instances produced by a generic local algorithm.
\begin{lemma}[Small overlap at $t = 0$]\label{lem:overlap-t-0}
    For any two hypergraphs $G_1, G_2 \sim \mathcal{H}_{d, k, n, 0}$, let $\sigma_1, \sigma_2$ be the random outputs of two $0$-shared randomness runs of a generic $p$-local algorithm on the instances. If $\sigma_1$ and $\sigma_2$ satisfy the optimality of \autoref{lem:overlap-t-0-cgpr}, then
    \[
        R_{n, 0}(\sigma_1, \sigma_2) \leq d^{-\frac{1}{k}}\, ,
    \]
    with probability $\geq 1 - e^{-O(n^{\gamma})}$.
\end{lemma}
\begin{proof}
    The proof is essentially the same as that in \cite[Section 3.3]{chen2019suboptimality} with the exception that we provide stronger concentration via \autoref{lem:concentration-local-functions}. \\
    By \autoref{lem:overlap-t-0-cgpr}, we have that, 
    \[
        \left|\frac{1}{n}\sum_{i=1}^{n}\sigma_{j_i}\right| < d^{-1/2k}(4\eta_0)^{1/k}\ ,\ \text{for } j = 1, 2\, ,
    \]
    with probability $\geq 1 - O(e^{-n})$. Choose $\eta_0 = \frac{1}{8}$, per \cite{chen2019suboptimality}. This implies,
    \[
        |\sigma_i| \leq d^{-1/2k}2^{-1/k} \le d^{-1/2k}\, .
    \]
    Furthermore, because the output of generic $p$-local algorithms concentrate, 
    \[
        \E_{G_i, \sigma_i}[|\sigma_i|] \leq d^{-1/2k}\, ,
    \]
    with probability $\geq 1 - e^{-O(n^{\gamma})}$ (as implied by \autoref{cor:conc-p-local-hamming-wt}). Now, note that at $t = 0$, the generic $p$-local algorithm is equally likely to generate $\sigma_i$, $i \in \{1, 2\}$ with a given hamming-weight (in the $\pm 1$ basis). Consequently, by independence, the largest expected overlap is the square of the maximum possible hamming weight. This yields,
    \[
        R_{n, 0}(\sigma_1, \sigma_2) \leq (d^{-1/2k})^2 = d^{-1/k}\, .
    \]
    To establish that this is less than $a$, simply choose $d \geq \max(d_0, a^{-k})$.
\end{proof}

We now assert that the output of two generic local algorithm runs with fully shared randomness on fully coupled (and therefore identical) hypergraphs $G_1$ and $G_2$ has overlap $1$.

\begin{lemma}[Large overlap at $t = 1$]\label{lem:overlap-t-1}
    For any two hypergraphs $G_1, G_2 \sim \mathcal{H}_{d, k, n, 1}$, given that $\sigma_1, \sigma_2$ are the outputs of two $1$-shared randomness runs of a $p$-local algorithm on the instances,
    \[
        R_{n, 1}(\sigma_1, \sigma_2) = 1\, .
    \]
\end{lemma}
\begin{proof}
    By \autoref{def:coupled-interpolation}, at $t=1$, the two hypergraphs $G_1$ and $G_2$ sampled from $\mathcal{H}_{d, k, n, 1}$ are in fact the same.
    Therefore every vertex's neighborhood is the same in $G_1$ and $G_2$, and so with $t=1$ in \autoref{def:coupled-runs}, $L^+ = V$ and so $\sigma = \sigma_1 = \sigma_2$ and the two outputs are identical.
\end{proof}
The arguments above immediately yield a contradiction to the coupled OGP stated in \autoref{thm:ogp-spin-glasses-coupled} via the intermediate value theorem, which asserts that given the endpoints $R_0(\sigma_1, \sigma_2) < a$ and $R_1(\sigma_1,\sigma_2) = 1$, $\exists\ t' \in (0, 1)$, such that $R_{t'}(\sigma_1, \sigma_2) \in [a, b]$ (with high probability) because of the continuity of $R_t$.

\subsection{Proof of Theorem~\ref{thm:qaoa-conc-log-depth}}

We now sketch a complete proof of \autoref{thm:qaoa-conc-log-depth}.

\begin{proof}
    The input problem $\Psi$ is a $(k, d)\text{-}\mathsf{CSP}(f)$ over $n$ variables chose as in \autoref{def:k-d-csp}. This allows one to encode the problem instance $\Psi$ as a random hypergraph $G \sim \mathcal{H}_{d, k, n}$ by choosing $h_j = f(x_{j_1},\dots,x_{j_k})$ to be the local cost function for the $j$-th clause of $\Psi$. The collection of variables in each clause corresponds to the variables in a hyperedge in $G_{\Psi}$. Let the number of clauses be $|E|$. By definition, $|E| \sim \poisson(dn/k)$. Now, we set the hamiltonian as,
    \[
        H^{G_{\Psi}}_{d, k, n} = \sum_{j=1}^{|E|}f(x_{j_1},\dots,x_{j_k})\, .
    \]
    By \autoref{thm:vanish-nghbd}, with the choice of $p$ stated in the hypothesis, the $2p$ neighborhood will be no larger than $n^A$ with probability $\geq 1 - e^{-n^a}$. The proof then concludes by invoking \autoref{cor:conc-p-local-energy} with $H^{G_{\Psi}}_{d, k, n}$, conditioned on the fact that the largest $2p$ neighborhood of $H^{G_{\Psi}}_{d, k, n}$ has size no more than $n^A$. 
\end{proof}

\subsection{Proof of Theorem~\ref{thm:obstruction-2}}

We now sketch a complete proof of \autoref{thm:obstruction-2}.

\begin{proof}
    Note that, once again, $\Psi$ can be encoded into its representative hypergraph $G_{\Psi}$ exactly as in the proof of \autoref{thm:qaoa-conc-log-depth}. Furthermore, by the hypothesis, the underlying problem $(k, d)\text{-}\mathsf{CSP}(f)$ satisfies a coupled OGP which can be interpreted as follows: For any $t \in [0, 1]$, choose $\poisson(tdn/k)$ clauses independently and uniformly at random and create an instance $\Psi$ from them. Then, independently sample $\poisson((1 - t)dn/k)$ clauses uniformly at random twice and create instances $\Psi_1$ and $\Psi_2$ from them. Let the final instances be $\Psi = \Psi_0 \cup \Psi_1$ and $\Psi' = \Psi_0 \cup \Psi_2$. This is exactly equivalent to the interpolation \autoref{def:coupled-interpolation} over the representative hypergraphs $G_{\Psi}$ and $G_{\Psi'}$. Note that the hypothesis implies we can assume these hypergraphs have $2p$-neighborhoods of size no more than $n^A$ with high probability (\autoref{thm:vanish-nghbd}). \\
    \paragraph{Concentration of objective value.} Set the representative hamiltonians $H^{G_\Psi}_{d, k, n}$ and $H^{G_{\Psi'}}_{d, k, n}$ to be the sum of $f$ acting on their respective clauses $\Psi_0 \cup \Psi_1$ and $\Psi_0 \cup \Psi_2$. Then, by \autoref{thm:qaoa-conc-log-depth}, both these functions will be concentrated with probability $\geq 1 - e^{-O(n^{\gamma})}$.
    \paragraph{Concentration of Hamming Weight.} By invoking \autoref{cor:conc-p-local-hamming-wt} for each instance, we conclude that the hamming weight is concentrated with near certainty.
    \paragraph{Concentration of Overlap of Output.} This follows directly by applying \autoref{thm:coup-hypergraph-qaoa-output-overlap} to $G_{\Psi}$ and $G_{\Psi'}$. \\
    
    Given these properties, the argument in the proof for \autoref{thm:obstruction-1} immediately implies the desired obstruction.
\end{proof}

\section{Strengthened McDiarmid's Inequality for biased distributions}\label{sec:stronger-mcdiarmids}
We introduce the notions of a martingale and a Doob martingale, followed by a concentration result for bounded difference martingales in Fan et al.~\cite{fan2012hoeffding}  , and conclude with a proof of a strengthened version of McDiarmid's inequality for highly-biased distributions.

\subsection{Martingales \& concentration}
\begin{definition}[Martingales]
A \emph{martingale} with respect to random variables $X_1,\dots, X_n$ is given by a sequence of random variables $Z_1, \dots, Z_n$ such that the conditional expectation of each $Z_i$ conditioned on all previous data points $X_{< i}$ in the sequence is,
\[
    \E[Z_i \mid X_1, \dots, X_{i-1}] = Z_{i-1}\, .
\]
Furthermore, all expectations are bounded: $\E[|Z_i|] < \infty$. If $X_i = Z_i$, then $Z_1,\dots,Z_n$ is said to be a martingale with respect to itself.   
\end{definition}
A common martingale is the so-called \emph{Doob martingale} where we set the random variables $Z_i$ to be the averages of some bounded function $f$ acting on $X_1,\dots, X_n$ conditioned on observations up to $X_{< i}$.
\begin{definition}[Doob martingales]
    The \emph{Doob martingale} of a function $f$ of random variables $X_1, \dots, X_n$ is the martingale sequence given by,
    \[
        Z_i = \E [f(X_1, \dots, X_n) \mid X_1, \dots, X_i]\, ,
    \]
    so long as $\E [|f(X_1, \dots, X_n)|] < \infty$. 
\end{definition}
We introduce some more technical definitions below.
\begin{definition}[Martingale difference]
    The \emph{martingale difference sequence} of a martingale $Z_1, \dots, Z_n$ is given by the sequence $Y_i = Z_i - Z_{i-1}$, noting that $\E[Y_i] = 0$ by definition.    
\end{definition}
\begin{definition}[Quadratic characteristic sequence]
    The \emph{quadratic characteristic sequence} of a martingale $Z_1, \dots, Z_n$ with martingale difference sequence $Y_2, \dots, Y_n$ consists of the values
    \[ 
        \iprod{Z}_{i} = \sum_{j \in [i]} \E [ Y_j^2 \mid Z_1, \dots, Z_{j-1}]. 
    \]    
\end{definition}

To prove the version of McDiarmid's inequality with a highly-biased distribution, we will need a result that generalizes the concentration bounds of Azuma and Bernstein, using the quadratic characteristic sequence as the martingale analogue of the variance. The generalization is given by Fan et al.~\cite{fan2012hoeffding} and was previously used to bound deviations of functions of the indicator variables of edges in sparse random graphs in Chen et al.~\cite{chen2022cut}.

\begin{theorem}[Theorem 2.1 and Remark 2.1 combined with equation (11) of \cite{fan2012hoeffding}]\label{thm:azuma-hoeffding}
Let $Z_1, \dots, Z_n$ be a martingale with martingale differences $Y_i$ satisfying $|Y_i| \leq 1$ for all $2 \leq i \leq n$. For every $0 \leq x \leq n$ and $\nu \geq 0$, we have
\begin{equation*}
\Pr \left[ \left| Z_n - Z_0 \right| \geq x \textup{ and } \iprod{Z}_n \leq \nu^2 \right]
\leq 2\exp \left( \frac{-x^2}{2(\nu^2+x/3)} \right).
\end{equation*}
\end{theorem}

\subsection{Strengthened McDiarmid's inequality}
\begin{lemma}[McDiarmid's inequality for biased distributions]
    \label{lem:stronger-mcdiarmids}\textbf{}
    Suppose that $X_1, \dots, X_n$ are sampled i.i.d.\ from a distribution $D$ over a finite set $\cX$, such that $D$ assigns probability $1-p$ to a particular outcome $\chi_0 \in \cX$. Let $f:\cX^n \to \R$ satisfy a bounded-differences inequality, so that
    \[ |f(x_1, \dots, x_{i-1}, x_i, x_{i+1}, \dots x_n) - f(x_1, \dots, x_{i-1}, x_i', x_{i+1}, \dots x_n)| \le c  \]
    for all $x_1,\dots, x_n, x_i' \in \cX$.
    Then
    \[ \Pr[ |f(X_1, \dots, X_n) - \E f(X_1, \dots, X_n)| \ge \eps ] \le 2\exp\left(\frac{-\eps^2}{2np(2-p)c^2 +2c\eps/3}\right). \]
\end{lemma}
\begin{proof}
$f$ is bounded because it has a finite domain.
Therefore, let $Z_i = \E [f(X_1, \dots, X_n) \mid X_1, \dots, X_i]$ be the Doob martingale of $f(X_1, \dots, X_n)$ and let $Y_i = Z_i - Z_{i-1}$ be its martingale difference sequence.

We will show that $|Y_i| \le c$.
We take
\[ g_j(x) := \E_{x_{j+1},\dots,x_n \sim D} \left[f(X_1, \dots, X_{j-1}, x, x_{j+1}, \dots, x_n)\right],\]
so that by definition,
\begin{align*}
Y_i
&= \E [f(X_1, \dots, X_n) \mid X_1, \dots, X_i] - \E [f(X_1, \dots, X_n) \mid X_1, \dots, X_{i-1}]. 
\\&= \E_{x_{i+1},\dots,x_n \sim D} f(X_1, \dots, X_{i-1}, X_i, x_{i+1}, \dots, x_n)
- \E_{x_{i},\dots,x_n \sim D} \left[f(X_1, \dots, X_{i-1}, x_{i}, x_{i+1}, \dots, x_n)\right]
\\&= g_i(X_i) - \E_{x \sim D} \left[g_i(x)\right]
\\&= \E_{x \sim D} \left[ g_i(X_i) -  g_i(x) \right].
\end{align*}
Then it becomes a simple matter to compute
\begin{align*}
|Y_i|
={} \left| \E_{x \sim D} \Big[g_i(X_i) - g_i(x)\Big] \right|
\le{} \E_{x \sim D} \left[\Big| g_i(X_i) - g_i(x) \Big|\right]
\le{}  c
.
\end{align*}

Next, we show that $\iprod{Z}_i \le (2-p)pc^2n$.
By definition, $\iprod{Z}_i = \sum_{j \in [i]} \E_{X_j \sim D} [ Y_j^2 \mid X_1, \dots, X_{j-1}]$ and it will suffice to show that $\E_{X_j \sim D} [ Y_j^2 \mid X_1, \dots, X_{j-1}] \ge (2-p)pc^2$. 
We start with
\[ \E_{X_j \sim D} \left[Y_j^2\right] = \E_{X_j \sim D} \left[\left(\E_{x \sim D} g(X_j) - g(x)\right)^2\right] \]
We split the expectation of $Y_j^2$ over $X_j \in \cX$ into two cases, one where $X_j = \chi_0$ and one where $X_j \ne \chi_0$.
Then 
\begin{align*}
\E_{X_j \sim D} \left[Y_j^2\right]
\;=\; \Pr_{X_j \sim D}[X_j = \chi_0] \E_{X_j \sim D}\left[Y_i^2 \mid X_j = \chi_0 \right] \;+\;  \Pr_{X_j \sim D}[X_j \ne \chi_0] \E_{X_j \sim D}\left[Y_i^2 \mid X_j \ne \chi_0\right]
\end{align*}
Since $\Pr_{X_j \sim D}[X_j \ne \chi_0] = p$ and $|Y_i| \le c$, we bound the second term in the above by $pc^2$.
Continuing with the first term,
\begin{align*}
\E_{X_j \sim D} \left[Y_j^2\right]
\le pc^2 +{}& \Pr_{X_j \sim D}[X_j = \chi_0] \left(g(\chi_0) - \E_{x \sim D} g(x) \right)^2
\\= pc^2 +{}& \Pr_{X_j \sim D}[X_j = \chi_0] \left(g(\chi_0) - \Pr_{x \sim D}[x = \chi_0]g(\chi_0) - \Pr_{x \sim D}[x \ne \chi_0]\E_{x \sim D}[g(x) \mid x \ne \chi_0] \right)^2
\\= pc^2 +{}& \Pr_{X_j \sim D}[X_j = \chi_0] \left(\Pr_{x \sim D}[x \ne \chi_0]g(\chi_0) - \Pr_{x \sim D}[x \ne \chi_0]\E_{x \sim D}[g(x) \mid x \ne \chi_0] \right)^2
\\= pc^2 +{}& \Pr_{X_j \sim D}[X_j = \chi_0] \Pr_{x \sim D}[x \ne \chi_0]\left(g(\chi_0) - \E_{x \sim D}[g(x) \mid x \ne \chi_0] \right)^2
\\\le  pc^2 +{}& (1-p)pc^2  
.\end{align*}

If we scale down the martingale $Z_1,\dots,Z_n$ by a factor of $c$, we obtain a new martingale with martingale differences bounded by $1$ in absolute value and with $i$th quadratic characteristic at most $p(2-p)i$.
Therefore, by \autoref{thm:azuma-hoeffding}, taking $\nu^2 = p(2-p)n$,
\begin{equation*}
\Pr \left[ \frac{1}{c}\left| Z_n - Z_0 \right| \geq \eps \textup{ and } \iprod{Z}_n \leq p(2-p)n \right]
\leq 2 \exp \left( \frac{-\eps^2}{2(2p(1-p)n+\eps/3)} \right).
\end{equation*}
Since $\iprod{Z}_n \leq p(2-p)n$ is always true, and absorbing a factor of $c$ into a rescaling of $\eps$,
\begin{equation*}
\Pr \left[ \left| Z_n - Z_0 \right| \geq \eps \right]
\leq 2 \exp \left( \frac{-\eps^2}{2np(2-p)c^2+2c\eps/3} \right).
\end{equation*}

\end{proof}

\section{Overlap-Gap Property for general case random \kxors}\label{sec:signed-interpolation}
In this section we prove that the coupled OGP exists even for signed versions of the random~\kxors~problem by extending the proof from~\cite[Section 4]{chen2019suboptimality}.

\subsection{The signed \kxors\ hamiltonian}
We extend~\cite[Lemma 4.1]{chen2019suboptimality} to handle an interpolation with \emph{random} signs for the variables. Explicitly, we will consider the following hamiltonian,
\begin{equation}\label{eq:signed-diluted-ham}
    H^G_{\mathrm{signed}} = -\sum_{i = 1}^m\prod_{j = 1}^k p_{ij}\sigma_{v_{ij}}\, ,
\end{equation}
where $p_{ij} \sim \{\pm 1\}$ are i.i.d.~Rademacher random variables~(\autoref{def:random-k-xor}). Notice that when we deterministically fix $p_{ij} = 1$ for every $j$-th variable in the $i$-th clause, this recovers the diluted $k$-spin glass hamiltonian $H^G$. Maximizing $H^G_{\mathrm{signed}}$ corresponds to solving a random instance of a \emph{signed} \kxors{} problem.

On expectation under the Rademacher distribution, the random signs will be balanced. This property allows us to amend the Poisson integration techniques in~\cite[Lemma 4.1]{chen2019suboptimality} to preserve the asymptotic behavior of the hamiltonian to be equivalent, up to a constant shift, to the unsigned setting under the coupled Guerra-Toninelli interpolation~\cite{guerra2004high}.

\subsection{The Guerra-Toninelli interpolations}
We work with two families of interpolations: The first is a gaussian interpolation between independent copies of $k$-mean field hamiltonians and the second is an interpolation between a diluted spin-glass hamiltonian and a dense one. For the gaussian interpolation, choose an instance of a $k$-mean field hamiltonian $H_k$ with independent copies $H_k'$ and $H_k''$ and interpolate smoothly as
\begin{align}
    H^1_k = \sqrt{t}H_k + \sqrt{1-t}H_k'\, ,\\
    H^2_k = \sqrt{t}H_k + \sqrt{1-t}H_k'' \, ,
\end{align}
where $t \in [0, 1]$. The diluted-to-dense interpolation is known as the Guerra-Toninelli interpolation~\cite{guerra2004high} and is given as
\begin{align}
    H(s, \sigma^1, \sigma^2) = \sum_{l=1}^2\left(\delta H^l_{\mathrm{signed}, \frac{d}{k}(1 - s), t}(\sigma^l) + \sqrt{s}\beta H^l_k(\sigma^l)\right)\, ,
\end{align}
where $H^1_{\mathrm{signed}, \frac{d}{k}(1 - s), t}$ and $H^2_{\mathrm{signed}, \frac{d}{k}(1 - s), t}$ are drawn from the distribution of the coupled interpolation defined in~\autoref{def:coupled-interpolation} with $\mathrm{Poisson}(\frac{d}{k}(1-s))$ edges and the additional requirement that the Rademacher variables also be re-sampled for every $t \in [0,1]$. \\
The last notion we need is that of an average with respect to the so-called \emph{Gibbs Measure}, which is a normalized probability given to every pair of configurations weighted by their coupled energy at time $s$. The Gibbs measure over $A \subseteq \{\pm 1\}^n$ and overlap set $S \subseteq [0, 1]$ is defined as
\begin{align}\label{eq:gibbs-measure}
    G_s(\sigma^1, \sigma^2) = \frac{\exp(H(s, \sigma^1, \sigma^2))}{\sum_{\sigma^1, \sigma^2 \in A, |R_{1,2}| \in S}{\exp(H(s, \sigma^1, \sigma^2))}}\, .
\end{align}
An average of a quantity $\mathcal{B}$ with respect to the Gibbs measure is denoted as $\langle \mathcal{B} \rangle_s$. The denominator of \autoref{eq:gibbs-measure}, denoted as $Z$, is a normalization term called the \emph{partition function}.

\subsection{Coupled Overlap-Gap Property for general case \kxors}
\begin{lemma}[Scaling of random \kxors{} under signed Guerra-Toninelli interpolation]
\label{lem:interpolation-kxor}
    For any $A \subseteq \{\pm1\}^n$, $S \subseteq [0, 1]$, $t \in [0, 1]$ and $d, k >0$, the following holds,
    \begin{align}\label{eq:interpolated-errors}
        &\frac{1}{n}\lE\left[\max_{\sigma^1, \sigma^2 \in A, |R_{1,2}| \in S} H_{\mathrm{signed}}^1(\sigma^1) + H_{\mathrm{signed}}^2(\sigma) \right] \\&= \frac{1}{n}\sqrt{\frac{d}{k}}\lE\left[\max_{\sigma^1, \sigma^2 \in A, |R_{1,2}| \in S} H_{k}^1(\sigma^1) + H_{k}^2(\sigma) \right]
        + O\left(\left(\frac{d}{k}\right)^{1/3}\right)\, .\nonumber
    \end{align}
\end{lemma}
\begin{proof}
We split the $H^G_{\mathrm{signed}}$ terms in the Guerra-Toninelli interpolated hamiltonian into two terms. In the first term, we will collect the positive signed hyperedges ($\mathrm{II}^+$) and in the second term we will collect the negative signed hyperedges ($\mathrm{II}^-$).
\begin{align*}
    H^{G}_{\mathrm{signed}} &= -\sum_{i = 1}^m\prod_{j = 1}^k p_{ij}\sigma_{v_{ij}} \\
    &= -\sum_{i=1}^m\prod_{j=1}^k(p_{ij})\prod_{l=1}^k(\sigma_{v_{ij}}) \\
    &:= -\sum_{i=1}^m p_i\prod_{j=1}^k(\sigma_{v_ij{}})\, , 
\end{align*}
where $p_i=\prod_{j=1}^k(p_{ij})$ are i.i.d.~Rademacher random variables for $i \in [m]$. This follows because each $p_{ij}$ is an independent Rademacher variable, and therefore, $\prod_{j=1}^k p_{ij} \sim \{\pm 1\}$ is also a Rademacher variable. The sum is now split into positive terms ($p_{ij}=1$) and negative terms ($p_{ij} = -1$).

In other words, we may write 
\[ H_{\mathrm{signed}, \frac{d}{k}(1 - s), t}^{\ell} = H_{\frac{d}{k}(1 - s)/2, t}^{\ell} - H_{\frac{d}{k}(1 - s)/2, t}^{\ell}{}',\]
where $H_{\frac{d}{k}(1 - s)/2, t}^{\ell}$ and $H_{\frac{d}{k}(1 - s)/2, t}^{\ell}{}'$ are two independent copies of the coupled distribution defined in~\autoref{def:coupled-interpolation} with $\mathrm{Poisson}(\frac{d}{k}(1-s)/2)$ hyperedges. So in the interpolation, we write
\[ \phi(s) = \frac{1}{n}\E\bigg[\log \sum \exp [\sum_{\ell \in \{1,2\}} \sqrt{s}\beta H^{\ell}_k(\sigma^{\ell}) + \delta H_{\frac{d}{k}(1 - s)/2, t}^{\ell}(\sigma^{\ell}) - \delta H_{\frac{d}{k}(1 - s)/2, t}^{\ell}{}'(\sigma^{\ell})]  \bigg]\]
We break the interpolation into a sum of a Gaussian interpolation and two Poisson interpolations,
\[ \phi'(s) = \mathrm{I} + \mathrm{II^+} + \mathrm{II^-}, \]
so that
\[\mathrm{I} = \lE\bigg[\frac{1}{Z}\left(\sum\exp(H(s, \sigma^1, \sigma^2))\sum_{\ell \in \{1, 2\}}\frac{1}{2\sqrt{s}}\beta H^\ell_k(\sigma^\ell)\right)\bigg] = \sum_{\ell \in \{1, 2\}}\bigg\langle\E\bigg[\frac{1}{2\sqrt{s}}\beta H^{\ell}_k(\sigma^\ell)\bigg]\bigg\rangle_s, \]
for which we may then apply Stein's Lemma.
A similar calculation applies for $\mathrm{II^+}$ and $\mathrm{II^-}$, for which we use the fact about Poisson random variables that
\[ \frac{d}{ds}\E f(\mathrm{Poisson}(s)) = \E f(\mathrm{Poisson}(s) + 1) - \E f(\mathrm{Poisson}(s)), \]
as in~\cite[Proof of Lemma 4.1, Page 13]{chen2019suboptimality}.

\textbf{Case I: Positive $p_{ij}$.} This is equivalent to the terms obtained by~\cite[Proof of Lemma 4.1, Page 15]{chen2019suboptimality}. Namely, the following term is obtained (after Taylor Expansion),
\begin{align*}
    \mathrm{II}^+ = -\frac{d}{k}\log\cosh(\delta) + \frac{d}{2k}\sum_{r \geq 1}\frac{\tanh(\delta)^r}{r}\left(t\lE[\langle \Delta(\sigma^1, \sigma^2)\rangle^r_s] + (1-t)\sum_{l=1}^2\lE[\langle \sigma^l_{i_1}\cdots\sigma^l_{i_k}\rangle^r_s] \right)
\end{align*}

Using replicas to represent the difference term $\Delta(\sigma^1, \sigma^2)$ and then evaluating it under a random choice of indices, followed by a second-order Taylor expansion of the preceding term,~\cite{chen2019suboptimality} obtain
\begin{align*}
    \mathrm{II}^+ &= -\frac{d}{k}\log \cosh(\delta) + \frac{d}{2k}\tanh(\delta)\lE[\langle m(\sigma^{1, 1})^k + m(\sigma^{1, 2})^k \rangle_s] - t\frac{d}{2k}\tanh(\delta)^2\lE[\langle (R^{1,1}_{1,2})^k\rangle_s] \\
    &+ \frac{d}{2k}\frac{\tanh(\delta)^2}{2}\lE[\langle (R^{1,1}_{1,2})^k\rangle_s + \langle (R^{1,2}_{2,2})^k\rangle_s + t\left(\langle (R^{1,2}_{1,2})^k\rangle_s + \langle (R^{1,2}_{2,1})^k\rangle_s \right)] + O\left(\frac{d}{k}\delta^3\right)\, .
\end{align*}
We divide the term in~\cite[Lemma 4.1]{chen2019suboptimality} by $\frac{1}{2}$ since we are working with half the edges (in expectation) in the modified interpolation.

\textbf{Case II: Negative $p_{ij}$.} This case is equivalent to that of positive signs up to a change in the sign of a field term that depends on the overlap between two replicas. This sign flip eliminates the magnetization term that appears in the statement of~\cite[Lemma 4.1]{chen2019suboptimality}. To this extent, we define a modified $\Delta$ function called $\Delta^-$, similar to~\cite{chen2019suboptimality}.
\begin{align*}
    \mathrm{II}^- ={}&  -\frac{dt}{2k} \left( \E \log \sum \exp H^-(s,\sigma^1, \sigma^2) - \E \log \sum \exp H(s,\sigma^1, \sigma^2) \right)
    \\&{}- \frac{(1-t)d}{2k} \left( \E \log \sum \exp H^-_1(s,\sigma^1, \sigma^2) - \E \log \sum \exp H(s,\sigma^1, \sigma^2) \right)
    \\&{}- \frac{(1-t)d}{2k} \left( \E \log \sum \exp H^-_2(s,\sigma^1, \sigma^2) - \E \log \sum \exp H(s,\sigma^1, \sigma^2) \right)
\end{align*}
\[ \Delta^{-}(\sigma^1, \sigma^2) = \sigma_{j_1}^1 \cdots \sigma_{j_K}^1 + \sigma_{j_1}^2 \cdots \sigma_{j_K}^2 + \tanh(\delta) \sigma_{j_1}^1 \cdots \sigma_{j_K}^1 \sigma_{j_1}^2 \cdots \sigma_{j_K}^2\]
Repeating the same perturbation based calculation on the interpolated hamiltonian with $p_{e} = -1$ in front of every hyperedge $e$ as is done in the positive case, it is not hard to obtain
\begin{align*}
    \mathrm{II}^- %
    &= -\frac{d}{k} \log \cosh \delta \\
    &- \frac{d}{2k}\sum_{r = 1}^{\infty} (-1)^{r-1}\frac{\tanh(\delta)^r}{r}\left( t\E\iprod{\Delta^-(\sigma^1,\sigma^2)}_s^r + (1-t)\left[\E\iprod{\sigma_{j_1}^1 \cdots \sigma_{j_K}^1}_s^r +
    \E\iprod
    {\sigma_{j_1}^2 \cdots \sigma_{j_K}^2}_s^r\right]\right)
\end{align*}
By the arguments in the proof of~\cite[Lemma 4.1]{chen2019suboptimality}, introducing replicas $\sigma^{\ell,1}$ and $\sigma^{\ell,2}$ for $\sigma^{1}$ and $\sigma^2$, the first term of this Taylor series is
\begin{align*}
 &\tanh(\delta)\left( t\E\iprod{\Delta^{-}(\sigma^1,\sigma^2)}_s + (1-t)\E\iprod{\sigma_{j_1}^1 \cdots \sigma_{j_K}^1}_s + (1-t)\E\iprod{\sigma_{j_1}^2 \cdots \sigma_{j_K}^2}_s  \right) 
 \\{}&=
 \tanh(\delta)\left( t\E\iprod{m(\sigma^{1,1})^k + m(\sigma^{1,2})^k + \tanh(\delta)(R^{1,1}_{1,2})^k}_s + (1-t)\E\iprod{m(\sigma^{1,1})^k + m(\sigma^{1,2})^k}_s \right)
 \\{}&=
 \tanh(\delta)\left(\E\iprod{m(\sigma^{1,1})^k + m(\sigma^{1,2})^k}_s + t \tanh(\delta)\E\iprod{(R^{1,1}_{1,2})^k}_s \right)
\end{align*}

We now now compute the second-order term of the Taylor expansion above and evaluate the terms $\lE[\langle \Delta^{-}(\sigma^1, \sigma^2) \rangle^2_s]$ and $\sum_{l=1}^2\lE[\langle \sigma^l_{j_1}\cdots\sigma^{l}_{j_k} \rangle^2_s]$. \\
Using the introduced replicas and averaging over the indices of the hyperedges as in~\cite[Proof of Lemma 4.1]{chen2019suboptimality} yields
\begin{align}\label{eq:exp-one-minus}
    &\lE[\langle \Delta^{-}(\sigma^1, \sigma^2) \rangle^2_s] = \lE\bigg[\bigg\langle\lE\nolimits' \bigg[\bigg(\sum_{l=1}^2\prod_{r = 1}^k\sigma^{1, l}_{j_r} + \tanh(\delta)\left(R^{1, 1}_{1, 2}\right)^k\bigg)\bigg(\sum_{l=1}^2\prod_{r = 1}^k\sigma^{2, l}_{j_r} + \tanh(\delta)\left(R^{2, 2}_{1, 2}\right)^k\bigg)\bigg]\bigg\rangle_s\bigg] \nonumber \\
    &= \lE\bigg[\bigg\langle \sum_{r_1, r_2 = 1}^2\left(R^{1, 2}_{r_1, r_2}\right)^k \bigg\rangle_s + \tanh(\delta)\left((R^{2, 2}_{1, 2})^k\sum_{l=1}^2\prod_{r = 1}^k\sigma^{1, l}_{j_r} + (R^{1, 1}_{1, 2})^k\sum_{l=1}^2\prod_{r = 1}^k\sigma^{2, l}_{j_r}\right) + \mathcal{O}(\tanh(\delta)^2)\bigg]\, .
\end{align}
As in~\cite[Proof of Lemma 4.1, Pg 15]{chen2019suboptimality}, for the second term we have that
\begin{align*}
    \sum_{l=1}^2\lE[\langle\sigma^l_{i_1}\cdots\sigma^l_{i_k}\rangle^2_s] = \lE\bigg[\bigg\langle\sum_{r=1}^2 (R^{1,2}_{r, r})^k\bigg\rangle_s\bigg]\, .
\end{align*}

Using the facts that the $\Delta^-(\sigma^1, \sigma^2) \leq 3$ and $|\sigma^l_{i_1}\cdots\sigma^{l}_{i_k}| \leq 1$, all terms of order $\geq 3$ in the Taylor expansion are no more than $L\frac{d}{k}\delta^3$ for an appropriate $L > 0$.
Putting together the expansions yields the following bound,
\begin{align*}
    \mathrm{II}^- &= -\frac{d}{k}\log\cosh(\delta) -\frac{d}{2k}\tanh(\delta)\left(\E\iprod{m(\sigma^{1,1})^k + m(\sigma^{1,2})^k}_s +  t\tanh(\delta)\E\iprod{(R^{1,1}_{1,2})^k}_s \right) \\
    &+ \frac{td}{2k}\frac{\tanh(\delta)^2}{2}\left(\lE\bigg[\bigg\langle \sum_{r_1, r_2 = 1}^2\left(R^{1, 2}_{r_1, r_2}\right)^k \bigg\rangle_s + \tanh(\delta)\left((R^{2, 2}_{1, 2})^k\sum_{l=1}^2\prod_{r = 1}^k\sigma^{1, l}_{j_r} + (R^{1, 1}_{1, 2})^k\sum_{l=1}^2\prod_{r = 1}^k\sigma^{2, l}_{j_r}\right) + \mathcal{O}(\tanh(\delta)^2)\bigg]\right) \\
    &+ \frac{d(1-t)}{2k}\frac{\tanh(\delta)^2}{2}\left(\lE\bigg[\bigg\langle\sum_{r=1}^2 (R^{1,2}_{r, r})^k\bigg\rangle_s\bigg]\right)\, .\nonumber
\end{align*}
Notice that all terms of $O(\tanh(\delta)^3)$ and higher powers therein can be absorbed into the term $L\frac{d}{k}\delta^3$. This finally yields
\begin{align}
     \mathrm{II}^- &= -\frac{d}{k}\log\cosh(\delta) -\frac{d}{2k}\tanh(\delta)\left(\E\iprod{m(\sigma^{1,1})^k + m(\sigma^{1,2})^k}_s\right)  -\frac{td}{2k}\tanh(\delta)^2\left(\E\iprod{(R^{1,1}_{1,2})^k}_s \right) \nonumber \\
    &+ \frac{d}{2k}\frac{\tanh(\delta)^2}{2}\left(\lE\bigg[\bigg\langle t\left(R^{1, 2}_{1, 2}\right)^k + t\left(R^{1, 2}_{2, 1}\right)^k + \left(R^{1, 2}_{1, 1}\right)^k + \left(R^{1, 2}_{2, 2}\right)^k\bigg\rangle_s\right)\bigg] +  O\left(\frac{d}{2k}\delta^3\right)\, .
\end{align}

\textbf{Combining $\mathrm{II}^+$ and $\mathrm{II}^-$.} We now add the positive and negative terms together to obtain the equivalent of $\mathrm{II}$ for $H^G_{\mathrm{signed}}$. As a result of the sign flip, the magnetiziation dependencies cancel out. %
This finally yields
\begin{align}
    \mathrm{II} = \mathrm{II}^+ + \mathrm{II}^- &= -\frac{2d}{k}\log\cosh(\delta) - \frac{td}{k}\tanh(\delta)^2\left(\E\iprod{(R^{1,1}_{1,2})^k}_s \right) \nonumber \\
    &+ \frac{d}{k}\frac{\tanh(\delta)^2}{2}\left(\lE\bigg[\bigg\langle t\left(R^{1, 2}_{1, 2}\right)^k + t\left(R^{1, 2}_{2, 1}\right)^k + \left(R^{1, 2}_{1, 1}\right)^k + \left(R^{1, 2}_{2, 2}\right)^k\bigg\rangle_s\right)\bigg] + O\left(\frac{d}{k}\delta^3\right)\, .
\end{align}
For fixed $d, k$ and $\delta$, we define $\beta$ as
\[
    \beta = \sqrt{\frac{d}{k}}\tanh(\delta)\, .
\]
This choice causes the overlap terms to cancel in $\phi(s)$, yielding the following rate of change of the free energy
\begin{align}
    \phi'(s) = \mathrm{I} + \mathrm{II} = -\frac{d}{k}\log\cosh(\delta) + \frac{d}{k}\tanh(\delta)^2 + O\left(\frac{d}{k}\delta^3\right) \overset{\delta \to 0}{\longrightarrow} O\left(\frac{d}{k}\delta^3\right) + O\left(\frac{d}{k}\delta^4\right)\, ,
\end{align}
where $\mathrm{I}$ is defined equivalently as in the proof of~\cite[Lemma 4.1, Pg 13]{chen2019suboptimality}. This immediately yields that
\begin{align*}
    \frac{1}{\delta}\left(\phi(1) - \phi(0)\right) = \frac{1}{\delta}\int_{0}^1\phi'(s)ds = O\left(\frac{d}{k}\delta^2\right)\, .
\end{align*}
The rest of the argument follows exactly as in~\cite[Proof of Lemma 4.1, Pg 16]{chen2019suboptimality} with the final substitution $\delta = \left(\frac{d}{k}\right)^{-1/3}$.
\end{proof}

Having proved the key interpolation lemma about the coupled free energies in the diluted and dense models, we now state the coupled OGP for the random \emph{signed}~\kxors\ problem.  
\begin{theorem}[Coupled OGP for random \kxors{}, $k$ even]
\label{thm:ogp-kxors-coupled}
For every even $k \geq 4$, there exists an interval $0 < a < b < 1$ and parameters $d_0 > 0$, $0 < \eta_0 < P(k)$ and $n_0 > 1$, such that, for any $t \in [0,1]$, $d \geq d_0$, $n \geq n_0$ and constant $L = L(\eta_0, d)$, with probability at least $1 - Le^{-n/L}$ over the $t$-coupled \kxors{} instance pair $(\Psi_1, \Psi_2) \sim \mathcal{H}^{\pm}_{n, d, k, t}$, whenever two spins $\sigma_1,\ \sigma_2$ satisfy
    \[
        \frac{H^{\Psi_i}(\sigma_i)}{n} \geq M(k, d)\left(1 - \frac{\eta_0}{P(k)}\right) \, ,
    \]
    then their overlap satisfies $|R(\sigma_1, \sigma_2)| \notin [a, b]$.
\end{theorem}
\begin{proof}
This can be proven by an exact copy of~\cite[Theorem 5]{chen2019suboptimality}, but using \autoref{lem:interpolation-kxor} instead of~\cite[Lemma 4.1]{chen2019suboptimality}.
\end{proof}

To obstruct local algorithms using the same framework mentioned in~\autoref{sec:tech overview}, it is critical that the overlap between nearly optimal solutions of independent random~\kxors~instances (with signs) be small. To show this, we first prove that if the overlap between pairs of solutions of two independent instances of the $k$-mean field model is bounded away from 0, then they are suboptimal.
\begin{lemma}\label{lem:subopt-k-mean}%
\label{lem:pspin-mag-convex}
Consider the parameters $[a,b]$ and $\eta_0$ from \autoref{thm:ogp-kxors-coupled}.
For large enough $n$, there is $\hat{\eta} > 0$ satisfying $\hat{\eta} > \eta_0$ such that
\[
\frac{1}{n} \E\bigg[\max_{|R_{1,2}| \in [a,1]} (H_k^1(\sigma^1) + H_k^2(\sigma^2))\bigg] < 2(P(k) - \hat{\eta})
\]
where $H_k^1$ and $H_k^2$ are random 0-coupled instances of a $k$-spin glass.
\end{lemma}
\begin{proof}

In the proof of \autoref{thm:ogp-kxors-coupled} as in that of~\cite[Theorem 5]{chen2019suboptimality}, the parameter $a$ is chosen from~\cite[Theorem 3]{chen2019suboptimality}, and the parameter $\eta_0$ is chosen so that
\[ \limsup_{n \to \infty} \E \frac{1}{n} \max_{|R_{1,2}| \in [a,b]} (H_k^1(\sigma^1) + H_k^2(\sigma^2)) < 2P(k) - 6\eta_0, \]
for all $t \in [0,1]$, where $R_{1,2}$ is the overlap between $\sigma^1$ and $\sigma^2$, and $H_k^1$ and $H_k^2$ are $t$-coupled instances of a $k$-spin glass.

However, in fact for the $t = 0$ case (as noted in the remark in the proof of \cite[Theorem 5]{chen2019suboptimality}), the proof cites~\cite[Theorem 2]{chen2018disorder}, which provides a bound for $|R_{1,2}| \in [a,1]$ with $|R_{1,2}| \in [a,b]$ being a subcase of that. This implies
\[ \limsup_{n \to \infty} \E \frac{1}{n} \max_{|R_{1,2}| \in [a,1]} (H_k^1(\sigma^1) + H_k^2(\sigma^2)) < 2P(k) - 6\eta_0, \]
at $t = 0$.

Continuing in the special case where $t=0$, then, let $\eta$ be so that 
\[ \limsup_{n \to \infty} \E \frac{1}{n} \max_{|R_{1,2}| \in [a,1]} (H_k^1(\sigma^1) + H_k^2(\sigma^2)) = 2(P(k) - \eta), \]
Since $\eta \ge 3\eta_0$,  we choose $\hat{\eta} = \eta + \epsilon$, for some $\epsilon$ which is allowed to be arbitrarily small as $n \to \infty$. 

\end{proof}

We now extend~\autoref{lem:subopt-k-mean} to setting of the diluted model with signs, and show that pairs of nearly optimal solutions of independent random~\kxors~instances (signed) have low overlap. 
\begin{lemma}
\label{lem:kxors-overlap}
Consider the parameters $[a,b]$ and $\eta_0$ from \autoref{thm:ogp-kxors-coupled}.
For large enough $n$, there are $\eta', L > 0$ with $\eta' > \eta_0$ such that
\begin{align*}
\frac{1}{n}\max_{\substack{\sigma^1,\sigma^2 \in \{\pm 1\}^n, \\ |R_{1,2}| \ge a}} (H_{\mathrm{signed}}^1(\sigma^1) + H_{\mathrm{signed}}^2(\sigma^2)) &{}\le 2(P(k) - \eta')\sqrt{\frac{d}{k}} + O(\sqrt[3]{d/k})
\\&\le 2M(k,d)\left(1 - \frac{\eta'}{P(k)} + O\left(\left(\frac{d}{k}\right)^{-1/6}\right)\right).
\end{align*}
with probability at least $1-2Le^{n/L}$ over the random choice of $0$-coupled \kxors{} instances $H_{\mathrm{signed}}^1$ and $H_{\mathrm{signed}}^2$.
\end{lemma}
\begin{proof}
By \autoref{lem:pspin-mag-convex}, there is a choices of $\hat{\eta}$ satisfying $\hat{\eta} > \eta_0$ such that
\[
\E \frac{1}{n} \max_{\substack{\sigma^1,\sigma^2 \in \{\pm 1\}^n \\ |R_{1,2}| \in [a,1]}} (H_k^1(\sigma^1) + H_k^2(\sigma^2)) < 2(P(k) - \hat{\eta}),
\]
where $H_{k}^1$ and $H_{k}^2$ are still $0$-coupled instances of the mean-field $k$-spin glass and $R_{1,2}$ is the overlap between $\sigma^1$ and $\sigma^2$.

By instantiating \autoref{lem:interpolation-kxor} with $S = [a,1]$ and $A = \{\pm 1\}^n$ and $t=0$, we see that
\[\frac{1}{n}\lE\max_{\substack{\sigma^1,\sigma^2 \in \{\pm 1\}^n \\ |R_{1,2}| \in [a,1]}} (H_{\mathrm{signed}}^1(\sigma^1) + H_{\mathrm{signed}}^2(\sigma^2))  = \frac{1}{n}\sqrt{\frac{d}{k}}\lE\max_{\substack{\sigma^1,\sigma^2 \in \{\pm 1\}^n \\ |R_{1,2}| \in [a,1]}} (H_{k}^1(\sigma^1) + H_{k}^2(\sigma^2)) \pm O(\sqrt[3]{d/k}),\]
Combining the two above inequalities,
\[ \frac{1}{n}\lE\max_{\substack{\sigma^1,\sigma^2 \in \{\pm 1\}^n \\ |R_{1,2}| \in [a,1]}} (H_{\mathrm{signed}}^1(\sigma^1) + H_{\mathrm{signed}}^2(\sigma^2)) \le 2(P(k) - \hat{\eta})\sqrt{\frac{d}{k}} + O(\sqrt[3]{d/k}) \]

By an application of Azuma's inequality and concentration of Poisson random variables, there is some $L$ as a function of $\hat{\eta}$ and $\eta_0$ such that 
\[ \frac{1}{n}\max_{\substack{\sigma^1,\sigma^2 \in \{\pm 1\}^n, \\ |R_{1,2}| \ge a}} (H_{\mathrm{signed}}^1(\sigma^1) + H_{\mathrm{signed}}^2(\sigma^2))  \le 2(P(k) - \eta')\sqrt{\frac{d}{k}} + O(\sqrt[3]{d/k})\]
for some $\eta'$ satisfying $\eta_0 < \eta' < \hat{\eta}$, with probability at least $1-2Le^{-n/L}$.
\end{proof}

\section{Discussion}\label{sec:conjectures}

 Our work conclusively establishes the coupled OGP as an obstruction to all local quantum algorithms on \emph{any} $(k, d)\text{-}\mathsf{CSP}(f)$. In doing this, the work hints at and leaves open many interesting questions for future work in areas that are at the intersection of Quantum inapproximability, Statistical Physics, Random Graph Theory, Combinatorial Optimization and Average-Case Complexity.

\subsection{Which CSPs have an OGP?}\label{subsec:which-csps-have-ogp}
While various sparse CSPs such as $\mathsf{k}\text{-}\mathsf{SAT}$, unsigned $\mathsf{max}\text{-}k\text{-}\mathsf{XOR}$ and $\mathsf{k}\text{-}\mathsf{NAE}\text{-}\mathsf{SAT}$ have been shown to exhibit clustering in their solution spaces at different clause-to-variable ratios~\cite{achlioptas2006solution, ding2016satisfiability, chen2019suboptimality}, it is not known whether this property is pervasive to most CSPs or something that happens to a select few. Therefore, in order to understand the complexity landscape of CSPs on typical instances better, the following open question is interesting to investigate:
\begin{conjecture}[Random Predicate CSPs and coupled OGP]
\label{prob:csp-ogps}
  Given a function $f$ chosen uniformly at random from the set of functions $\mathcal{B}_k = \{g\ \mid\ g:\{\pm 1\}^k \to \{0, 1\}\}$, $(k, d)\text{-}\mathsf{CSP}(f)$ has a coupled-OGP for sufficiently large $k$ and $d$ with high probability (over the choice of $f$ and instance $\Psi \sim (k,d)\text{-}\mathsf{CSP}(f)$).
\end{conjecture}

Notice that the conjecture above is specifically interested in the solution geometry of a CSP in the \emph{unsatisfiable} regime (large $d$). A positive resolution to the above conjecture will make the obstructions stated in~\autoref{thm:informal-obstruct-everything} hold for \emph{almost all} CSPs for the family of generic local algorithms. \\
Another question of interest is which properties (if any) about a predicate $f$ can be identified which would conclusively imply that a random instance $\Psi$ of a $(k, d)$-$\mathsf{CSP}(f)$ will have an OGP.

\begin{problem}[Properties of coupled-OGP predicates]\label{conj:prop-coupled-ogp}
    Can we enumerate a set of necessary and sufficient conditions on $f$ to be such that $(k,d)\text{-}\mathsf{CSP}(f)$ satisfies a coupled-OGP for sufficiently large $k$ and $d$?
\end{problem}

\subsection{Beyond \texorpdfstring{$\log$}{log}-depth obstructions for \texorpdfstring{$QAOA_p$}{QAOA(p)}?}\label{subsec:beyond-log-depth}
Work on obstructing $QAOA_p$ using an OGP heavily relies on the locality of the algorithm at shallow depths. It is interesting to investigate whether this obstruction can be extended beyond the $\epsilon\log(n)$-depth regime to make this a non-local obstruction. Recent work~\cite{gamarnik2020low, wein2022optimal} suggests that the OGP may actually result in stronger obstructions than just local ones, and it would be interesting to see if these techniques can be generalized to the setting of $QAOA_p$ to yield obstructions that are non-local.

\begin{problem}[Poylogarithmic obstructions to $QAOA_p$ in the OGP regime]\label{prob:polylog-obst}
  Given a $QAOA_p$ circuit with depth $p \leq \epsilon\left(\log(n)\right)^c$ for some $c > 1$, does there exist $\epsilon_0 > 0$, such that $QAOA_p$ is obstructed on a $(k, d)\text{-}\mathsf{CSP}(f)$ with a coupled OGP from outputting solutions that are better than (1 - $\epsilon_0$) approximations to the optimal?
\end{problem}

\subsection{A Quantum OGP and lifting ``classical" obstructions}
The idea of the OGP obstructing families of algorithms that are \emph{stable} under small perturbations~\cite{gamarnik2020low} motivates the idea of a quantized version of the OGP, to apply to  \emph{quantum} CSPs. To define such a property over quantum states, however, there would need to be a metric that is very similar to the classical hamming distance over $\mathbb{F}_2$ and has the property that it is invariant over permutations of the canonical basis, while still quantifying entanglement in a desired way. One such possible metric is a \emph{quantum} version of the Wasserstein distance of first order that was proposed by De-Palma et al.~\cite{de2021quantum}. In particular, given a natural generalization of~\autoref{def:coupled-ogp-informal} to a quantized setting using a quantum version of the Wasserstein distance of first order, it is interesting to investigate if a larger family of quantum circuits up to some depth $p(n)$ can be obstructed by a family of $d$-local hamiltonians $\{H_n\}_{n \geq n_0}$ that possess a qOGP (quantized Overlap-Gap Property). A result of this type could imply a way to generically ``lift" classical obstructions for \emph{stable} classical algorithms to a corresponding family of quantum algorithms.

\subsection{Message-Passing algorithm for \texorpdfstring{$\mathsf{MAX\text{-}CUT}$}{MAX-CUT} of \emph{all} \texorpdfstring{$d$}{d}-regular graphs?}\label{subsec:max-cut-d-reg}
Finding an efficient classical algorithm that can output cuts that are arbitrary approximations of the optimal ones for $d$-regular graphs is a long-standing open problem in Random Graph Theory and Theoretical Computer Science. Recently, this problem was nearly completely solved by Alaoui et al.~\cite{alaoui2021local} as they constructed a Message-Passing algorithm for random regular graphs of very large degree under the widely believed no-OGP assumption about the SK model. However, the problem does not provide a complete solution as it needs the degree $d$ to be larger than $O(\frac{1}{\epsilon})$ in order to output a $(1-\epsilon)$-optimal cut. A natural question is whether, under a no-OGP assumption, the result can be extended to output $(1-\epsilon)$-optimal cuts for $d$-regular graphs for \emph{any} $d \geq 3$.

\begin{conjecture}[AMP algorithm for Random $d$-Regular Graphs]
  There exists a $poly(n, \frac{1}{\epsilon})$ time algorithm $A$ that outputs a (1 - $\epsilon)$-approximate cut of a random $d$-regular graph $G$ with high probability under a ``no-OGP" assumption for \emph{any} $d \geq 3$.
\end{conjecture}

Note that the approach of Alaoui et al.~\cite{alaoui2021local} critically relies on the Guerra-Tonnineli interpolation between the $\mathcal{G}_{n, d}$ model and the SK-model which will \emph{only} work for $d \geq O\left(\frac{1}{\epsilon}\right)$. Consequently, a solution that works for \emph{all} $d \geq 3$ will require a fundamentally different approach. A natural question that is motivated by the above conjecture is to then investigate if there is any range of degree for which the \maxcut~problem over $d$-regular graphs possesses an OGP. Given the belief that the SK model does not exhibit an OGP, this would only be an interesting question in the relatively low-degree regime.

\begin{problem}[Random $d$-Regular Graphs don't have an OGP]
  Does the $\mathsf{MAX}$-$\mathsf{CUT}$ problem on random $d$-regular graphs have an OGP for some $d \geq 3$? If so, for what $\{d_0, d_1\} \subset \mathbb{N}$ does the problem exhibit an OGP?
\end{problem}

The $QAOA_p$ algorithm was initiated and analyzed on the $\mathsf{MAX}$-$\mathsf{CUT}$ of $d$-regular graphs and positive answers to the conjectures above will close the scope for any quantum advantage on the problem.

\newpage

\section*{Acknowledgements}
We thank Jonathan Wurtz for many insightful discussions about QAOA. We are grateful to Amartya Shankha Biswas for patiently explaining the \emph{factors of i.i.d.}\ framework to us. We would also like to thank Antares Chen for many invigorating and profound discussions which culminated as the open problem proposed in~\autoref{conj:prop-coupled-ogp}. Lastly, we would like to thank Boaz Barak for providing detailed and helpful feedback on a prior version of this manuscript, and David Gamarnik for his explanations on the state of the art results in the research area.

\bibliography{./main.bib}
\bibliographystyle{alpha-betta}

\newpage
\appendix
\section{Proof of Proposition~\ref{prop:p-local-vs-factors}}\label{sec:proof-bell-generalization}
\begin{proof}
    We describe a circuit that implements a Bell experiment: a Bell pair of entangled qubits is created, and then a unitary transformation is randomly and independently applied to each qubit before they are measured (equivalently, a random basis is chosen for each measurement).
    
    Consider 4 qubits $\ket{a_ca_eb_eb_c}$ in the state $\ket{0}^{\otimes 4}$. Apply the $H$ gate to $a_c$, $a_e$ and $b_c$ so that they enter into the $\ket{+}$ state, and leave $b_e$ as is. Using $a_e$ as the control qubit, apply a CNOT gate to $b_e$. This results in the creation of a Bell pair $\ket{\phi_+}$ between $a_e$ and $b_e$. Now apply a controlled unitary $C$-$U$ to $a_e$ using $a_c$ as the control qubit. Similarly, apply the controlled unitary $C$-$U$ to $b_e$ using $b_c$ as the control qubit. Finally, measure all qubits in the $Z$ basis. \Snotes{Make consistent CNOT/controlled-NOT/C-U/controlled-U} \\
    
    \noindent Now, form a graph over the qubits of this circuit, with edges between pairs of qubits that are interacted on by the same gate, as well as a self-loop on $b_e$ so as to distinguish it from $a_e$:
    \[
        G = (\{a_c, a_e, b_e, b_c\},\; \{\{a_c, a_e\}, \{a_e, b_e\}, \{b_e, b_c\}, \{b_e\}\})\, .
    \]
    An 1-local algorithm on graphs which implements the above quantum circuit when run on $G$ is as follows:
    \begin{enumerate}
        \item Create a qubit for each vertex, in the $\ket{0}$ state.
        \item Apply a Hadamard gate to each vertex without a self-loop.
        \item For every edge between a vertex of degree 4 with a self-loop and a vertex of degree 2, apply a controlled-not gate from the vertex of degree 2 to the one of degree 4.
        \item For every edge incident to a vertex of degree $1$, apply a controlled-U gate from that vertex to the one at the other endpoint of the edge.
        \item Measure all qubits in the $Z$ basis and output the results.
    \end{enumerate}
    
    \begin{figure}[ht!]
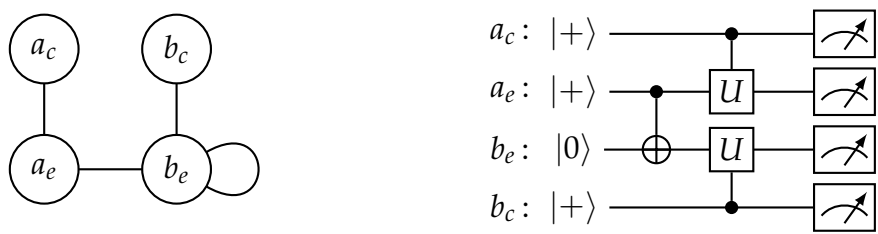
\label{fig:interfere-graph}
    \centering
    \begin{minipage}[c]{0.35\textwidth}
    \centering
    \begin{qcircuit}[]
		\tikzset{
			xscale=0.8,
			yscale=0.8,
			vertex/.style={draw, circle, minimum size=26pt, inner sep=0pt, outer sep=0pt},
			}
			
			\node[vertex] at (0,0) (ac) {\large $a_c$};
			\node[vertex] at (0,-2) (ae) {\large $a_e$} edge [-] (ac);
			\node[vertex] at (2.2,-2) (be) {\large $b_e$} edge [-] (ae);
			\node[vertex] at (2.2,0) (bc) {\large $b_c$} edge [-] (be);
			\draw (be) to [in=30,out=-30,loop,looseness=6] (be);
		\end{qcircuit}
    \end{minipage}
    \begin{minipage}[c]{0.5\textwidth}
    \centering
    \begin{qcircuit}[]
		\tikzset{
			xscale=0.5,
			yscale=0.7,
			dot/.style={fill, circle, inner sep=0pt, outer sep=0pt, minimum size=5pt},
			cross/.style={path picture={ 
					\draw[thick,black](path picture bounding box.north) -- (path picture bounding box.south) (path picture bounding box.west) -- (path picture bounding box.east);
				}},
				target/.style={draw,circle,cross,minimum size=0.26 cm},
			}
			\node at (-1.9,0) (label1) {\large $a_c$\,:};
			\node at (-1.9,-1.1) (label2) {\large $a_e$\,:};
			\node at (-1.9,-2.2) (label3) {\large $b_e$\,:};
			\node at (-1.9,-3.3) (label4) {\large $b_c$\,:};
			\node at (-0.2,0) (init1) {\large $\ket{+}$};
			\node at (-0.2,-1.1) (init2) {\large $\ket{+}$};
			\node at (-0.2,-2.2) (init3) {\large $\ket{0}$};
			\node at (-0.2,-3.3) (init4) {\large $\ket{+}$};
			\node[meter] at (7,0) (meas1) {} edge [-] (init1);
			\node[meter] at (7,-1.1) (meas2) {} edge [-] (init2);
			\node[meter] at (7,-2.2) (meas3) {} edge [-] (init3);
			\node[meter] at (7,-3.3) (meas4) {} edge [-] (init4);
			\node[dot] at (2,-1.1) (bellcnotsource) {};
			\node[target] at (2,-2.2) (bellcnottarget) {} edge [-] (bellcnotsource);
			\node[dot] at (4,0) (alicecontrolusource) {};
			\node[gate] at (4,-1.1) (alicecontrolutarget) {\large $U$} edge [-] (alicecontrolusource);
			\node[dot] at (4,-3.3) (bobcontrolusource) {};
			\node[gate] at (4,-2.2) (bobcontrolutarget) {\large $U$} edge [-] (bobcontrolusource);
		\end{qcircuit}
    \end{minipage}
    \caption{A graph $G$ (left) on which the algorithm described in the proof of \autoref{prop:p-local-vs-factors} executes a quantum circuit (right) which performs a Bell experiment. $a_e$ and $b_e$ are entangled, as suggested by their names. $a_c$ and $b_c$ are independent and correspond to control qubits.}
    \label{fig:interference-graph}
    \end{figure}
    
    We show that this algorithm is $1$-local by doing casework on each "type" of vertex.
    It is ok for the output on a vertex to depend on the degree of the vertex and whether it has any self-loops, as these are functions of the $1$-neighborhood of the vertex.
    Thus the types of vertices are as follows:
    \begin{description}
        \item[Degree-1 vertices]
        The output of a vertex of degree 1 is always an independent uniform distribution over $\{-1,1\}$.
        The only interaction of the corresponding qubit is being the source of a controlled-$U$ gate.
        This gate commutes with Pauli-$Z$ operators on the source qubit.
        Therefore, the measurement in the $Z$ basis commutes past the controlled-$U$ gate, effectively making this part of a circuit a $Z$ measurement on a $\ket{+}$ state, followed by classical control determining whether to apply a $U$ gate on its neighbor.
        Since this measurement is always independent and invariant with respect to the hypergraph, the algorithm is $1$-local on these vertices.
        \item[Degree-2 vertices]
        These can interact with degree-1 vertices and degree-4 vertices with self-loops.
        If we consider the Heisenberg picture and propagate the $Z$ measurements backward through the circuit, the first gates we encounter are controlled-$U$ gates from the degree-1 vertices.
        Recall that we can treat those controlled-$U$ gates as a single-qubit $U$ gate controlled classically by a single random bit.
        If the $U$ gate happens, then a $Z$ measurement after the $U$ gate is equivalent to a $U^{\dagger}ZU$ measurement before the $U$ gate; otherwise it remains a $Z$ measurement.
        If there's more than one degree-$1$ neighbor and more than one consequent $U$ gate, then our measurement is $U^{2\dagger}ZU^{2}$ instead, and we've reached the initialization of our qubit.
        If the other neighbor is a degree-$4$ vertex with a self-edge, then there's a controlled-NOT acting on our degree-$2$ vertex, and $U^{\dagger}ZU$ after the controlled-NOT is equivalent to some linear combination of phase changes on the degree-$4$ vertex multiplied by some single-qubit unitary on the degree-$2$ vertex.
        At this point, our measurement operator has reached the initializations of 2 qubits, both in the $1$-neighborhood of the degree-$2$ vertex we started with.
        If there was no $U$ gate applied in step $4$, and only one or more controlled-NOTs instead, then a $Z$ on the target after a controlled-NOT gate is equivalent to the product of $Z$ on the target and $Z$ on the source before the gate.
        Then the $Z$ on the target will reach an initialization, while the $Z$ on the source commutes past other being on the source of other controlled-NOT gates, to also reach the initialization.
        Thus we have described how the measurement is a function of edges in the $1$-neighborhood, and how the support of the measurement in the Heisenberg picture is within the $1$-neighborhood, so that if all other measurements also have supports in their respective $1$-neighborhoods, the algorithm is $1$-local on vertices of degree $2$.
        \Snotes{Ugh, make this... more readable? Maybe??? Is it worth it?}
        \item[Degree-4 vertices]
        These can interact with degree-$1$ vertices and degree-$2$ vertices, but only if they have a self-loop.
        The action of degree-$1$ vertices (in step $4$ of the algorithm) is identical with the degree-2 case above.
        If the backwards-propagating measurement is $Z$ at the time when it hits the source of a controlled-NOT gate, the $Z$ on the source after the controlled-NOT is equivalent to a $Z$ on the source multiplied by an $X$ on the target before the gate.
        The $Z$ on the source then propagates back through another controlled-NOT if it exists, remaining a $Z$ until it hits the initialization.
        The $X$s on the target degree-2 vertices commute through at most one other controlled-NOT gate on the target end until it hits the initialization of the degree-2 vertex.
        At the end of this process, the support of the measurement operator on the qubit initializations is fully determined by and restricted to the $1$-neighborhood of the vertex.
        On the other hand, if the measurement is $U^{\dagger}ZU$ by the time it touches the source of a controlled-not, it propagates backwards into a linear combination of single-qubit unitaries on the source and I or X operators on the target.
        The unitary on the source then hits the initialization of the degree-$4$ vertex, while the $I$ or the $X$ commutes past any other controlled-NOTs the degree-2 vertex might be the target of, to hit the initialization of the adjacent degree-2 vertex.
        Again, the measurements in the Heisenberg picture are fully determined by and restricted to the $1$-neighborhood.
        \item[All other degrees]
        Vertices of any other degree do not interact with other vertices in the algorithm.
    \end{description}
    By the above casework, this algorithm is $1$-local.
    
    Note that this process \emph{cannot} be encoded in a $1$-local factors of i.i.d.\ algorithm since this setup allows for signaling strategies that violate Bell inequalities, whereas $1$-local factors of i.i.d.\ can be explained by using latent variables to describe the evolution of the randomness of the $1$-neighborhood of every vertex.
    An extension of this circuit which would involve generating $p$ entangled qubits in a similar process generalizes the argument to $p$-local algorithms.
\end{proof}

\section{Proof of Theorem~\ref{thm:vanish-nghbd}}\label{sec:proof-vanishing-nhbhd}

\begin{proof}
    First, note that given $\frac{dn}{k}$ hyperedges (which is the expectation of the $\poisson(dn/k)$ distribution from which the number of edges are sampled), each of size $k$, the expected number of hyperedges some vertex $v_i$ shows up in is,
    \begin{equation}\label{eq:indicator-exp-hyperedges}
        \mathlarger{\Ex_{|E|}}[\Pi_i(E(HG))] = \frac{dn}{k}\cdot \Pr_{e \sim [n]^k}[v_i \in e] \leq \frac{dn}{k}\cdot\left(\frac{n^{k-1}}{n^k}\right)k = \frac{dn}{k}\cdot\left(\frac{1}{n}\right)k = d\, .
    \end{equation}
    Now, consider another model of a $k$-uniform hypergraph in which we sample $n^k$ edges independently, each with probability $p$. This induces a $Bin(n^k, p)$ distribution on the number of hyperedges. To compare this with our model, we compare the expected number of hyperedges as,
    \[
        pn^k = \frac{dn}{k} \implies p = \frac{d}{kn^{k-1}}\, .
    \]
    Note that the degree distribution of a vertex (which is equivalent to the number of hyperedges it appears in) in this model is given as,
    \[
        Bin\left(kn^{k-1}, \frac{d}{kn^{k-1}}\right)\, .
    \]
    The distribution above converges to $\poisson(d)$ in the large $n$ limit, and its moment generating function is dominated by that of the $\poisson(d)$ distribution for all large but finite $n$. Formally,
    \[
        \phi\left(Bin\left(kn^{k-1}, \frac{d}{kn^{k-1}}\right)\right) = (1 - d(1 - e^t))^{kn^{k-1}} \leq e^{d(e^t - 1)} = \phi(\poisson(d))\, .
    \]
    Therefore, the number of vertices in the 1-neighborhood of any vertex $v \in V(G)$ can be bounded from above (in the large $n$ limit) as,
    \[
        (k-1)\poisson(d)\, .
    \]
    So, we will consider the scaled Galton-Watson process above starting at some vertex $v \in V(G)$, which is itself a Galton-Watson process~\cite{harris1963theory}.
    We upper bound the scaled Galton-Watson process above with the Galton-Watson process induced by the $\poisson(d(k-1))$ distribution (at an appropriate level) by comparing their respective probability generating functions. Note that the probability generating function of $(k-1)\poisson(d)$ is $f_{\poisson(d)}(z^{k-1})$. To look at the $x$-th neighborhood, we will use the probability generating function of the $x$-th level of the underlying Galton-Watson process, denoted as $f_x$. Now, by \cite[Theorem 4.1]{harris1963theory},
    \[
        f_x = f_1^{\circ x} := \overbrace{f_{1}  \circ \dots \circ f_{1}}^x \, ,
    \]
    where $f_1 = f_{\poisson(d)}(z^{k-1}) = f_{(k-1)\poisson(d)}(z)$. By comparing the pgf of the Galton-Watson process of $f_{(k-1)\poisson(d)}(z)$ and the Galton-Watson process of $f_{\poisson(d(k-1))}(z)$, one can observe that,
    \[
        f_{x-1} \leq g_{x}\, ,
    \]
    where $g_{x} = g^{\circ x}_1$, and $g_1 = f_{\poisson(d(k-1))}(z)$. We now repeat the argument in \cite[Neighborhood Theorem]{farhi2020quantum} that bounds the size of a $\poisson(d(k-1))$ branching process. \\
    Let $Z_x$ denote the size of the $x$-th generation of a $\poisson(d(k-1))$ branching process. Note that $Z_0 = 1$, $Z_1 = \poisson(d(k-1))$ and, more generally, one can look at the moment generating function of the $\poisson(d(k-1))$ branching process as,
    \[
        \E_{\mathrm{Branching}(d(k-1))}[e^{tZ_x}] = e^{d(k-1)(\E[e^{tZ_{x-1}}] - 1)}\, .
    \]
    We denote by $\phi_x(t)$ the moment generating function of the $\poisson(d(k-1))$ branching process. It is straightforward to see by an inductive argument used in \cite[Neighborhood Size Theorem]{farhi2020quantum} that,
    \[
        \phi_x\left(\left(\frac{\ln 2}{d(k-1)}\right)^x\right) \leq e\, ,\ \forall k \geq 0.
    \]
    Furthermore, by an application of Markov's inequality to the moment generating function, the following is true for any $u,\ t > 0$,
    \begin{equation}\label{eq:markov-poisson-process}
        \Pr_{\mathrm{Branching}(d(k-1))}\left[Z_x \geq u\left(\frac{(d(k-1))}{\ln 2}\right)^x\right] \leq e^{-u\left(\frac{d(k-1)}{\ln 2}\right)^x}\phi_{x}(t) \leq e^{-u}e\, ,
    \end{equation}
    where the last inequality follows by a choice of $t = \left(\frac{\ln 2}{d(k-1)}\right)^x $. \\

    We now bound the probability that the total number of nodes in the branching process at height $x$ is \emph{at least} $c$ via a union bound,
    \[
        \Pr_{\mathrm{Branching}(d(k-1))}\left[\sum_{i=1}^{x}Z_i \geq c\right] \leq \sum_{i=1}^x\Pr[Z_i \geq \frac{c}{x}]\, .
    \]
    Let $c = (d(k-1))^{sx}\left(\frac{d(k-1)}{\ln 2}\right)^x$ and $u = \frac{d^{sx}}{x}$, where $s$ will be chosen later to demonstrate the existence of $A$. Substituting these into \autoref{eq:markov-poisson-process} yields a lower bound on the size of the neighborhoods induced by the $\poisson(d(k-1))$ branching process as,
    \begin{equation}\label{eq:p-nghbd-size}
        \Pr_{\mathrm{Branching}(d(k-1))}\left[\sum_{i=1}^{x}Z_i \geq (d(k-1))^{sx}\left(\frac{d(k-1)}{\ln 2}\right)^x\right] \leq xe^{-\frac{d^{sx}}{x}}e \leq e^{-d^{sx/2}}\, ,
    \end{equation}
    where we assume a sufficiently large choice of $x$. Denote by $p_r(v_i)$ the probability that a vertex $v_i$ in a random $k$-uniform hypergraph with $p = \frac{d}{kn^{k-1}}$ has a $x$-neighborhood with size \emph{at least} that of \autoref{eq:p-nghbd-size}, conditioned on the hypergraph having $r$ hyperedges. Extending the argument in \cite[Neighborhood Size Theorem]{farhi2020quantum} further,
    \[
        \sum_{r = dn/k}^{\infty}\Pr[\text{G has r hyperedges}]\cdot p_r(v_i) \leq \sum_{r = 0}^{\infty}\Pr[\text{G has r hyperedges}]\cdot p_r(v_i) \leq e^{-d^{sx/2}}\, .
    \]
    Since an increase in the number of sampled edges will only increase $p_r(v_i)$, it follows that,
    \[
        \sum_{r = dn/k}^{\infty}\Pr[\text{G has r hyperedges}]\cdot p_{dn/k}(v_i) \leq e^{-d^{sx/2}}\, .
    \]
    Now, the number of edges are distributed as $\poisson(dn/k)$. Since $\E[\poisson(dn/k)] = \frac{dn}{k}$,
    \[
        \sum_{r = dn/k}^{\infty}\Pr[\text{G has r hyperedges}] = \Pr[\text{G has }\geq dn/k \text{ hyperedges}] \, .
    \]
    Note that the mean of the $\poisson(dn/k)$ distribution is an integer. Consequently, the median of the distribution is also $\frac{dn}{k}$ \cite{10.2307/2160389}. Then,
    \begin{align*}
        \Pr[\poisson(dn/k) \geq dn/k] \in \left(\frac{1}{2} - \Pr[\poisson(dn/k) = dn/k], \frac{1}{2} + \Pr[\poisson(dn/k) = dn/k]\right)\, .
    \end{align*}
    Now, by applying Stirling's approximation,
    \[
        \Pr[\poisson(dn/k) = dn/k] \leq \frac{1}{\sqrt{2\pi dn/k}}\, .
    \]
    Consequently,
    \[
        \Pr[\poisson(dn/k) \geq dn/k] \geq \frac{1}{2} - \Pr[\poisson(dn/k) = dn/k] \geq \frac{1}{2} - \frac{1}{\sqrt{2\pi dn/k}} = \frac{1}{2} - o_{d, n}(1)\, .
    \]
    This yields an upper bound for the probability of a $x$-neighborhood exceeding the desired size as, 
    \[
        p_{dn/k}(v_i) \leq \left(2 + o_{d, n}(1)\right) e^{-d^{sx/2}} \leq e^{-d^{sx/3}}\, ,
    \]
    for sufficiently large $x$. \\
    We set
    \[
        2p \leq x - 1= \frac{(1-\tau)\log n}{\log(\frac{d(k-1)}{\ln 2})} - 1
    \]
    for some $\tau \in (0, 1)$. Then, let $\log = \log_{d(k-1)}$ and define $L = \log_{d(k-1)}\left(\frac{1}{\ln 2}\right)$. Consequently, the remaining argument follows exactly as in the last part of the proof of \cite[Neighborhood Size Theorem]{farhi2020quantum}. Specifically, \autoref{eq:p-nghbd-size} reduces to \cite[Eq. 75, Neighborhood Size Theorem]{farhi2020quantum} with $w = (1 - \tau)$ given the parameter choices above after some algebra.
\end{proof}

\end{document}